\documentclass{article}

\usepackage{amsmath, xypic, graphicx}
\usepackage{stmaryrd,xspace,amssymb}
\usepackage{amsthm}
\usepackage[active]{srcltx}
\usepackage{url}

\newcommand{\truob}{\underline{\mathbb{T}}^{\vert\psi\rangle}}
\newcommand{\stat}{\vert\psi\rangle}
\newcommand{\dirac}{\vert\psi\rangle\langle\psi\vert}
\newcommand{\pseu}{\underline{\mathfrak{w}}^{\vert\psi\rangle}}
\newcommand{\daspo}{\delta^{o}(p)_{C}}
\newcommand{\daspi}{\delta^{i}(p)_{C}}
\newcommand{\dasao}{\delta^{o}(a)_{C}}
\newcommand{\dasai}{\delta^{i}(a)_{C}}
\newcommand{\uSs}{\underline{\Sigma}_{\ast}}
\newcommand{\Ss}{\Sigma_{\ast}}

\newcommand{\IR}{\mathbb{IR}}
\newcommand{\Q}{\mathbb{Q}}
\newcommand{\uC}{\underline{\mathbb{C}}}
\newcommand{\uIR}{\underline{\mathbb{IR}}}
\newcommand{\unR}{\underline{\mathbb{R}}}
\newcommand{\drie}{\vartriangleleft}

\newcommand{\se}{\section} 
\newcommand{\su}{\subsection}

\newcommand{\beq}{\begin{equation}}
\newcommand{\eeq}{\end{equation}}
\newcommand{\bea}{\begin{eqnarray}}
\newcommand{\eea}{\end{eqnarray}}

\newcommand{\C}{{\mathbb C}} 
 
\newcommand{\N}{{\mathbb N}} \newcommand{\R}{{\mathbb R}}
 
\newcommand{\Cat}[1]{\ensuremath{\mathrm{\textbf{#1}}}}

\newcommand{\id}[1]{\ensuremath{\mathrm{id}}}
\newcommand{\op}{\ensuremath{^{\mathrm{op}}}}

\newcommand{\Set}{\Cat{Sets}\xspace}
\newcommand{\Sh}{\ensuremath{\mathrm{Sh}}}

\newcommand{\uS}{\underline{\Sigma}}
\newcommand{\uA}{\underline{A}}

\newtheorem{dork}{Definition}[section]
\newtheorem{tut}[dork]{Theorem}
\newtheorem{poe}[dork]{Proposition}
\newtheorem{lem}[dork]{Lemma}
\newtheorem{cor}[dork]{Corollary}

\begin{document}

\title{A Comparison of Two Topos-Theoretic Approaches to Quantum Theory}
\author{Sander A.M. Wolters\\
\small \textit{Institute for Mathematics, Astrophysics, and Particle Physics}\\
\small \textit{Radboud Universiteit Nijmegen,} \\
\small \textit{Heyendaalseweg 135, 6525 AJ Nijmegen, The Netherlands.}\\
\small \textit{e-mail: s.wolters@math.ru.nl}.}
\date{August 2011}
\maketitle

\begin{abstract}
The aim of this paper is to compare the two topos-theoretic approaches to quantum mechanics that may be found in the literature to date. The first approach, which we will call the contravariant approach, was originally proposed by Isham and Butterfield, and was later extended by D\"oring and Isham. The second approach, which we will call the covariant approach, was developed by Heunen, Landsman and Spitters.  

Motivated by coarse-graining and the Kochen-Specker theorem, the contravariant approach uses the topos of presheaves on a specific context category, defined as the poset of commutative von Neumann subalgebras of some given von Neumann algebra. In particular, the approach uses the spectral presheaf. The intuitionistic logic of this approach is given by the (complete) Heyting algebra of closed open subobjects of the spectral presheaf. We show that this Heyting algebra is, in a natural way, a locale in the ambient topos, and compare this locale with the internal Gelfand spectrum of the covariant approach.

In the covariant approach, a non-commutative C*-algebra (in the topos $\mathbf{Set}$) defines a commutative C*-algebra internal to the topos of covariant functors from the context category to the category of sets. We give an explicit description of the internal Gelfand spectrum of this commutative C*-algebra, from which it follows that the external spectrum is spatial.

Using the daseinisation of self-adjoint operators from the contravariant approach, we give a new definition of the daseinisation arrow in the covariant approach and compare it with the original version. States and state-proposition pairing in both approaches are compared. We also investigate the physical interpretation of the covariant approach.

\end{abstract}

\se{Introduction}

The goal of this paper is to compare two different but related recent applications of topos theory to quantum theory. Both approaches provide an intuitionistic logic for quantum mechanics, which forms an alternative to the orthodox quantum `logic' of Birkhoff and von Neumann \cite{bine}. In these alternatives, the propositions about the system under investigation form a Heyting algebra, in contrast to the orthomodular lattice of projections in orthodox quantum logic. Much has already been said about the relationship between the topos-theoretic approaches and orthodox quantum logic \cite{di2, di3, di, hls2}. In this paper, however, we are interested in the relationship between the logics of these two topos-based approaches.

Some familiarity with basic topos theory is required. A basic introduction can be found in Goldblatt \cite{gol}. A more extensive introduction is given by Mac Lane and Moerdijk \cite{mm}, which covers more than enough to understand the topos theory used in both topos approaches to quantum theory. Of course, everything and more can also be found in Johnstone \cite{jh1}. Another useful reference is Borceux \cite{bor}.

\su{Contravariant or Coarse-Graining Approach}

As far as the author knows, the oldest application of topos theory to quantum mechanics is due to Adelman and Corbett \cite{adco}, but apparently it has not influenced later authors, and indeed it will play no role in this paper either. Of the two topos-theoretic approaches to quantum theory that are to be compared in this paper, the older approach originates with Isham \cite{ish} and Butterfield and Isham \cite{buhais, butish1, butish2, butish3}. We will call this the \textbf{contravariant approach}. The \textbf{coarse-graining approach} would also have been a suitable name, for coarse-graining is one of the guiding principles of the contravariant approach, as we shall see shortly. 

It is important to remark that the formalism of Butterfield and Isham, initially intended as a reformulation of quantum mechanics, was extended by D\"oring and Isham to a topos approach to theories of physics in general \cite{di1, di2, di3, di4}. In this more general perspective, a language is associated to a physical system, and a physical theory for that system is a suitable representation of this language in a topos. This is more general than using topoi of contravariant functors, as proposed for theories of quantum mechanics. In this wider setting, the name contravariant approach seems misplaced. However, this more general picture will play only a minor role in this paper, so it seems harmless to use the name ``contravariant approach". A historical overview of the more general framework has been given by Isham \cite{ish2}, and the comprehensive review \cite{di} covers most of the ideas of the contravariant approach as of 2009. 

We will now sketch some of the important ideas in the contravariant approach. In this approach a quantum system is described by a von Neumann algebra $A$. We can typically think of $A$ as the von Neumann algebra of bounded operators on some Hilbert space $\mathcal{B}(\mathcal{H})$, but the approach works for any von Neumann algebra.\footnote{In what follows we only need the fact  that the projections of the operator algebra form a complete lattice, so we may generalize von Neumann algebras to AW*-algebras \cite{ber}.} \textbf{Contextuality}, motivated by the Kochen-Specker theorem \cite{kosp}, is an important ingredient of the contravariant approach \cite{buhais, butish1, butish2, butish3}. In the contravariant approach a classical \textbf{context} is represented by an Abelian von Neumann subalgebra\footnote{Only von Neumann algebras $C$ where the unit of $C$ is the unit of $A$ are included. Usually the trivial context $\mathbb{C}1$ is excluded as a context.} of $A$. Such classical contexts form a poset $\mathcal{V}(A)$, where the partial order is given by inclusion. Next, one considers the category $[\mathcal{V}(A)\op,\mathbf{Set}]$, of contravariant functors from $\mathcal{V}(A)$ to $\mathbf{Set}$. Working with this functor category allows one to work with all classical contexts at the same time, whilst keeping track of relations between the different contexts. 

The category $[\mathcal{V}(A)\op,\mathbf{Set}]$ is an example of a topos.\footnote{By topos we will always mean a Grothendieck topos. Every elementary topos encountered in this paper is a Grothendieck topos. In particular, it has a natural numbers object. } A topos is a highly structured category that has many different faces \cite[Preface]{jh1}.

Another important concept in the contravariant approach is \textbf{coarse-graining} (\cite[Section 5]{di}, or any other paper on the contravariant approach). Let $C',C\in\mathcal{V}(A)$ be contexts such that $C'\subset C$. Considering a self-adjoint operator $a\in C_{sa}$ and an open subset $\Delta\subseteq\R$, the proposition $a\in\Delta$ is represented by a (spectral) projection operator $p=[a\in\Delta]$. Because $C$ is a von Neumann algebra, it follows that $p\in C$. For the `coarser' context $C'\subset C$, it may very well be that $p\notin C'$. In this context the projection $p$ is replaced by an approximation using the available projection operators of $C'$, namely $\delta^{o}(p)_{C'}$, the smallest projection operator $q$ in $C'$, such that $p\leq q$. Note that we associate a weaker proposition (i.e. larger projection operator) to the coarser context $C'$, compared to the context $C$. Although we may not be able to assign the truth value `true' to the propostion $p=[a\in\Delta]$, it may be the case that we may assign `true' to the weaker proposition $\delta^{o}(p)_{C'}$. If $C''\subset C'\subset C$ is an even coarser context and $\delta^{o}(p)_{C'}$ is assigned `true', then $\delta^{o}(p)_{C''}\geq\delta^{o}(p)_{C'}$ is also `true'.\footnote{This is not ambiguous, as it does not matter if we take $p\in C$ and approximate it in $C''$, or if we take $\delta^{o}(p)_{C'}\in C'$ and approximate it in $C''$. The outcome is the same.} This means that the collection of $C'\in\mathcal{V}(A)$, $C'\subseteq C$ such that $\delta^{o}(p)_{C'}$ is true is a sieve on $C$.\footnote{This is a collection $t_{C}\subseteq\mathcal{V}(A)$ such that if $C'\in t_{C}$, then $C'\subseteq C$ and if $C''\subseteq C'\in t_{C}$, then $C''\in t_{C}$. Note that a sieve in $\mathcal{V}(A)$ is the same as an ideal of this poset.} This can be seen as another motivation for using the topos $[\mathcal{V}(A)\op,\mathbf{Set}]$. The subobject classifier of this topos, denoted $\underline{\Omega}$, is crucial for the notion of truth in this topos. It is defined as follows: for a context $C\in\mathcal{V}(A)$, the set $\underline{\Omega}(C)$ is the set of sieves on $C$. 

More generally, let $p$ be any projection operator in $A$, and $C$ any context. Then $\delta^{o}(p)_{C}$ is defined to be the smallest projection operator $q$ in $C$ with the property $p\leq q$. If $p\in C$, then clearly $\delta^{o}(p)_{C}=p$. Following the literature on the contravariant approach, we call $\delta^{o}(p)_{C}$ the \textbf{outer daseinisation} of $p$ in $C$. Similarly, the \textbf{inner daseinisation} of $p$ in context $C$, denoted by $\delta^{i}(a)_{C}$, approximates $p$ in $C$ by taking the largest projection operator $q$ in $C$ such that $q\leq p$.

Next, consider the so-called \textbf{spectral presheaf} $\uS:\mathcal{V}(A)\op\to\mathbf{Set}$. At a context $C\in\mathcal{V}(A)$, the set $\uS(C)$ is defined as the Gelfand spectrum $\Sigma_{C}$ of $C$, seen as a commutative C*-algebra. If $C'\subseteq C$, this gives a restriction map $\uS(C)\to\uS(C')$, $\lambda\mapsto\lambda|_{C'}$. For a projection $p$, outer daseinisation gives, for every context $C$, a closed open subset of the spectrum $\Sigma_{C}$, namely the support of the Gelfand transform of the projection operator $\delta^{o}(p)_{C}$. These closed open subsets combine to give a subobject of the spectral presheaf $\underline{\delta}^{o}(p)\rightarrowtail\uS$. Such subobjects are special cases of closed open subobjects of the spectral presheaf: a subobject $\underline{U}\rightarrowtail\uS$ is called a \textbf{closed open subobject} if for every $C\in\mathcal{V}(A)$ the set $\underline{U}(C)$ is closed and open in $\Sigma_{C}$. The set of closed open subobjects of the spectral presheaf is denoted by $\mathcal{O}_{cl}\uS$.

In the logic of the contravariant approach, the spectral presheaf $\uS$ plays the role of a state space. In accordance with coarse-graining, the closed open subobjects of $\uS$ represent propositions about the system. As shown in \cite{di2} and \cite[Appendix 1]{di}, the set $\mathcal{O}_{cl}\uS$ may be given the structure of a Heyting algebra. 

Let $\underline{1}$ be the terminal object of the topos $[\mathcal{V}(A)\op,\mathbf{Set}]$, given by the constant functor $\underline{1}(C)=\{\ast\}$. The arrows $\underline{1}\to\uS$ would be natural candidates for pure states in the contravariant quantum logic. However, in many cases, including $A=\mathcal{B}(\mathcal{H})$ with $\text{dim}(\mathcal{H})>2$, the Kochen-Specker theorem prohibits the existence of such arrows \cite{buhais}, \cite{di2}. Instead of taking points, one therefore considers pseudo-states $\pseu$ (see Subsection 4.1), which are subobjects of $\uS$.\footnote{Alternatively states can be described as measures (\cite{doe3}, \cite{doe2}), which we also discuss in Subsection 4.1.}

We can combine a proposition $\underline{S}\in\mathcal{O}_{cl}\uS$ with a pseudostate $\pseu$ such that it gives a truth value in $[\mathcal{V}(A)\op,\mathbf{Set}]$: this is defined as an arrow $\underline{1}\to\underline{\Omega}$. For every context $C$ we thus obtain a sieve on $C$ in accordance with the idea of coarse-graining. At context $C$, the truth value is given by
\begin{equation}
\nu(\pseu\subseteq\underline{S})_{C}=\{C'\in\mathcal{V}(A)\mid C'\subseteq C,\ \pseu(C')\subseteq\underline{S}(C')\}\in\underline{\Omega}(C).
\end{equation}

\su{Covariant or Bohrification Approach}

The other active topos-theoretic approach to quantum mechanics is the \textbf{covariant approach}. Another suitable name would be \textbf{Bohrification}. The covariant approach was initiated in Heunen, Landsman and Spitters \cite{hls}, and further developed in \cite{hls2}. A more detailed description can be found in \cite{hls3}, and an explicit discussion for finite dimensional systems is given in \cite{chls}. We now give a brief sketch of the covariant approach. The first steps appear to look like the contravariant approach, but soon the covariant approach takes a different direction. 

The covariant approach is inspired by algebraic quantum theory \cite{haag}, insofar as the system under investigation is described by a C*-algebra $A$, which we assume to be unital. A second ingredient is Bohr's doctrine of classical concepts \cite{bohr}, or rather a particular mathematical interpretation of this principle. This principle states that we can only look at a quantum system from the point of view of some classical context. The classical contexts are represented by unital\footnote{The unit is included for technical reasons.} commutative C*-subalgebras of $A$.\footnote{We demand that the unit of the context $C$ is equal to the unit of $A$.}  These classical contexts, partially ordered by inclusion, form a poset $\mathcal{C}(A)$. 

We consider the following interpretation of a context in the covariant setting. The selfadjoint elements of a context represent physical quantities. As $C$ is commutative, these physical quantities are compatible, and therefore they fit in a single measurement context. In looking for a physical interpretation for the covariant approach, we think of a context $C$ as a measurement context, (where we measure one or more of the compatible observable quantities that correspond to an element of $C_{sa}$). The reader who does not want to use operational notions (although using operational notions does not imply having an instrumentalist interpretation of the theory) may think of a context more abstractly as a classical snapshot of the system, or a stage of knowledge about the system\footnote{Before we continue, we first issue a warning. The goal of this paper is to compare the covariant and the contravariant approaches. For this comparison it helps to have a physical interpretation for various constructions in the covariant approach. Please keep in mind that the physical interpretation given to contexts (and later to the internal language and elementary propositions) need not exactly be the interpretation that the originators of this approach have in mind. It is simply the most natural interpretation, according to the present author}. 

If we take Bohr's doctrine of classical contexts seriously, then instead of talking about an obervable represented by $a\in A_{sa}$, we should always use some suitable classical context and consider pairs $(C,a)$ with $C$ a context and $a\in C_{sa}$. More generally, we would like to talk about observables in an arbitrary context. For example if $p,q\in A_{sa}$ correspond to physical quantities that are complementary, we would like to talk about $q$ in any context that includes $p$ (using only expressions that make sense in this context). In Section 3 we will see how the daseinisation maps of the contravariant approach help in achieving this, but for now we stick with pairs $(C,a)$ with $a\in C$. Treating physical quantities as such pairs, it is natural to consider the following topos. Let $\mathcal{C}_{d}(A)$ be the set of all unital commutative C*-subalgebras of $A$, seen as a discrete category (so we forget about the order relation of $\mathcal{C}(A)$). Next, consider the topos $\mathcal{T}_{d}=[\mathcal{C}_{d}(A),\Set]$. An object of $\mathcal{T}_{d}$ is equivalent to a bundle over $\mathcal{C}_{d}(A)$, or a $\mathcal{C}_{d}(A)$ indexed family of sets. The observables in context, given by pairs $(C,a)$, with $a\in C_{sa}$ provide such a collection of sets and define an object $\underline{A}$ of $\mathcal{T}_{d}$.

Every topos can be seen as a (generalized) universe of sets and has an internal language, the Mitchell-B\'enabou language of the topos, and a semantics for this language, the Kripke-Joyal semantics (\cite{mm}, Chapter VI). For the rather simple topos $\mathcal{T}_{d}$ consisting of functors from the discrete category $\mathcal{C}_{d}(A)$, the internal logic is a copy of the logic of $\Set$ for every context $C$. Mathematics internal to this topos is the same as classical mathematics (the mathematics of $\Set$) while keeping track of a context.

Internal to the topos $\mathcal{T}_{d}$, the generalized set $\underline{A}$ is a \underline{commutative} C*-algebra.\footnote{The claim can be shown using the proof of Theorem 5 of \cite{hls} which simplifies as we work with a simpler topos.} The definition of a C*-algebra in a topos can be found in \cite{banmul, banmul2, banmul3}. In the work of Banaschewski and Mulvey \cite{banmul, banmul2, banmul3} a version of Gelfand duality is presented that holds in any topos, expressing a duality between the category of unital commutative C*-algebras and the category of compact completely regular locales.\footnote{A \textbf{locale} can be thought of as a pointfree description of a topological space. In this picture a compact completely regular locale corresponds to a compact Hausdorff space, which is automatically a compact completely regular space. For an introduction to locales see \cite[Chapter 2]{jh4}, \cite[Chapter IX]{mm}, and \cite{vic}.} A more explicit and fully constructive description of Gelfand duality is given in \cite{coq, coqsp}. 

By this topos generalization of Gelfand duality, the internal observable algebra $\underline{A}$ is isomorphic to the internal set $C(\uS,\mathbb{C})$ of continuous complex valued maps on a certain internal compact regular locale. The frame of this locale is given by the $\mathcal{C}_{d}(A)$ indexed family of sets $(\mathcal{O}\Sigma_{C})_{C\in\mathcal{C}_{d}(A)}$,  where $\mathcal{O}\Sigma_{C}$ denotes the topology of the Gel'fand spectrum of $C$.\footnote{The internal spectrum can be calculated as in Appendix A of \cite{hls}. The proof given there simplifies as our topos has a simple semantics.} Furthermore, it can be shown\footnote{Again, in the same way as in \cite{hls}.} that a state of the system, in the sense of a normalized, positive linear functional on the C*-algebra $A$, defines an internal probability valuation on the spectrum $\uS$.

To summarize, guided by algebraic quantum theory and Bohr's doctrine of classical concepts, we have arrived at the topos $\mathcal{T}_{d}$ and the use of its internal language. Internally (that is, keeping track of a classical context), the description of the system looks more like classical physics. Physical quantities are represented by continuous functions on a space and states are represented by probability valuations on this space. The topos $\mathcal{T}_{d}$ is interesting neither mathematically, as it is too simple, nor physically, as it does not allow for relations between different contexts. In order to remedy this, we replace $\mathcal{C}_{d}(A)$ by the poset $\mathcal{C}(A)$.
 
Starting from the poset $\mathcal{C}(A)$, the two simplest topoi to consider are the topos of covariant functors $[\mathcal{C}(A),\Set]$ and the topos of presheaves $[\mathcal{C}(A)^{\op},\Set]$. Using $[\mathcal{C}(A)^{\op},\Set]$ has the advantage that it connects better with the contravariant approach, but we are immediately faced with a problem. How do we see the (contextual) observables as an object $\underline{A}$, of this presheaf topos? In particular, for every inclusion $C\subset C'$ in $\mathcal{C}(A)$, we need a function $C'\to C$. If we are willing to restrict attention from C*-algebras to von Neumann algebras, then the daseinisation of self-adjoint operators from the contravariant approach, discussed in Subsection 3.1, could be used to define $\underline{A}_{sa}$ in $[\mathcal{C}(A)^{\op},\Set]$ by taking for every $C\in\mathcal{V}(A)$, $\underline{A}_{sa}(C)=C_{sa}$ and for every $C\subseteq C'$, the restriction map $C'_{sa}\to C_{sa}$ is given by $a\mapsto\delta^{o}(a)_{C}$ (we could also have used inner daseinisation). Although this gives a well defined object of $[\mathcal{C}(A)^{\op},\Set]$, this object is not in any obvious way an internal commutative C*-algebra. Consider for example the addition maps $+_{C}:C_{sa}\times C_{sa}\to C_{sa}$, $(a,b)\mapsto a+b$. With restrictions given by outer or inner daseinisation, these maps do not combine to a natural transformation. In the discrete case, internal to the topos, the observables (and states) looked more classical, at least at a mathematical level. It is not clear how this can be attained for the topos of presheaves.

The other option is using the topos of covariant functors $[\mathcal{C}(A),\Set]$. Define the functor $\underline{A}:\mathcal{C}(A)\to\mathbf{Set}$ by $\underline{A}(C)=C$, and for $C'\subseteq C$ take $\underline{A}(C')\to\underline{A}(C)$ to be the inclusion $C'\hookrightarrow C$. This `tautological functor' is called the \textbf{Bohrification} of $A$. Internal to  $[\mathcal{C}(A),\Set]$, the Bohrification $\underline{A}$ is a unital commutative C*-algebra (Theorem 5 of \cite{hls}). By constructive Gelfand duality, the Bohrification $\underline{A}$ is isomorphic to the internal C*-algebra of complex valued function on a compact completely regular locale $\underline{\Sigma}_{\underline{A}}$ the spectrum of $\underline{A}$. States on $A$ translate to probability valuations on the spectrum (Section 4 of \cite{hls}). In choosing the observables to be covariant functors, we keep the internal more classical description.

In the discrete case $\mathcal{T}_{d}=[\mathcal{C}_{d}(A),\Set]$, working internally just means that we used contexts, which makes sense physically. We subsequently switched to the topos $\mathcal{T}=[\mathcal{C}(A),\Set]$ which, like the discrete case, gives an internal description in which the formalism of quantum theory looks like the formalism of classical physics, and, in contrast to the discrete case, allows for relations between contexts. Does the internal language of this new topos $\mathcal{T}$ make sense physically? Suppose that $\phi$ represents a formula in the internal language of $\mathcal{T}$. To make things more explicit, let $\phi$ express `relative to a certain state, an observable quantity represented by $a$ only takes values in $\Delta\subset\mathbb{R}$'. At the end of Section 4 we will give the precise description of $\phi$, but for now it remains a black box. If $\phi$ holds at context $C$ in the internal language, i.e. $C\Vdash\phi$, then we interpret this as follows. By only making use of the measurements corresponding to $C$ we can verify that the claim made by $\phi$ holds. This is still vague, especially the `\textit{the claim made by $\phi$ holds'} part, and we will return to this issue in Subsection 3.5 and Section 4.

The Kripke-Joyal semantics for the functor topos $[\mathcal{C}(A),\mathbf{Set}]$ is the same as that of a Kripke model for intuitionistic logic. The interpretation given above is just a physical version of this Kripke model.  Note that the `information order' of this Kripke model agrees with physical intuition in the following sense. If $C'\subset C$ in $\mathcal{C}(A)$, then $C'$ is lower in the `information order' of the Kripke model than $C$, and from the physics point of view one can describe fewer physical observations from $C'$ (compared to $C$). Also note that for the presheaf topos $[\mathcal{C}(A)^{\op},\Set]$, the Kripke-Joyal semantics can also be seen as a Kripke model, but the information order for this model is opposite to the physical intuition, making the Kripke model perspective unattractive in this case.

In the covariant approach, propositions about the system are represented by open subsets of the spectrum, or equivalently by the points of its associated frame. As a locale is a complete Heyting algebra, the spectrum therefore automatically has a Heyting algebra structure internally, but the set of opens of the spectrum also give a complete Heyting algebra in $\Set$. Thus, like its contravariant counterpart, the logic of the covariant quantum approach is in general intuitionistic (and hence distributive, unlike conventional quantum logic based on orthomodular lattices). The states of the covariant approach are (internal) probability valuations on $\underline{\Sigma}_{\underline{A}}$, which are equivalent to quasi-states on $\underline{A}$, \cite{hls}. States combine in a natural way with propositions, yielding truth values (formulae such as $\phi$ above) as points of $\underline{\Omega}$. Here, $\underline{\Omega}$ is the subobject classifier of $[\mathcal{C}(A),\mathbf{Set}]$. If $C\in\mathcal{C}(A)$, then an element of $\underline{\Omega}(C)$ is a cosieve on $C$.\footnote{This is a collection $t_{C}\subseteq\mathcal{C}(A)$ such that if $C'\in t_{C}$, then $C\subseteq C'$, and if $C''\supseteq C'\in t_{C}$, then $C''\in t_{C}$.} A truth value is equivalent to a cosieve $t_{C}$ on $C$, for every context $C$, such that, if $C\subseteq C'$, then $t_{C}\cap(\uparrow C')=t_{C'}$. Here $\uparrow C'$ stands for the set of all contexts $C''\in\mathcal{C}(A)$ such that $C'\subseteq C''$. Using cosieves fits well with the interpretation given above. If we can verify a claim using the measurments corresponding to $C$, and if $C'\supset C$ represents a context using more refined measurements, then clearly we can verify that same claim using the measurements corresponding to $C'$.

\su{Differences Between the Two Approaches}

Clearly, there are differences between the two approaches. To name a few: 

\begin{itemize}
\item The contravariant approach uses von Neumann algebras, whereas the covariant approach uses C*-algebras. This difference has to do with daseinisation, which plays an important role in the contravariant approach but is far less significant in the covariant approach. Daseinisation makes heavy use of the additional structure that von Neumann algebras have to offer, notably the abundance of projections.
\item The covariant approach makes extensive use of the internal (Mitchell-B\'enabou) language and the corresponding Kripke-Joyal semantics of the topos $[\mathcal{C}(A),\mathbf{Set}]$. In the contravariant approach the language and corresponding semantics of $[\mathcal{V}(A)^{\op},\Set]$ do not play a role (thus far). This does not mean that internal constructions are unimportant for the contravariant approach. As an example, consider the assignment of truth values to propositions and states (see e.g. \cite{di} Section 6) This assignment is natural when the topos is seen as a generalized universe of sets. Another example is the value object $\underline{\mathbb{R}}^{\leftrightarrow}$, which is shown (see e.g. \cite{di} Subsection 8.6) to be  an internal commutative monoid. 
\item The contravariant approach uses coarse-graining, which does not appear in covariant quantum logic. This point is connected to the previous one regarding the use of internal language. The covariant approach relies on the internal language of the topos and uses the corresponding Kripke-Joyal semantics. The semantics of the contravariant approach, on the other hand is guided by the idea of coarse-graining.
\item The state spaces are constructed in a very different way. In the contravariant approach the state object is the spectral presheaf, which is obtained by assembling all the Gelfand spectra of the commutative subalgebras. In the covariant approach the state space is the external description of the locale obtained by taking the constructive Gelfand spectrum of the internal commutative C*-algebra obtained from all the commutative subalgebras. Are these objects, which live in different topoi, related in any way?
\item States are defined in a completely different way. However, in \cite{doe2} and \cite{doe3} D\"oring describes contravariant states as measures on the closed open subobjects of the spectral presheaf. This description looks like it is closely related to the covariant notion of state. 
\end{itemize}

We will study these differences and some others in the next sections. 

Section 2 discusses the two different state spaces. To summarize, let $\Sigma\equiv\Sigma(A)$ be the disjoint union of all the Gelfand spectra $\Sigma_{C}$, where $C\subseteq A$ is a context.\footnote{We ignore the difference in context categories between the approaches (i.e. between C*-algebras and von Neumann algebras) for the moment.} The set $\Sigma$ may be equipped with two different topologies. The first topology $\mathcal{O}\Sigma^{\ast}$ is connected to the contravariant approach. We show that there is an injective morphism (of complete Heyting algebras) from $\mathcal{O}_{cl}\uS$ into $\mathcal{O}\Sigma^{\ast}$. The second topology $\mathcal{O}\Sigma_{\ast}$ is connected to the covariant approach. We show that its associated locale is the external description of the constructive Gelfand spectrum $\underline{\Sigma}_{\underline{A}}$. This result is of interest independently of the comparison between the two approaches.

Section 3 investigates daseinisation and elementary propositions, i.e. propositions of the form $a\in\Delta$. In the covariant approach there is a daseinisation arrow and there are elementary propositions as well, but in the development so far these have not played a fundamental role. Nonetheless, by restricting from C*-algebras to von Neumann algebras we can use the daseinisation techniques of the contravariant approach in the covariant approach. This leads to an explicit description of the daseinisation arrow as well as of elementary propositions in the covariant approach, and at the end of the day, the two notions turn out to be closely related.

Section 4 deals with states and the assignment of truth values in both approaches. Using the covariant daseinisation developed in Section 3, we introduce a counterpart of the contravariant pseudo-states into the covariant approach and compare these with the original notion of states used in the covariant approach. Subsequently we compare the covariant states with the definition of contravariant states as measures by D\"oring.

\se{State Spaces}

As we have seen, the quantum state spaces in the two approaches to topos quantum logic are constructed in different ways. In the contravariant approach the state space, or rather state object, is the spectral presheaf $\underline{\Sigma}$. Recall that this is the presheaf that assigns to every context $C$ (which is an Abelian von Neumann subalgebra of the von Neumann algebra $A$ associated to the system under investigation) its Gelfand spectrum. Let $\mathcal{V}(A)$ be the poset of contexts, with partial order given by inclusion, viewed as a category. In the contravariant approach we use the topos $[\mathcal{V}(A)\op,\mathbf{Set}]$ of contravariant functors from the context category $\mathcal{V}(A)$ to the category $\mathbf{Set}$. 

In the covariant approach the observable algebra is a unital C*-algebra $A$. A classical context $C$ is a unital commutative C*-subalgebra of $A$. The context category $\mathcal{C}(A)$ is the poset of classical contexts partially ordered by inclusion, viewed as a category. The algebra $A$ defines a functor $\underline{A}:\mathcal{C}(A)\to\mathbf{Set}$, which is a commutative C*-algebra in the internal language of the topos $[\mathcal{C}(A),\mathbf{Set}]$ of functors $\mathcal{C}(A)\to\mathbf{Set}$. The corresponding quantum state space $\underline{\Sigma}_{\underline{A}}$ is a compact regular locale, internal to $[\mathcal{C}(A),\mathbf{Set}]$, which is obtained by applying a constructive version of Gelfand duality to $\underline{A}$. Trivially, instead of looking at the topos of covariant functors $\mathcal{C}(A)\to\mathbf{Set}$, we can equivalently look at the topos of presheaves (that is, contravariant functors) $\mathcal{C}(A)\op\to\mathbf{Set}$.

In this section we will see that even though the state objects of the two different approaches are constructed in different ways and live in different topoi, there are strong connections between the two. Before we can get started, we need to deal with the difference in context categories. The contravariant approach uses abelian von Neumann subalgebras in defining the context category $\mathcal{V}(A)$, whereas the covariant approach uses unital commutative C*-subalgebras in defining the context category $\mathcal{C}(A)$. This difference will be important in Section 3 when we discuss daseinisation. However, in the current section it plays no role at all. We can use either the category $\mathcal{C}(A)$ or the category $\mathcal{V}(A)$ in both the covariant and the contravariant approaches. Whenever we compare the state spaces of the approaches, we can safely ignore the differences that arise from the differences in context categories.

In Subsection 2.1 we focus on the contravariant approach. We will define a topological space $\Sigma^{\ast}$ and a continuous map $\pi:\Sigma^{\ast}\to\mathcal{V}(A)$. The associated frame $\mathcal{O}\Sigma^{\ast}$ is closely connected to the contravariant approach, as follows: Theorem 2.2 shows that there is an injective morphism of complete Heyting algebras $\mathcal{O}_{cl}\underline{\Sigma}\to\mathcal{O}\Sigma^{\ast}$, where $\mathcal{O}_{cl}\uS$ is the complete Heyting algebra of closed open subobjects of the spectral presheaf. The propositions in the contravariant approach are elements of $\mathcal{O}_{cl}\uS$. The map $\pi:\Sigma^{\ast}\to\mathcal{V}(A)$ defines a locale $\uS^{\ast}$, internal to $[\mathcal{V}(A)\op,\mathbf{Set}]$. This locale is shown to be compact (Corollary 2.7), but in general it is not regular (Corollary 2.10). Proposition 2.3 demonstrates that $\mathcal{O}_{cl}\uS$ itself also defines a locale internal to $[\mathcal{V}(A)\op,\mathbf{Set}]$ in a natural way.

Subsection 2.2 deals with the covariant approach. In a similar vein, we define a topological space  $\Sigma_{\ast}$ and a continuous map $\pi:\Sigma_{\ast}\to\mathcal{C}(A)$, which turn out to be closely related to the space $\Sigma^{\ast}$ and map $\pi:\Sigma^{\ast}\to\mathcal{V}(A)$ of Subsection 2.1. Corollary 2.18 shows that the map $\pi:\Sigma_{\ast}\to\mathcal{C}(A)$ is just the external description of the spectrum of $\uA$ in $[\mathcal{C}(A),\mathbf{Set}]$. In Subsection 2.4 there is a brief discussion of the Gelfand transform of the Bohrification $\uA$.

\su{Contravariant Approach}

We start by investigating if the state object in the contravariant approach, i.e. the spectral presheaf $\underline{\Sigma}$, is, in a natural way, a locale (and consequently a Heyting algebra) internal to the topos of contravariant functors $[\mathcal{V}(A)^{\op},\mathbf{Set}]$. Subsequently we show how the frame $\mathcal{O}\uS^{\ast}$ of this locale relates to the propositions of the contravariant approach, the closed open subobjects of the spectral presheaf. It is also interesting to check if  $\underline{\Sigma}^{\ast}$ might be a compact completely regular locale. If so, we could recognize it internally as the spectrum of a commutative C*-algebra. By Corollary 2.10, this will turn out not to be the case.

\subsubsection{Spectral Presheaf as an Internal Locale}

Let $A$ be a von Neumann algebra and let $\mathcal{V}(A)$ be the poset category corresponding to the poset of all abelian von Neumann subalgebras\footnote{Such that the unit of the subalgebra is equal to the unit of $A$. Usually the trivial context $\mathbb{C}1$ is not considered a context in this approach.} of $A$, partially ordered by inclusion. The spectral presheaf $\underline{\Sigma}:\mathcal{V}(A)\op\to\mathbf{Set}$ is given by
\begin{equation}
\underline{\Sigma}(C)=\Sigma_{C},\ \ \ \rho_{CD}:=\underline{\Sigma}(i_{DC}):\Sigma_{C}\to\Sigma_{D},\ \lambda\mapsto\lambda |_{D},
\end{equation}
with $C,D\in\mathcal{V}(A)$ and $i_{DC}:D\to C$ is the inclusion of $D$ into $C$. Here $\Sigma_{C}$ denotes the Gelfand spectrum of $C\in\mathcal{V}(A)$.\footnote{The Gelfand spectrum $\Sigma_{C}$ is the set of characters on $C$, given the relative weak* topology.} Recall that a subobject $\underline{U}\to\underline{\Sigma}$ is called a closed open subobject if for every $C\in\mathcal{V}(A)$ the set $\underline{U}(C)\subseteq\Sigma_{C}$ is both open and closed in $\Sigma_{C}$. 

Equip the set $\mathcal{V}(A)$ with a topology by declaring all downwards closed sets to be open\footnote{This is the 'anti-Alexandrov' topology on $\mathcal{V}(A)$. We could have taken the Alexandrov topology, which consists of all upwards closed sets, and which is used in the covariant approach. However, this topology does not get us any closer to $\mathcal{O}_{cl}\underline{\Sigma}$. We will use the Alexandrov topology at a later stage (see below Theorem 2.13).} (these are all sets $U\subseteq\mathcal{V}(A)$ such that if $C\in U$ and $D\subseteq C$, then $D\in U$). This topology has the principal downsets $\downarrow C=\{D\in\mathcal{V}(A)\mid D\subseteq C\}$ as a basis. Using the correspondence
\begin{equation}
\overline{\Sigma}(\downarrow C) = \underline{\Sigma}(C),
\end{equation}
it is easy to check that a presheaf $\underline{\Sigma}$ on $\mathcal{V}(A)$ (seen as a poset category) is equivalent to a sheaf $\overline{\Sigma}$ on $\mathcal{V}(A)$ (seen as a space), equipped with the downset topology. Recall that a sheaf on a topological space $X$ is equivalent to an \'etale space over $X$ \cite[Chapter II]{mm}. An \'etale space is a topological space $Y $over a topological space X, i.e. a continuous map $f:Y\to X$ that is a local homeomorphism in the following sense: for any $y\in Y$ there is a neighborhood $V$ of $y$ in $Y$ such that $f(V)$ is open in $X$ and $f|_{V}:V\to f(V)$ is a homeomorphism. Given the spectral presheaf $\underline{\Sigma}$, seen as a sheaf $\overline{\Sigma}$, we can construct the corresponding \'etale space $\Sigma$. We thus obtain the local homeomorphism
\begin{equation}
\pi:\Sigma\to\mathcal{V}(A),\ \  (C,\lambda)\mapsto C,
\end{equation}
\begin{equation}
\Sigma=\{(C,\lambda)\mid C\in\mathcal{V}(A),\lambda\in\Sigma_{C}\}=\coprod_{C\in\mathcal{V}(A)}\Sigma_{C},
\end{equation}
 where $\Sigma$ has the topology generated by the basis
 \begin{equation}
 \mathcal{W}=\{W_{C,\lambda}\mid C\in\mathcal{V}(A),\lambda\in\Sigma_{C}\},\ \ W_{C,\lambda}=\{(D,\lambda|_{D})\mid D\subseteq C\}.
 \end{equation}
It is shown in \cite[Section C1.6]{jh1} that for a locale $X$ in $\mathbf{Set}$ the slice category $\mathbf{Loc}/X$ is equivalent the the category $\mathbf{Loc}(\Sh(X))$ of locales internal to $\Sh(X)$. Here $\mathbf{Loc}/X$ denotes the category that has locale maps $f:Y\to X$, for arbitrary locales $Y$ in $\mathbf{Set}$, as objects. Let $f$ and $g$ be such maps. An arrow $h:f\to g$ is given by a commuting triangle of locale maps. 
$$\xymatrix{
Y \ar[dr]_{f} \ar[rr]^{h} & & Z \ar[dl]^{g} \\
& X &}$$
Given a locale map $f:Y\to X$, a locale $\mathcal{I}(f)$ internal to $\Sh(X)$ is constructed as follows. First note that a locale map $f:Y\to X$ induces a geometric morphism $f:\Sh(Y)\to\Sh(X)$. Let $\Omega_{Y}$ be the subobject classifier of $\Sh(Y)$. This object is an internal locale of $\Sh(Y)$. The direct image $f_{\ast}$ of the geometric morphism $f$ is cartesian and preserves internal complete posets. Hence $\mathcal{I}(f)=f_{\ast}(\Omega_{Y})$ is an internal locale of $\Sh(X)$. 

Applying this to the case at hand, a locale internal to $\Sh(\mathcal{V}(A))$ is equivalent to a locale map $L\to\mathcal{V}(A)$, where $L$ is a locale in $\mathbf{Set}$. We can now recognize the continuous map $\pi:\Sigma\to\mathcal{V}(A)$ in (4) as a locale internal to $\Sh(\mathcal{V}(A))$. The spectral presheaf $\underline{\Sigma}$ thus yields a locale in $[\mathcal{V}(A)\op,\mathbf{Set}]$, with associated frame
 \begin{equation}
 \mathcal{O}\underline{\Sigma}(C)=\mathcal{O}\overline{\Sigma}(\downarrow C)=\mathcal{O}\Sigma|_{\downarrow C}=\mathcal{O}\Sigma\cap B_{C,\Sigma_{C}},
 \end{equation}
where $B_{C,\Sigma_{C}}=\{(D,\lambda)\mid D\subseteq C,\lambda\in\Sigma_{D}\}$. 

What are the points of this locale? A point of the internal locale $\underline{\Sigma}$ is equivalent to a continuous cross-section of $\pi$ (this follows from the identification of $\mathbf{Loc}(\Sh(X))$ and $\mathbf{Loc}/X$). This is a locale map
 \begin{equation}
 \phi:\mathcal{V}(A)\to\Sigma,\ \ \phi(C)=(C,\tilde{\phi}(C)),
 \end{equation}
 where, of course, $\tilde{\phi}(C)\in\Sigma_{C}$. As this map is continuous, we obtain
 \begin{equation}
 \phi^{-1}:\mathcal{O}\Sigma\to\mathcal{O}\mathcal{V}(A),\ \ W_{C,\lambda}\mapsto\{D\in\mathcal{V}(A)\mid\tilde{\phi}(D)=\lambda(D)\}.
 \end{equation}
 As $\phi^{-1}(W_{C,\lambda})$ is open, it is downward closed. This implies that if $\tilde{\phi}(C)=\lambda$ and $D\subseteq C$, then $\tilde{\phi}(D)=\lambda|_{D}$. This shows that a point of the locale $\underline{\Sigma}$ corresponds to a global point of the spectral presheaf $\underline{\Sigma}$. So whenever the Kochen-Specker theorem tells us that the spectral presheaf has no global sections (which depends on $A$), this is equivalent to $\underline{\Sigma}$ having no internal points as a locale. This is a localic reformulation of a similar result by Butterfield and Isham \cite{butish1, butish2}.
 
 \subsubsection{The Locales $\Sigma^{\ast}$ and $\mathcal{O}_{cl}\uS$}
 
In the previous subsection we discussed a procedure that yields a locale in $\Sh(\mathcal{V}(A))$ from any contravariant functor $\mathcal{V}(A)\op\to\mathbf{Set}$. The reader may have noticed that the internal locale associated to the spectral presheaf is just the exponential $\mathcal{P}\underline{\Sigma}=\underline{\Omega}^{\uS}$. On the other hand, the propositions in the contravariant approach are represented by closed open subobjects of the spectral presheaf. The propositions therefore correspond to certain points of the locale $\mathcal{P}\underline{\Sigma}$. In principle, we would like to change the topology on $\Sigma$ to a coarser topology in such a way that the opens of $\Sigma$ correspond to the closed open subobjects in this new topology, instead of having an open for every subobject of $\underline{\Sigma}$. We will do something slightly different however\footnote{As we shall see in Subsection 2.2, this makes it easier to draw a comparison between the locale obtained and the spectrum of the covariant approach.}, by taking a topology where the opens of $\Sigma$ correspond to open subobjects of the spectral presheaf. As a basis for this topology, take
 \begin{equation}
 \mathcal{B}=\{B_{C,u}\mid C\in\mathcal{V}(A),u\in\mathcal{O}\Sigma_{C}\},\ B_{C,u}=\{(D,\lambda|_{D})\mid D\leq C, \lambda\in u\}.
 \end{equation}
 We need to check that this defines a basis for a topology on $\Sigma$. If $(C,\lambda)\in\Sigma$, then $(C,\lambda)\in B_{C,\Sigma_{C}}$. Now suppose that 
 \begin{equation}
 (C,\lambda)\in B_{C_{1},u}\cap B_{C_{2},v},\ \ C\subseteq C_{1},C_{2}.
 \end{equation}
 It is demonstrated in \cite{di2} (and \cite[Appendix 1]{di}) that the restriction maps of the spectral presheaf $\rho_{DC}$ (with $D\subseteq C$) are open. Hence $\rho_{CC_{1}}(u)$ and $\rho_{CC_{2}}(v)$ are open neighborhoods of $\lambda$ in $\Sigma_{C}$. Take $w=\rho_{CC_{1}}(u)\cap\rho_{CC_{2}}(v)$, then $(C,\lambda)\in B_{C,w}$ and $B_{C,w}\subseteq B_{C_{1},u},B_{C_{2},v}$. This demonstrates that $\mathcal{B}$ is indeed a basis for a topology. 
 
\begin{dork}
Let $\mathcal{O}\Sigma^{\ast}$ be the topology generated by the basis $\mathcal{B}$. For any $C\in\mathcal{V}(A)$ and $U\subseteq\Sigma$, define the set $U_{C}:=U\cap\Sigma_{C}$. Then $U\in\mathcal{O}\Sigma^{\ast}$ iff:
\begin{enumerate}
\item $\forall C\in\mathcal{V}(A),\ U_{C}\in\mathcal{O}\Sigma_{C}$.
\item If $\lambda\in U_{C}$ and $D\subseteq C$ then $\lambda |_{D}\in U_{D}$.
\end{enumerate}
We use the shorthand notation $\Sigma^{\ast}$ for the topological space $(\Sigma,\mathcal{O}\Sigma^{\ast})$.
\end{dork} 
 
Consider the projection map from (5) once again, but this time with $\Sigma$ equipped with the topology of Definition 2.1. We write this as $\pi:\Sigma^{\ast}\to\mathcal{V}(A)$. The projection map $\pi$ is no longer a local homeomorphism, but it is easily checked to be continuous\footnote{Thus it defines a frame map $\pi^{-1}:\mathcal{O}\mathcal{V}(A)\to\mathcal{O}\Sigma^{\ast}$, where $\mathcal{O}\mathcal{V}(A)$ is the downset topology on the poset $\mathcal{V}(A)$. Recall that a frame morphism is a function that preserves finite meets and all joins.}. This follows from $\pi^{-1}(\downarrow C)=B_{C,\Sigma_{C}}$. As before, $\pi$ defines a locale in $[\mathcal{V}(A)\op,\mathbf{Set}]$. Its associated frame is given by 
 \begin{equation}
 \mathcal{O}\underline{\Sigma}^{\ast}(C)=\mathcal{O}\overline{\Sigma}_{I}(\downarrow C)=\mathcal{O}B_{C,\Sigma_{C}},
 \end{equation}
where $\mathcal{O}B_{C,\Sigma_{C}}$ denotes the relative topology on $B_{C,\Sigma_{C}}\subseteq\Sigma^{\ast}$. Just like before, depending on $A$,  the Kochen-Specker theorem may prevent the locale $\underline{\Sigma}^{\ast}$ from having points. This can be shown in the same way as earlier (cf. the end of 2.1.1). A point of $\underline{\Sigma}^{\ast}$ gives a continuous cross-section
 \begin{equation}
 \phi:\mathcal{V}(A)\to\Sigma^{\ast},\ \ \phi(C)=(C,\tilde{\phi}(C)).
 \end{equation}
 Continuity of the cross-section entails that given any open neighborhood $U$ of $\tilde{\phi}(C)$ in $\Sigma_{C}$, and any $D\in\ \downarrow C$, there exists a $\lambda\in U$, such that $\lambda|_{D}=\tilde{\phi}(D)$. Suppose that $\tilde{\phi}(D)\neq\tilde{\phi}(C)|_{D}$. Then $\rho^{-1}_{DC}(\{\tilde{\phi}(D)\}^{c})$ is an open neighborhood of $\tilde{\phi}(C)$ in $\Sigma_{C}$. Yet it contains no element that restricted to $D$ yields $\tilde{\phi}(D)$. As this contradicts the continuity of $\phi$, we find that $\tilde{\phi}(D)=\tilde{\phi}(C)|_{D}$. So once again, also with this new topology, a point of the locale $\underline{\Sigma}^{\ast}$ amounts to a global point of the spectral presheaf. 
 
We will now compare the natural Heyting algebra structure of $\mathcal{O}\Sigma^{\ast}$ with that of $\mathcal{O}_{cl}\uS$, the set of closed open subobjects of the spectral presheaf. The Heyting algebra structure of $\mathcal{O}_{cl}\uS$, is defined as follows \cite{di2}. Let $\underline{R},\underline{S}\in\mathcal{O}_{cl}\uS$. Then
\begin{equation}
(\underline{R}\wedge\underline{S})(C)=\underline{R}(C)\cap\underline{S}(C),
\end{equation}
\begin{equation}
(\underline{R}\vee\underline{S})(C)=\underline{R}(C)\cup\underline{S}(C),
\end{equation}
The Heyting arrow, and consequently  the negation of $\mathcal{O}_{cl}\underline{\Sigma}$, is given by
\begin{equation}
(\underline{R}\Rightarrow\underline{S})(C)=\text{int}\left(\bigcap_{D\leq C}\{\lambda\in\Sigma_{C}\mid \text{if}\  \lambda|_{D}\in\underline{R}(D)\ \text{then}\ \lambda|_{D}\in\underline{S}(D)\}\right);
\end{equation}
\begin{equation}
(\neg\underline{S})(C)=\text{int}\left(\bigcap_{D\leq C}\{\lambda\in\Sigma_{C}\mid\lambda|_{D}\in\underline{S}(D)^{c}\}\right),
\end{equation}
where \textit{int} means taking the interior and $(-)^{c}$ means taking the set-theoretic complement. Define the map
\begin{equation}
I:\mathcal{O}_{cl}\uS\to\mathcal{O}\Sigma^{\ast}\ \ I(\underline{S})=\coprod_{C\in\mathcal{V}(A)}\underline{S}(C).
\end{equation}
It is easy to check that this map is well-defined and is an injective Heyting algebra morphism. When we consider an $n$-level system $A=M_{n}(\mathbb{C})$, then for any $C\in\mathcal{V}(A)$ we have that $U\subseteq\Sigma_{C}$ is open iff it is closed and open. In that case, $\mathcal{O}_{cl}\underline{\Sigma}\cong\mathcal{O}\Sigma^{\ast}$ as Heyting algebras. 

The Heyting algebra $\mathcal{O}_{cl}\uS$ is in fact a complete Heyting algebra and $I$ preserves arbitrary joins (because it is an open map), making $I$ into a morphism of complete Heyting algebras.

\begin{tut}
The projection map $\pi:\Sigma^{\ast}\to\mathcal{V}(A)$ is continuous and defines a locale $\underline{\Sigma}^{\ast}$ in $[\mathcal{V}(A)\op,\mathbf{Set}]$. There exists an injective morphism of complete Heyting algebras $I: \mathcal{O}_{cl}\underline{\Sigma}\to\mathcal{O}\Sigma^{\ast}$.
\end{tut}

A morphism of complete Heyting algebras is in particular a morphism of frames\footnote{But a morphism of frames need not preserve the implication arrow and therefore need not be a morphism of Heyting algebras}. Let $\Sigma_{cl}$ be the locale (in $\mathbf{Set}$) corresponding to the frame $\mathcal{O}_{cl}\uS$. The injective frame map $I$ defines a surjection of locales $\Sigma^{\ast}\twoheadrightarrow\Sigma_{cl}$. The projection $\pi$ factors through this locale map, giving the following commutative triangles in $\mathbf{Loc}$ and $\mathbf{Frm}$ respectively:
$$\xymatrix{
\Sigma^{\ast} \ar[d]_{\pi} \ar@{->>}[r] & \Sigma_{cl} \ar@{.>}[dl]^{\pi_{cl}} & & \mathcal{O}\Sigma^{\ast} & \ \mathcal{O}_{cl}\uS \ar@{>->}[l]_{I} \\
\mathcal{V}(A) & & & \mathcal{O}\mathcal{V}(A) \ar[u]^{\pi^{-1}} \ar@{.>}[ru]_{\pi_{cl}^{-1}} &}$$
To prove that we have such commuting triangles, let $U\in\mathcal{O}\mathcal{V}(A)$ be any downwards closed set. Define $\underline{S}_{U}:\mathcal{V}(A)\op\to\mathbf{Set}$ by $C\mapsto\Sigma_{C}$ if $C\in U$ and $C\mapsto\emptyset$ if $C\notin U$. It is easy to check that $\underline{S}_{U}\in\mathcal{O}_{cl}\uS$ and that $\pi^{-1}(U)=I(\underline{S}_{U})$. We find the following proposition.

\begin{poe}
Let $\Sigma_{cl}$ be the locale associated to the frame $\mathcal{O}_{cl}\uS$ in $\mathbf{Set}$. Then the map $\pi_{cl}:\Sigma_{cl}\to\mathcal{V}(A)$, defined by $\pi^{-1}_{cl}(U)=\underline{S}_{U}$, is a locale map and thus defines a locale $\underline{\Sigma}_{cl}$ internal to $[\mathcal{V}(A)\op,\mathbf{Set}]$. The map $I:\mathcal{O}_{cl}\uS\to\mathcal{O}\Sigma^{\ast}$ defines an internal surjection of locales $\uS^{\ast}\twoheadrightarrow\uS_{cl}$.
\end{poe}

Despite the fact that $\mathcal{O}_{cl}\uS$ is more closely related to $\pi_{cl}:\Sigma_{cl}\to\mathcal{V}(A)$ than to $\pi:\Sigma^{\ast}\to\mathcal{V}(A)$, in what follows we will only use the map $\pi$ and the space $\Sigma^{\ast}$. The reason is that $\Sigma^{\ast}$ is closely related to the state space of the Bohrification approach (See Corollary 2.18 below).  

Next, we show that internal locale $\underline{\Sigma}^{\ast}$ is compact. 
 \begin{dork}
 Let $L$ be a locale. Then $L$ is \textbf{compact} if for any $S\subseteq L$ such that $1_{L}=\bigvee S$, there is a finite $F\subseteq S$ such that $1_{L}=\bigvee F$. Here $1_{L}$ denotes the top element of $L$. Equivalently, one can say that $L$ is compact if for every ideal $I$ of $L$ such that $\bigvee I=1_{L}$, we have $1_{L}\in I$. 
 \end{dork}
The following definition and lemma help to show that $\underline{\Sigma}^{\ast}$ is compact. A proof of Lemma 2.6 can be found in \cite{jh2}. 
 \begin{dork}
 A continuous map of spaces $f:Y\to X$ is called \textbf{perfect} if the following two conditions are satisfied:
 \begin{enumerate}
 \item $f$ has compact fibres: if $x\in X$ then $f^{-1}(x)$ is compact in $Y$.
 \item $f$ is closed: if $C$ is closed in $Y$, then $f(C)$ is closed in $X$.
 \end{enumerate}
 \end{dork}
\begin{lem}{(\cite{jh2}, Proposition 1.1)}
Let $f:Y\to X$ be continuous. If $f$ is perfect, then the internal locale $\mathcal{I}(f)=f_{\ast}(\Omega_{\Sh(Y)})$ in $\Sh(X)$ is compact.
\end{lem}

In the previous lemma, $f_{\ast}$ denotes the direct image part of the geometric morphism associated to $f$, and $\Omega_{\Sh(Y)}$ denotes the subobject classifier of $\Sh(Y)$.

\begin{cor}
The locale $\uS^{\ast}$ in $[\mathcal{V}(A)\op,\mathbf{Set}]$ is compact.
\end{cor}

\begin{proof}
If we can show that $\pi:\Sigma^{\ast}\to\mathcal{V}(A)$ is a closed map that has compact fibres, then $\underline{\Sigma}^{\ast}$ is a compact locale. The fact that $\pi$ has compact fibres is evident. Let $X$ be a closed subset of $\Sigma^{\ast}$. Then $X=\bigcap_{i\in I}X_{C_{i},u_{i}}$ where $X_{C_{i},u_{i}}=B_{C_{i},u_{i}}^{c}$ for some $C_{i}\in\mathcal{V}(A)$ and $u_{i}\in\mathcal{O}\Sigma_{C_{i}}$. If $C\notin\pi(X)$ then for every $\lambda\in\Sigma_{C}$ we have $(C,\lambda)\in X^{c}$. Take any $D\subseteq C$ and $\lambda'\in\Sigma_{D}$. There is a $\lambda\in\Sigma_{C}$ such that $\lambda|_{D}=\lambda'$. As $(C,\lambda)\in X^{c}$ there is some $j\in I$ such that $(C,\lambda)\in B_{C_{j},u_{j}}$. By definition, $(D,\lambda')\in B_{C_{j},u_{j}}$. It follows that for every $D\subseteq C$ and any $\lambda'\in\Sigma_{D}$ we have $(D,\lambda')\notin X$. We find that $\pi(X)^{c}$ is downward closed, hence open. This proves that $\pi$ is closed.
\end{proof}

\begin{dork}{(\cite{jh4}, III.1, 1.1)}
Let $L$ be a locale and $x,y\in L$. Then $x$ is \textbf{well inside} $y$, denoted by $x\eqslantless y$, if there exists a $z\in L$ such that $z\wedge x=0_{L}$ and $z\vee y=1_{L}$. A locale $L$ is called \textbf{regular} if every $x\in L$ satisfies
\begin{equation}
x=\bigvee\{y\in L|y\eqslantless x\}.
\end{equation}
\end{dork}

Regularity of the internal locale $\underline{\Sigma}^{\ast}$ can conveniently be checked from its external description $\pi$, as shown by the following lemma.

\begin{lem}{(\cite{jh3} Lemma 1.2)}
Let $f:Y\to X$ be continuous. Then $f_{\ast}(\Omega_{\Sh(Y)})$ is regular iff for any open $U\in\mathcal{O}Y$ and $y\in U$ there is a neigborhood $N$ of $f(y)$ in $X$, and there exist opens $V,W\in\mathcal{O}Y$ such that $y\in V$, $V\cap W=\emptyset$ and $f^{-1}(N)\subseteq U\cup W$.
\end{lem}

Up to this point it did not matter wether we excluded the trivial algebra $\mathbb{C}1$ from the set of contexts or not. For the discussion of regularity that follows, it does matter, so we need to be precise about it. Usually the trivial algebra is excluded in the contravariant approach\footnote{The author could not find a reason for doing so in the papers on the contravariant approach, but on the next page we will see that removing the trivial context has an impact on the Heyting negation of $\mathcal{O}_{cl}\uS$}. However, in discussions of composite systems in the contravariant approach (see e.g. Section 11 of \cite{di}) the trivial context is included. For the moment, we will include the trivial subalgebra as a context.

\begin{cor}
Let $A$ be a von Neumann algebra such that $\mathcal{V}(A)\neq\{\mathbb{C}\cdot1\}$. Then the locale $\uS^{\ast}$ in $[\mathcal{V}(A)\op,\mathbf{Set}]$ is not regular.
\end{cor}

\begin{proof}
By the previous lemma, $\underline{\Sigma}^{\ast}$ is regular iff for any $U\in\mathcal{O}\Sigma^{\ast}$ and any $(C,\lambda)\in U$ there exist opens $V,W\in\mathcal{O}\Sigma^{\ast}$ such that $(C,\lambda)\in V$, $V\cap W=\emptyset$ and $B_{C,\Sigma_{C}}\subseteq U\cup W$. By assumption, there exists a context $C$ such that $\Sigma_{C}$ has at least two elements. This follows from the Gelfand-Mazur Theorem, which implies that if $\Sigma_{C}$ is a singleton, then $C\cong\mathbb{C}$. Take any two distinct $\lambda_{1},\lambda_{2}\in\Sigma_{C}$. We have $(C,\lambda_{1})\in U:=B_{C,\Sigma_{C}\backslash\{\lambda_{2}\}}$. If $\underline{\Sigma}^{\ast}$ is regular, there are $V,W\in\mathcal{O}\Sigma^{\ast}$ such that $(C,\lambda_{1})\in V$, $(C,\lambda_{2})\in W$ and $V\cap W=\emptyset$. In particular, for every $D\subseteq C$ we find that $\lambda_{1}|_{D}\neq\lambda_{2}|_{D}$. For $D=\mathbb{C}\cdot1$ this condition is not satisfied, so that the compact locale $\underline{\Sigma}^{\ast}$ is not regular.
\end{proof}

Hence the space $\Sigma^{\ast}$ with the topology $\mathcal{O}\Sigma^{\ast}$ is not regular, but it does satisfy the $T_{0}$-axiom. If we leave out the trivial context, Corollary 2.10 becomes slightly weaker. For example, the locale $\uS^{\ast}$ associated to the von Neumann algebra $A=M_{2}(\mathbb{C})$ is regular (the space $\Sigma^{\ast}$ has the discrete topology in this case).  In general the locale $\uS^{\ast}$ is not regular. For example, using the proof of Corollary 2.10 it is not hard to show that the locale $\uS^{\ast}$ is not regular for $A=M_{n}(\mathbb{C})$, for any $n>2$.

The nonregularity of the locale $\uS^{\ast}$ can also be seen logically.  If $\uS^{\ast}$ were regular, then (in the internal language of $[\mathcal{V}(A)\op,\mathbf{Set}]$) for any $\underline{U}\in\mathcal{O}\uS^{\ast}$ we would have
\begin{equation}
\underline{U}=\bigvee\{\underline{V}\in\mathcal{O}\uS^{\ast}\mid\neg\neg\underline{V}=\underline{V},\ \underline{V}\subseteq\underline{U}\}.
\end{equation}
We investigate the nonregularity of $\uS^{\ast}$ by taking a closer look at the negation $\neg$ of $\mathcal{O}_{cl}\uS$. For the moment, we include the trivial context in $\mathcal{V}(A)$, as it makes the investigation of the negation easier.

Let $U\in\mathcal{O}\Sigma^{\ast}$ be an open defined by a proposition $\underline{S}\in\mathcal{O}_{cl}\underline{\Sigma}$ in the sense that $I(\underline{S})=U$. The proposition $\neg\underline{S}$ corresponds to the open $\neg U$, which is the interior of the complement of $U$. We find that $\lambda\in(\neg U)_{C}$ iff for all $D\subseteq C$ we have $\lambda |_{D}\notin U_{D}$. If $U\neq\emptyset$, then $U_{\C\cdot 1}=\Sigma_{\C\cdot 1}$. If $\lambda\in\Sigma_{C}$, then $\lambda|_{\C\cdot 1}\in U_{C}$. It follows that $\neg U=\emptyset$. Thus the negation of any proposition $\underline{S}$ that is not the empty subobject, is $\bot$. This leads to a new proof of Corollary 2.10.The only elements $\underline{V}\in\mathcal{O}\uS^{\ast}$ such that $\neg\neg\underline{V}=\underline{V}$ are the top and bottom element of the frame. Again we conclude that $\uS^{\ast}$ is nonregular, but we also see that the double negation of any element $\underline{S}\in\mathcal{O}_{cl}\uS$ is either the bottom element (if $\underline{S}=\underline{\bot}$), or the top element (if $\underline{S}\neq\underline{\bot}$). 

Next we remove the trivial context from the set $\mathcal{V}(A)$, and the Heyting negation of $\mathcal{O}_{cl}\underline{\Sigma}$ will not be trivial anymore. Still, we have a similar situation. Take, for example an $n$-level system $A=M_{n}(\mathbb{C})$, with $n\geq3$, and pick $C\in\mathcal{V}(A)$ such that $\Sigma_{C}$ has at least 3 elements. In the contravariant approach one is often interested in propositions of the form $\underline{[a\epsilon\Delta]}$, as will be discussed in the next section. Their precise definition does not matter at the moment, but what does matter is that it follows from this definition that $\underline{[a\epsilon\Delta]}(C)\neq\emptyset$. Pick a $\lambda_{1}\in\underline{[a\epsilon\Delta]}(C)$ and another $\lambda_{2}\in\Sigma_{C}$. We assumed that there is a third distinct element $\lambda_{3}\in\Sigma_{C}$. This $\lambda_{3}$ corresponds to a projection operator $p$.\footnote{For an n-level system the characters $\lambda_{i}$ in $\Sigma_{C}$ correspond to mutually orthogonal projection operators $p_{i}$ such that $\sum_{i}p_{i}=1$ and $\lambda_{i}(p_{j})=\delta_{ij}$ with $\delta_{ij}$ the Kronecker delta.} Let $D=\{p\}''$ be the context generated by the projection $p$. In $D$ we have $\lambda_{1}|_{D}=\lambda_{2}|_{D}$. We conclude that $\underline{\neg[a\epsilon\Delta]}(C)=\emptyset$ for an $n$-level system with $n\geq3$, and any context $C$ such that $C$ is not an atom in $\mathcal{V}(A)^{\times}$ in the sense that $\Sigma_{C}$ has at least 3 elements.

What does $\uS^{\ast}$ not being regular, and the related behavior of the negation, teach us? At the level of the locale, it shows that $\uS^{\ast}$ cannot be the spectrum of some internal C*-algebra, hence in that respect it is different from the covariant state locale. At the level of the negation, it is less clear what the previous observations mean. The negation of $\uS^{\ast}$ (and the related one of $\mathcal{O}_{cl}\uS$) follows from the perspective of $\uS^{\ast}$ as a locale in $[\mathcal{V}(A)^{\op},\Set]$, but it is not clear to the author to what extent the internal language of the topos $[\mathcal{V}(A)^{\op},\Set]$ is important for the contravariant approach. Also note that the Heyting arrow and the associated negation of $\mathcal{O}_{cl}\uS$ do not seem to play an important role in the contravariant approach thus far.

\subsection{The Spectrum of the Bohrification of $A$}

In the previous section we looked at the state object $\uS$ of the contravariant approach as a locale in order to compare it with the state space of the covariant approach, the spectrum  $\underline{\Sigma}_{\underline{A}}$ of $\underline{A}$. Next, we turn our attention to this locale $\underline{\Sigma}_{\underline{A}}$. The bulk of this section will be devoted to a new description (Corollary 2.18) of the spectrum $\underline{\Sigma}_{\underline{A}}$, which makes it easier to compare with $\underline{\Sigma}^{\ast}$. This description is also helpful when considering daseinisation in the covariant approach, as explained in Section 3.2. We start with a short discussion of constructive Gelfand duality and its application to covariant topos quantum logic. The relevant references are the pioneering work of Banaschewski and Mulvey on Gelfand duality in topoi \cite{banmul, banmul2, banmul3}, the more explicit and fully constructive description of the Gelfand isomorphism by Coquand \cite{coq} and Coquand and Spitters \cite{coqsp}, and finally the work  \cite{hls, hls2} on covariant quantum logic by Heunen, Landsman and Spitters.

\begin{dork}
Let $C$ be a commutative C*-algebra, and define
\begin{equation}
C^{+}=\{a\in C_{sa}\mid a\geq 0\}=\{a\in C\mid\exists b\in C, a=b^{\ast}b\}.
\end{equation}
Now define the following relation on $C^{+}$: $a\precsim b$ whenever there is an $n\in\mathbb{N}$ such that $a\leq nb$. Define the equivalence relation $a\approx b$ whenever $a\precsim b$ and $b\precsim a$. Let $L_{C}$ denote the set of equivalence classes, and let $D^{C}_{a}$ denote the image $[a^{+}]$ in $L_{C}$, where $a=a^{+}-a^{-}$ is the decomposition in positive elements of $a\in C_{sa}$.
\end{dork}

The lattice operations on $C_{sa}$ (with respect to the partial order $a\leq b$ iff $(b-a)\in C^{+}$) respect the equivalence relation of the definition, turning $L_{C}$ into a distributive lattice. The following proposition connects the lattice $L_{C}$ to the spectrum $\Sigma_{C}$. The details can be found in \cite[Appendix A.1]{hls}.

\begin{poe}
The frame $\text{RIdl}(L_{C})$ of regular ideals of the distributive lattice $L_{C}$, is isomorphic to $\mathcal{O}\Sigma_{C}$.
\end{poe}

The internal Gelfand spectrum $\uS_{\uA}$ for a C*-algebra $A$ is calculated as follows. First, define the poset $\mathcal{C}(A)$ of all unital commutative C*-subalgebras of $A$, where the order is given by inclusion. Next consider the topos $[\mathcal{C}(A),\mathbf{Set}]$ of covariant functors $\mathcal{C}(A)\to\mathbf{Set}$. In particular, consider the functor 
\begin{equation}
\underline{A}:\mathcal{C}(A)\to\mathbf{Set},\ \  \underline{A}(C)=C,\ \  \underline{A}(i_{DC})=i_{DC},
\end{equation}
where $i_{DC}$ is the inclusion $D\subseteq C$. The object $\underline{A}$, called the Bohrification of $A$, is a commutative C*-algebra internal to $[\mathcal{C}(A),\mathbf{Set}]$. As shown in \cite{hls}, its internal spectrum can be described in various ways. We will use the following description, which can be found in \cite[Theorem 29]{hls}:
\begin{tut}
Let $\underline{L}_{\underline{A}}:\mathcal{C}(A)\to\mathbf{Set}$ be the functor 
\begin{equation}
\underline{L}_{\underline{A}}(C)=L_{C},\ \  \underline{L}_{\underline{A}}(i_{CE}):L_{C}\to L_{E},\ \ \  D_{a}^{C}\mapsto D_{a}^{E}.
\end{equation}
Given any $C\in\mathcal{C}(A)$, the set $\mathcal{O}\uS_{\uA}(C)$ consists of all subfunctors $\underline{U}\in Sub(\underline{L}_{\underline{A}} |_{\uparrow C})$ that satisfy the following property: for every $C'\supseteq C$ we have $D^{C'}_{a}\in\underline{U}(C')$ if and only if for every $q\in\mathbb{Q}^{+}$ there exists a finite set $U_{0}\subseteq\underline{U}(C')$ such that $D^{C'}_{a-q}\leq\bigvee U_{0}$.
\end{tut} 
In this theorem, the functor $\underline{L}_{\underline{A}}|_{\uparrow C}:(\uparrow C)\to\mathbf{Set}$ denotes the restriction of  $\underline{L}_{\underline{A}}$ to $\uparrow C=\{C'\in\mathcal{C}(A)\mid C\subseteq C'\}$.

Let $\mathcal{C}(A)$ denote the set of contexts with the Alexandrov topology (in which the open sets are exactly the upwards closed sets). Thus $V\in\mathcal{O}\mathcal{C}(A)$ if $C'\in V$ whenever $C\in V$ and $C\subseteq C'$. Defining a covariant functor $\underline{F}:\mathcal{C}(A)\to\mathbf{Set}$ is equivalent to giving a sheaf $\overline{F}:\mathcal{O}\mathcal{C}(A)\op\to\mathbf{Set}$ through the correspondence
\begin{equation}
\overline{F}(\uparrow C)=\underline{F}(C).
\end{equation}
Just as in the previous subsection, a locale in $\Sh(\mathcal{C}(A))$ is equivalent to a locale map $L\to\mathcal{C}(A)$ in $\mathbf{Set}$. The locale $\uS_{\uA}$ is externally described by the frame map
\begin{equation}
\pi^{-1}:\mathcal{O}\mathcal{C}(A)\to\mathcal{O}\uS_{\uA}(\C\cdot 1),\ \ (\uparrow C)\mapsto \underline{L}_{\uA}|_{\uparrow C}.
\end{equation}

Next, we prove some lemmas that will be helpful in the investigation of $\uS_{\uA}$.

\begin{lem}
For $C\in\mathcal{C}(A)$ and $a\in C_{sa}$, let $X^{C}_{a}=\{\lambda\in\Sigma_{C}\mid\lambda(a)>0\}$. If $D^{C}_{a}=D^{C}_{b}$ in $L_{C}$, for $a,b\in C_{sa}$,  then $X^{C}_{a}=X^{C}_{b}$ in $\Sigma_{C}$. In other words, the map
\begin{equation}
L_{C}\to\mathcal{O}\Sigma_{C},\ \ D^{C}_{a}\mapsto X^{C}_{a}.
\end{equation}
is well defined.
\end{lem}

\begin{proof}
Let $\lambda\in X^{C}_{a}$. As $D^{C}_{a}=D^{C}_{b}$, there exists an $n\in\mathbb{N}$ such that $a^{+}\leq nb^{+}$ in $C^{+}$. If $\lambda(a)>0$, this can only hold if $\lambda(b)>0$. Consequently $\lambda\in X^{C}_{b}$. The converse is analogous.
\end{proof}

Note that the opens $X^{C}_{a}$, with $a\in A_{sa}$, form a basis for the topology on $\Sigma_{C}$.

\begin{lem}
Let $D^{C}_{a}\in L_{C}$ and suppose $U\subseteq L_{C}$ satisfies the condition that for every $q\in\mathbb{Q}^{+}$ there exists a finite set $U_{0}\subseteq U$ such that $D^{C}_{a-q}\leq\bigvee U_{0}$. We will denote this situation as $D^{C}_{a}\drie_{C} U$. Then the following equality holds in $\Sigma_{C}$:
\begin{equation}
X^{C}_{a}=\bigcup_{D^{C}_{b}\drie_{C}\{D^{C}_{a}\}}X^{C}_{b}.
\end{equation}
\end{lem}

\begin{proof}
As $D^{C}_{a}\drie_{C}\{D^{C}_{a}\}$, we immediately have $X^{C}_{a}\subseteq\bigcup_{D^{C}_{b}\drie_{C}\{D^{C}_{a}\}}X^{C}_{b}$. Assume that $\lambda\in X^{C}_{b}$, where $D^{C}_{b}\drie_{C}\{D^{C}_{a}\}$. By definition of $\drie_{C}$, for every $q\in\Q^{+}$ we have $D^{C}_{b-q}\leq D^{C}_{a}$. By definition of $L_{C}$, there exists an $n\in\mathbb{N}$ such that $na^{+}-(b-q)^{+}\geq 0$. By assumption, $\lambda(b)>0$. Pick any $q\in\Q^{+}$ such that $\lambda(b-q)>0$. As $na^{+}-(b-q)^{+}\geq 0$, and $\lambda(b-q)>0$, it follows that $\lambda(a)>0$. This shows that $\bigcup_{D^{C}_{b}\drie_{C}\{D^{C}_{a}\}}X^{C}_{b}\subseteq X^{C}_{a}$, completing the proof.\footnote{Under the identification $\mathcal{O}\Sigma_{C}\cong\text{RIdl}(L_{C})$, Lemma 2.15 boils down to the claim that the map of Lemma 2.14 is equal to the canonical map \cite[(67),(78)]{hls}: 
\begin{equation}
f:L_{C}\to\text{RIdl}(L_{C})\cong\mathcal{F}(L_{C},\drie_{C}),\ \ D^{C}_{a}\mapsto\mathcal{A}(\downarrow D^{C}_{a}).\nonumber
\end{equation}}
\end{proof}

Now we return to the set $\Sigma=\{(C,\lambda)\mid C\in\mathcal{C}(A),\ \lambda\in\Sigma_{C}\}$ of the previous subsection. This time we equip it with a different topology.
\begin{dork}
The space $\Sigma_{\ast}$ is the set $\Sigma$ with the topology $\mathcal{O}\Sigma_{\ast}$, where $U\in\mathcal{O}\Sigma_{\ast}$ iff the following two conditions are satisfied:
\begin{enumerate}
\item $\forall C\in\mathcal{C}(A),\ U_{C}\in\mathcal{O}\Sigma_{C}$; 
\item If $\lambda\in U_{C}, C\subseteq C'$ and $\lambda'\in\Sigma_{C'}$ such that $\lambda'|_{C}=\lambda$, then $\lambda'\in U_{C'}$.
\end{enumerate}
\end{dork}
We leave it to the reader to verify that this defines a topology. The next theorem and corollary show that, up to isomorphism,  $\Sigma_{\ast}$ is the external description of the spectrum $\uS_{\uA}$.

\begin{tut}
Let $\uS_{\uA}$ be the spectrum of $\underline{A}$. Let $\mathcal{O}\Sigma_{\ast}$ be as in Definition 2.16, and $\underline{U}$ as in Theorem 2.13. Then the map
\begin{equation}
\Psi: \mathcal{O}\uS_{\uA}(\mathbb{C}\cdot1)\to\mathcal{O}\Sigma_{\ast},\ \ \Psi(\underline{U})_{C}=\bigcup_{D^{C}_{a}\in\underline{U}(C)}X^{C}_{a},
\end{equation}
is well defined and defines an isomorphism of frames.
\end{tut} 

\begin{cor}
The projection map
\begin{equation}
\pi:\Sigma_{\ast}\to\mathcal{C}(A),\ \ (C,\lambda)\mapsto C,
\end{equation}
is continuous and defines a locale $\uS_{\ast}$ internal to $[\mathcal{C}(A),\mathbf{Set}]$. Up to isomorphism, this locale is the internal spectrum of $\underline{A}$. The frame associated to this locale is given by
\begin{equation}
\mathcal{O}\uS_{\ast}: C\to\mathcal{O}\Sigma_{\ast}|_{\uparrow C}=\{U\in\mathcal{O}\Sigma_{\ast}\mid U\subseteq\coprod_{C'\in(\uparrow C)}\Sigma_{C'}\},
\end{equation}
where for  $C\subseteq C'$ the transition map $\mathcal{O}\uS_{\ast}(C)\to\mathcal{O}\uS_{\ast}(C')$ is given by

\noindent $U\mapsto\coprod_{C''\in(\uparrow C')}U_{C''}$. Hence $\uSs$ is ismorphic  to $\uS_{\uA}$ as a locale.
\end{cor} 
Before we get started with the proof of Theorem 2.17, we first need one more lemma.

\begin{lem}
Let $C$ be a commutative C*-algebra, and take $a,b\in C^{+}$. If $X^{C}_{a}=X^{C}_{b}$ then $D^{C}_{a}\vartriangleleft_{C}\{D^{C}_{b}\}$.
\end{lem}

\begin{proof}
Let $\lambda\notin X^{C}_{a}$ and $q\in\mathbb{Q}^{+}$. Then $\lambda(b)=\lambda(a)=\lambda((a-q)^{+})=0$. Take any open neighboorhood $U_{\lambda}$ of $\lambda$ in $\Sigma_{C}$ such that for all $\lambda'\in U_{\lambda}$ we have $\lambda'((a-q)^{+})=0$. Thus on $U_{\lambda}$ we have $(a-q)^{+}\leq b$. Take such a neighborhood for every $\lambda\notin X^{C}_{a}$, and define $F=(\bigcup_{\lambda\notin X^{C}_{a}}U_{\lambda})^{c}$. This is a closed set in the compact space $\Sigma_{C}$, hence it is compact too. For every $\lambda\in F$ we have $\lambda(b)>0$. Because $F$ is compact, the Gelfand transform $\hat{b}$ assumes its minimum value on $F$. Define $\delta=\text{min}\{\lambda(b)\mid \lambda\in F\}$, then $\delta>0$. Also, define $\alpha=\text{max}\{\lambda((a-q)^{+})\mid\lambda\in F\}$. Choose any $n\in\mathbb{N}$ such that $n>\alpha/\delta$. Then
\begin{equation}
\forall_{\lambda\in F}\ \ (\hat{a}-q)^{+}(\lambda)\leq\alpha<n\delta\leq n\hat{b}(\lambda).
\end{equation}
In fact, we found that for all $\lambda\in\Sigma_{C}$ we have $(\hat{a}-q)^{+}(\lambda)< n\hat{b}(\lambda)$. Consequently, $D^{C}_{a-q}\leq D^{C}_{b}$, proving the lemma.
\end{proof}

Now we can prove Theorem 2.17.

\begin{proof}
The first thing to check is that $\Psi$ is well defined. Because of Lemma 2.14, we only need to check that $\Psi(\underline{U})$ is open in $\Sigma_{\ast}$. First of all, note that $\Psi(\underline{U})_{C}\in\mathcal{O}\Sigma_{C}$. Next, assume that $\lambda\in\Psi(\underline{U})_{C}$. Take any $C'\supseteq C$ and $\lambda'\in\rho^{-1}_{C'C}(\lambda)$. As $\lambda\in\Psi(\underline{U})_{C}$, there is a $D^{C}_{a}\in\underline{U}(C)$ such that $\lambda(a)>0$. By definition, $\underline{U}$ is a subobject of $\underline{L}_{\underline{A}}$, so that $D^{C'}_{a}\in\underline{U}(C')$. Furthermore $\lambda'(a)=\lambda'|_{C}(a)=\lambda(a)>0$, hence $\lambda'\in\Psi(\underline{U})_{C'}$. This completes the proof that $\Psi(\underline{U})$ is open in $\Sigma_{\ast}$. 

Next, we prove that $\Psi$ is injective. Assume that for $\underline{U},\underline{V}\in\mathcal{O}\uS_{\uA}(\mathbb{C}\cdot 1)$ we have $\Psi(\underline{U})=\Psi(\underline{V})$. Pick any $D^{C}_{a}\in\underline{U}(C)$. If $\lambda\in\Sigma_{C}$ is such that $\lambda(a)>0$, then $\lambda\in\Psi(\underline{U})_{C}=\Psi(\underline{V})_{C}$. This in turn implies that there is a $D^{C}_{b_{\lambda}}\in\underline{V}(C)$ such that $\lambda(b_{\lambda})>0$. For every $\lambda\in X^{C}_{a}$ there is a $b_{\lambda}\in C^{+}$ such that $D^{C}_{b_{\lambda}}\in\underline{V}(C)$ and $\lambda\in X^{C}_{b_{\lambda}}$. The opens $X^{C}_{b_{\lambda}}$ cover $X^{C}_{a}$ in $\Sigma_{C}$. Take any $q\in\Q^{+}$. Then the closure of $X^{C}_{(a-q)^{+}}$ is a closed subset of $X^{C}_{a}$. As $\Sigma_{C}$ is compact, there is a finite subcover $\{X^{C}_{b_{i}}\}_{i=1}^{n}$ of (the closure of) $X^{C}_{(a-q)^{+}}$. Using $a,b_{i}\in C^{+}$, it is clear that $X^{C}_{b_{i}}\cap X^{C}_{(a-q)^{+}}=X^{C}_{b_{i}\wedge (a-q)^{+}}$. It is easily checked that $D^{C}_{b_{i}\wedge (a-q)^{+}}\leq D^{C}_{b_{i}}$, and subsequently that
\begin{equation}
D^{C}_{b_{i}\wedge (a-q)^{+}}\vartriangleleft_{C}\{D^{C}_{b_{i}}\}\subseteq\underline{V}(C),
\end{equation}
which in turn implies that $D^{C}_{b_{i}\wedge (a-q)^{+}}\in\underline{V}(C)$ for all $i\in\{1,...,n\}$.

Next, suppose that $D^{C}_{x},D^{C}_{y}\in\underline{V}(C)$ for certain $x,y\in C^{+}$. Then
\begin{equation}
\forall q\in\mathbb{Q}^{+}\ \ D^{C}_{(x\vee y)-q}\leq D^{C}_{x\vee y}=D^{C}_{x}\vee D^{C}_{y}\ .
\end{equation}
We conclude that $D^{C}_{x\vee y}\vartriangleleft_{C}\{D^{C}_{x},D^{C}_{y}\}$, which implies $D^{C}_{x\vee y}\in\underline{V}(C)$. Define $d=\bigvee_{i=1}^{n}b_{i}\wedge (a-q)^{+}$. Then by the previous argument $D^{C}_{d}\in\underline{V}(C)$. Thus far we found that if $D^{C}_{a}\in\underline{U}(C)$, then for every $q\in\Q^{+}$ there is a $d\in C^{+}$ such that $X^{C}_{d}=X^{C}_{(a-q)^{+}}$ and $D^{C}_{d}\in\underline{V}(C)$. By Lemma 2.19, $D^{C}_{(a-q)^{+}}\vartriangleleft_{C}\{D^{C}_{d}\}\subseteq\underline{V}(C)$. For every $q\in\Q^{+}$ we found that $D^{C}_{(a-q)^{+}}\in\underline{V}(C)$. By the definition of $\drie_{C}$ it follows that $D^{C}_{a}\in\underline{V}(C)$. This proves that for every $C\in\mathcal{C}(A)$ we have $\underline{U}(C)\subseteq\underline{V}(C)$. The reverse inclusion can be found in exactly the same way, proving injectivity of $\Psi$.

Next we prove that $\Psi$ is surjective. Pick any $U\in\mathcal{O}\Sigma_{\ast}$.  Define
\begin{equation}
\underline{U}:C\to\{D^{C}_{a}\in L_{C}\mid X_{a}\subseteq U_{C}\}.
\end{equation}
It follows from the definition that $\Psi(\underline{U})\subseteq U$. If $\lambda\in U_{C}$, then there exists an $X^{C}_{b}$ such that $\lambda\in X^{C}_{b}\subseteq U_{C}$. This is because the $X^{C}_{b}$ make up a basis for  $\mathcal{O}\Sigma_{C}$ and $U_{C}$ is open in $\Sigma_{C}$. In this way we find that for all $C\in\mathcal{C}(A)$ we have $U_{C}\subseteq\Psi(\underline{U})_{C}$. We need to check that $\underline{U}$ as defined in (34) is an element of $\mathcal{O}\uS_{\uA}(\mathbb{C}\cdot 1)$. First we check that $\underline{U}$ is a subobject of $\underline{L}_{\underline{A}}$. Assume that $D^{C}_{a}\in\underline{U}(C)$. By definition, $X^{C}_{a}\subseteq U_{C}$. Thus if $\lambda\in\Sigma_{C}$ implies $\lambda(a)>0$, then $\lambda\in U_{C}$. Now suppose that $C'\supseteq C$, $\lambda'\in\Sigma_{C'}$ and $\lambda'(a)>0$. Then $\lambda'|_{C}(a)>0$, hence $\lambda'|_{C}\in U_{C}$. By definition of $\mathcal{O}\Sigma_{\ast}$, we have $\lambda'\in U_{C'}$. We found $X^{C'}_{a}\subseteq U_{C'}$, implying $D^{C'}_{a}\in\underline{U}(C')$. In short, if $D^{C}_{a}\in\underline{U}(C)$ and $C'\supseteq C$, then $D^{C'}_{a}\in\underline{U}(C')$, proving that $\underline{U}$ is a subobject of $\underline{L}_{\underline{A}}$.

It remains to show that
\begin{equation}
D^{C}_{a}\vartriangleleft_{C}\underline{U}(C)\Rightarrow D^{C}_{a}\in\underline{U}(C).
\end{equation}
Assume that $D^{C}_{a}\vartriangleleft_{C}\underline{U}(C)$. In other words, for every $q\in\mathbb{Q}^{+}$, there exist $D^{C}_{b_{1}},...,D^{C}_{b_{n}}\in\underline{U}(C)$ such that $D^{C}_{a-q}\leq D^{C}_{b}$, with $b=\bigvee_{i=1}^{n}b_{i}$. For every $i\in\{1,...,n\}$, one has $X^{C}_{b_{i}}\subset U_{C}$, implying that $X^{C}_{b}=\bigcup_{i=1}^{n}X^{C}_{b_{i}}\subseteq U_{C}$. By definition of $\underline{U}$ we find $D^{C}_{b}\in\underline{U}(C)$. The assumption translates to
\begin{equation}
\forall q\in\mathbb{Q}^{+}\ \ \exists D^{C}_{b}\in\underline{U}(C)\ \ \ D^{C}_{a-q}\leq D^{C}_{b}.
\end{equation}
If $D^{C}_{a-q}\leq D^{C}_{b}$, then there exists an $n\in\mathbb{N}$ such that $(a-q)^{+}\leq nb^{+}$. This can only hold if $X^{C}_{(a-q)}\subseteq X^{C}_{b}$. Hence
\begin{equation}
X^{C}_{(a-q)}\subseteq X^{C}_{b}\subseteq U_{C}\Rightarrow\  \forall q\in\mathbb{Q}^{+},\ X^{C}_{(a-q)}\subseteq U_{C}.
\end{equation}
Let $\lambda\in X^{C}_{a}$, so that $\lambda(a)>0$. There exists a $q\in\mathbb{Q}^{+}$ such that $\lambda((a-q)^{+})>0$. By definition of $X^{C}_{(a-q)^{+}}$, $\lambda\in X^{C}_{(a-q)^{+}}\subseteq U_{C}$. In short, $X^{C}_{a}\subseteq U_{C}$ and $D^{C}_{a}\in\underline{U}(C)$. This settles surjectivity of $\Psi$.

Next we prove that $\Psi$ preserves all finite meets. The empty meet case is easy. The empty meet in $\mathcal{O}\uS_{\uA}(\mathbb{C}\cdot1)$ is $\underline{L}_{\underline{A}}$. This is mapped by $\Psi$ to $\Sigma_{C}$, which is the empty meet of $\mathcal{O}\Sigma_{C}$. Now consider binary meets. The meet operation of $\mathcal{O}\uS_{\uA}(\mathbb{C}\cdot1)$ is inherited from $\text{Sub}(\underline{L}_{\underline{A}})$. Let $\underline{U},\underline{V}\in\mathcal{O}\uS_{\uA}(\mathbb{C}\cdot1)$. Then
\begin{equation}
(\underline{U}\wedge\underline{V})(C)=\underline{U}(C)\cap\underline{V}(C);
\end{equation}
\begin{equation}
\Psi(\underline{U}\wedge\underline{V})_{C}=\{\lambda\in\Sigma_{C}\mid \exists D^{C}_{a}\in\underline{U}(C)\cap\underline{V}(C),\text{s.t.}\  \lambda(a)>0\};
\end{equation}
\begin{equation}
(\Psi(\underline{U})\cap\Psi(\underline{V}))_{C}=\{\lambda\in\Sigma_{C}\mid\exists D^{C}_{a}\in\underline{U}(C),\exists D^{C}_{b}\in\underline{V}(C),\lambda(a),\lambda(b)>0\}.
\end{equation}
So clearly $\Psi(\underline{U}\wedge\underline{V})\subseteq\Psi(\underline{U})\cap\Psi(\underline{V})$. Take any $\lambda\in\Psi(\underline{U})_{C}\cap\Psi(\underline{V})_{C}$. Then there exists a $D^{C}_{a}\in\underline{U}(C)$ and a $D^{C}_{b}\in\underline{V}(C)$ such that $\lambda(a)>0$ and $\lambda(b)>0$. Define $c\in C^{+}$ by $c=a\wedge b$. Then $D^{C}_{c}\in\underline{U}(C),\underline{V}(C)$ and $\lambda(c)>0$ showing that $\lambda\in\Psi(\underline{U}\wedge\underline{V})_{C}$. This completes the proof that $\Psi$ preserves finite meets.

The last thing we need to prove is that $\Psi$ preserves all joins. Take a family of objects $\{\underline{U}_{i}\}_{i\in I}$ in $\mathcal{O}\uS_{\uA}(\mathbb{C}\cdot1)$. Then
\begin{equation}
\left(\bigvee_{i\in I} \underline{U}_{i}\right)=\{D^{C}_{a}\in L_{C}\mid D^{C}_{a}\vartriangleleft_{C}\bigcup_{i\in I}\underline{U}_{i}(C)\};
\end{equation}
\begin{align}
\left(\bigcup_{i\in I}\Psi(\underline{U}_{i})\right)_{C} &= \bigcup_{i\in I}\{\lambda\in\Sigma_{C}\mid \exists D^{C}_{a}\in\underline{U}_{i}(C)\ s.t.\ \lambda(a)>0\}\\
&= \{\lambda\in\Sigma_{C}\mid \exists i\in I\  \exists D^{C}_{a}\in\underline{U}_{i}(C)\ s.t.\ \lambda(a)>0\}\\
&= \{\lambda\in\Sigma_{C}\mid \exists D^{C}_{a}\in\bigcup_{i\in I}\underline{U}_{i}(C)\ s.t.\ \lambda(a)>0\}.
\end{align}
Combining this with the observation
\begin{equation}
\forall_{C\in\mathcal{C}(A)}\ \ \bigcup_{i\in I}\underline{U}_{i}(C)\subseteq\left(\bigvee_{i\in I}\underline{U}_{i}\right)(C),
\end{equation}
we find
\begin{equation}
\bigcup_{i\in I}\Psi(\underline{U}_{i})\subseteq\Psi\left(\bigvee_{i\in I}\underline{U}_{i}\right).
\end{equation}
Now suppose that $\lambda\in\Psi(\bigvee_{i\in I}\underline{U}_{i})_{C}$. Then there exists a $D^{C}_{a}\in(\bigvee_{i\in I}\underline{U}_{i})(C)$ such that $\lambda(a)>0$. The fact that $D^{C}_{a}\in(\bigvee_{i\in I}\underline{U}_{i})(C)$ means that
\begin{equation}
\forall q\in\mathbb{Q}^{+}\ \exists D^{C}_{b_{1}},...,D^{C}_{b_{n}}\in\bigcup_{i\in I}\underline{U}_{i}(C),\ \ \ D^{C}_{(a-q)^{+}}\leq D^{C}_{b},
\end{equation}
where $b=\bigvee_{i=1}^{n}b_{i}$. As before, it follows that $X^{C}_{(a-q)}\subseteq X^{C}_{b}$. Because $\lambda(a)>0$, we can pick a $q'\in\mathbb{Q}^{+}$ small enough such that $\lambda((a-q')^{+})>0$. In short, $\lambda\in X^{C}_{(a-q')}$. As for every $q\in\mathbb{Q}^{+}$ we know $X^{C}_{(a-q)}\subseteq X^{C}_{b}$, we know that $\lambda\in X^{C}_{(a-q')}\subseteq X^{C}_{b}=\bigcup_{i=1}^{n} X^{C}_{b_{i}}$. There is some $j\in\{1,...,n\}$ such that $\lambda\in X^{C}_{b_{j}}$ and $D^{C}_{b_{j}}\in\bigcup_{i\in I}\underline{U}_{i}(C)$. Hence
\begin{equation}
\lambda\in\{\lambda'\in\Sigma_{C} \mid\exists D^{C}_{x}\in\bigcup_{i\in I}\underline{U}_{i}(C)\  \text{s.t.} \ \lambda'(x)>0\}=\left(\bigcup_{i\in I}\Psi(\underline{U}_{i})\right)_{C}.
\end{equation}
This completes the proof that $\Psi$ preserves joins.
\end{proof}

In \cite{spit} Spitters computes the external description of the internal spectrum $\underline{\Sigma}_{\ast}$, but in terms of formal topology (i.e. sites) \cite{acz}, and by a different technique, namely by using iterated forcing.

From now on we will identify the spectrum $\uS_{\uA}$ with the internal locale $\uS_{\ast}$. By Corollary 2.18 it is harmless to do so.

\begin{cor}
The locale associated to the frame $\mathcal{O}\Sigma_{\ast}$ is spatial. 
\end{cor}

Despite the fact that the internal spectrum $\underline{\Sigma}_{\ast}$ may have no global points because of the Kochen-Specker Theorem, its external description has enough points to be spatial. In the next section we will see that $\Sigma_{\ast}$ need not be sober (let alone completely regular), because the locale associated to the frame $\mathcal{O}\Sigma_{\ast}$ may have points that do not correspond to elements of the space $\Sigma_{\ast}$.

\begin{cor}
The internal locale $\underline{\Sigma}_{\ast}$ is compact and completely regular.
\end{cor}

\begin{proof}
We already knew this for general commutative unital C*-algebras in a topos from constructive Gelfand duality, which establishes a duality between unital commutative C*-algebras and compact completely regular locales\footnote{In $\Set$ completely regular locales are equivalent to compact Hausdorff spaces (using the axiom of choice).} \cite{banmul, banmul2, banmul3}. However, Corollary 2.18 presents a way to check compactness and complete regularity directly. Indeed, Lemma 2.6 and Lemma 2.9 applied to the projection 

\noindent $\pi:\Sigma_{\ast}\to\mathcal{C}(A)$ prove the corollary.
\end{proof}

Consider the spectrum $\underline{\Sigma}_{\ast}=\underline{\Sigma}_{\ast}(\underline{A})$ for an $n$-level system $A=M_{n}(\mathbb{C})$. For every $C\in\mathcal{C}(A)$ the Gelfand spectrum $\mathcal{O}\Sigma_{C}$ is isomorphic to $\mathcal{P}(C)$ as a frame, where $\mathcal{P}(C)$ is the set of projection operators in $C$, partially ordered as $p\leq q$ if $p\C^{n}\subseteq q\C^{n}$.  Let $C\subseteq C'$ in $\mathcal{C}(A)$. Take $U_{C}\in\mathcal{O}\Sigma_{C}$ corresponding to the projection operator $P_{C}\in C$ and $U_{C'}\in\mathcal{O}\Sigma_{C'}$ corresponding to the projection operator $P_{C'}\in C'$. We have $\rho^{-1}_{C'C}(U_{C})\subseteq U_{C'}$ if and only if $P_{C'}\geq P_{C}$. This demonstrates that for an $n$-level system there is a bijection
\begin{equation}
\mathcal{O}\Sigma_{\ast}\cong\{S:\mathcal{C}(A)\to\mathcal{P}(A)\mid S(C)\in\mathcal{P}(C),\ C\subseteq C'\Rightarrow S(C)\leq S(C')\}.
\end{equation}
This description in terms of maps $S$ is exactly the externalization of $\mathcal{O}\underline{\Sigma}_{\ast}$ for an $n$-level system given in \cite{chls}. It is a straightforward exercise to verify that the Heyting algebra structure given in \cite{chls} coincides with the Heyting algebra structure of $\mathcal{O}\Sigma_{\ast}$.

Corollary 2.18 gives an explicit description of the internal spectrum of $\underline{A}$, and of the opens of this locale. In the covariant approach, the opens of the spectrum represent propositions about the system under investigation. This might seem a good time to take a closer look at the opens of the spectrum and see if we can understand them physically. However, we will wait until the end of Section 3 before we reflect on the physical interpretation of the covariant approach. The reason is that in Section 3 we use the daseinisation of self-adjoint operators from the contravariant approach to define elementary propositions $[a\in\Delta]$ in the covariant approach (as certain open subsets of $\Sigma_{\ast}$). For the remainder of Section 2 we will continue the investigation of the spaces $\Sigma_{\ast}$ and $\Sigma^{\ast}$.

\subsection{Sobriety of the Quantum State Spaces}

In this subsection we investigate the sobriety of the spaces $\Sigma^{\ast}$ and $\Sigma_{\ast}$. Let $\tilde{\Sigma}^{\ast}$ and $\tilde{\Sigma}_{\ast}$ denote the locales associated to the frames $\mathcal{O}\Sigma^{\ast}$ and $\mathcal{O}\Sigma_{\ast}$ respectively\footnote{We write tildes in order to distinguish the locales from the spaces $\Sigma^{\ast}$ and $\Sigma_{\ast}$}. As pointed out near the end of the previous subsection, the external description of the spectrum $\uS_{\uA}$ has been computed in \cite{spit} in terms of formal topology. In particular, the points of the external description $\uS_{\uA}(\C\cdot 1)$, called `consistent ideals of partial measurement outcomes' there, are given explicitly. The points of the locale $\tilde{\Sigma}_{\ast}$, presented in this subsection (Lemma 2.24), found using the space $\Sigma_{\ast}$, agree with these `consistent ideals'. This must be so, for it follows from Corollary 2.18 that $\tilde{\Sigma}_{\ast}\cong\uS_{\uA}(\C\cdot1)$ as locales.

We start with a small summary of the previous two subsections. In Subsection 2.1 we discussed the locale $\underline{\Sigma}^{\ast}$ internal to $[\mathcal{V}(A)\op,\mathbf{Set}]$. This locale is compact, but generally not regular, and typically has no global points because of the Kochen-Specker Theorem. It is connected to the contravariant approach as follows. By Theorem 2.2, there exists an injective morphism of complete Heyting algebras from the complete Heyting algebra of propositions $\mathcal{O}_{cl}\underline{\Sigma}$ to the externalization $\mathcal{O}\Sigma^{\ast}$ of the internal frame $\mathcal{O}\underline{\Sigma}^{\ast}$ (seen as a complete Heyting algebra). The externalized locale $\Sigma^{\ast}$ is spatial, but in general it is neither compact nor regular. Recall that its frame is given as follows:

\begin{dork}
Let $\Sigma=\{(C,\lambda) \mid C\in\mathcal{V}(A),\lambda\in\Sigma_{C}\}$. Then $U\in\mathcal{O}\Sigma^{\ast}$ iff
\begin{enumerate}
\item $\forall C\in\mathcal{V}(A)\ U_{C}\in\mathcal{O}\Sigma_{C}$.
\item If $\lambda\in U_{C}$ and $C'\subseteq C$, then $\lambda |_{C'}\in U_{C'}$.
\end{enumerate}
\end{dork}

In Subsection 2.2 we discussed the locale $\underline{\Sigma}_{\ast}$ internal to the topos $[\mathcal{C}(A),\mathbf{Set}]$. This locale is compact, regular, and may have no points because of the Kochen-Specker Theorem either. We showed that the the externalization of the (internal) Gelfand spectrum of the Bohrified C*-algebra $\underline{A}$ is isomorphic to $\mathcal{O}\Sigma_{\ast}$, the externalization of $\mathcal{O}\underline{\Sigma}_{\ast}$. This is an isomorphism of frames. The external locale $\Sigma_{\ast}$ is spatial and compact,\footnote{Let $\{U_{i}\}_{i\in I}$ be a cover of $\Sigma_{\ast}$. Then there exists a $j\in I$ such that $(\C\cdot1,\ast)\in U_{j}$. Because $U_{j}$ is open $U_{j}=\Sigma_{\ast}$. The set $\{U_{j}\}$ trivially gives a finite subcover.} but may in general not be regular. We recall its frame for convenience:

\begin{dork}
Let $\Sigma=\{(C,\lambda) \mid C\in\mathcal{V}(A),\lambda\in\Sigma_{C}\}$. Then $U\in\mathcal{O}\Sigma_{\ast}$ iff
\begin{enumerate}
\item $\forall C\in\mathcal{C}(A)\ U_{C}\in\mathcal{O}\Sigma_{C}$.
\item If $\lambda\in U_{C}$ and $C\subseteq C'$, then $\lambda'\in U_{C'}$ whenever $\lambda' |_{C}\in U_{C}$.
\end{enumerate}
\end{dork}

After this recap we start with the investigation of the sobriety of the spaces $\Sigma^{\ast}$ and $\Sigma_{\ast}$. We start with the  space $\Sigma_{\ast}$ of the covariant approach. A point of $\Sigma_{\ast}$ by definition corresponds to a frame map $p:\mathcal{O}\Sigma_{\ast}\to\underline{2}$. If we define $U$ to be the union of all $V\in\mathcal{O}\Sigma_{\ast}$ such that $p(V)=0$, then $U\in\mathcal{O}\Sigma_{\ast}$ is the largest open set mapped to 0 by $p$. This can be translated to the following condition. If there are $U_{1},U_{2}\in\mathcal{O}\Sigma_{\ast}$ such that $U=U_{1}\cap U_{2}$, then either $U_{1}=U$ or $U_{2}=U$. Switching to complements, one can equivalently look at irreducible closed sets. These are sets $F$ that are closed with respect to $\mathcal{O}\Sigma_{\ast}$ such that if there exist closed sets $F_{1}$ and $F_{2}$ with the property $F=F_{1}\cup F_{2}$, then either $F=F_{1}$ or $F=F_{2}$.

\begin{lem}
Let $F$ be closed in $\Sigma_{\ast}$. Then $F$ is irreducible if and only if the following two conditions are satisfied:
\begin{enumerate}
\item $\forall C\in\mathcal{C}(A)$: if $F_{C}\neq\emptyset$, then $F_{C}$ is a singleton.
\item $\forall C_{1},C_{2}\in\mathcal{C}(A)$: if $F_{C_{1}}$ and $F_{C_{2}}$ are both nonempty, then there exists a $C_{3}\in\mathcal{C}(A)$ such that $C_{1},C_{2}\subseteq C_{3}$ and $F_{C_{3}}$ is nonempty.
\end{enumerate}
\end{lem}

\begin{proof}
By definition of $\mathcal{O}\Sigma_{\ast}$, a set $F$ is closed iff the following two conditions are satisfied. First, for every $C\in\mathcal{C}(A)$ the set $F_{C}$ is closed in $\Sigma_{C}$. Second, if $\lambda\in F_{C}$ and $D\subseteq C$, then $\lambda|_{D}\in F_{D}$. 

Conversely, assume that there is a $C\in\mathcal{C}(A)$ such that $F_{C}$ has more than one element. The set $F_{C}$ is reducible in $\Sigma_{C}$, so there are closed $F_{1 C},F_{2 C}\subset F_{C}$ with the property $F_{1 C}\cup F_{2 C}=F_{C}$. Define the sets $F_{i}$, $i=1,2$ as follows. For any $C'\supseteq C$ take $(F_{i})_{C'}=\rho^{-1}_{C'C}(F_{i C})\cap F_{C'}$. For all other $C'\in\mathcal{C}(A)$ take $(F_{i})_{C'}=F_{C'}$. It is easily verified that the sets $F_{i}$ are closed in $\Sigma_{\ast}$, that $F_{i}\subset F$, and that $F_{1}\cup F_{2}=F$. Hence the first condition of the lemma is a necessary condition for irreducibility.

Assume that there are contexts $C_{1},C_{2}\in\mathcal{C}(A)$ such that $F_{C_{1}}$ and $F_{C_{2}}$ are nonempty and that for every $C'\in\mathcal{C}(A)$ with the property $C_{1},C_{2}\subseteq C'$ we have $F_{C'}=\emptyset$. In that case, define $F_{i}$ with $i\in\{1,2\}$, as follows. If $C'\supseteq C_{i}$ then $(F_{i})_{C'}=\emptyset$. For all other $C'\in\mathcal{C}(A)$ take $(F_{i})_{C'}=F_{C'}$. Again this produces closed sets $F_{1},F_{2}\subset F$ such that $F=F_{1}\cup F_{2}$. Thus the second condition in the lemma has also been shown to be necessary. 

Assume that $F$ satisfies both conditions of the lemma. Let $F=F_{1}\cup F_{2}$ and $F\neq F_{2}$. Then there is a $\lambda\in F_{C}$ such that $\lambda\in(F_{1})_{C}$ and $\lambda\notin(F_{2})_{C}$. Pick any $\lambda'\in F_{C'}$. By assumption, there is a context $C''\in\mathcal{C}$ such that $\lambda''\in F_{C''}$ and $C,C' \subseteq C''$. Evidently, $\lambda=\lambda''|_{C}$ and $\lambda'=\lambda''|_{C'}$.  As  $\lambda\notin(F_{2})_{C}$ and $F_{2}$ is closed, we find  $\lambda''\notin(F_{2})_{C''}$. As $F=F_{1}\cup F_{2}$, one has $\lambda''\in(F_{1})_{C''}$. Using that $F_{1}$ is closed, we find that $\lambda'\in(F_{1})_{C'}$. Thus $F\subseteq F_{1}$, proving irreducibility.
\end{proof}

\begin{tut}
Let $\mathcal{C}(A)$ satisfy the following \textbf{ascending chain property}: every chain of contexts
\begin{equation}
C_{1}\subseteq C_{2}\subseteq C_{3}\subseteq ...,
\end{equation}
stabilizes, in the sense that there exists an $n\in\mathbb{N}$ such that for all $m\geq n$ we have $C_{m+1}=C_{m}$. Then the space $\Sigma_{\ast}(A)$ is sober. In particular, if $A$ is finite-dimensional, then $\Sigma_{\ast}(A)$ is sober.
\end{tut}

\begin{proof}
Take any totally ordered subset of $\mathcal{Q}_{B}=\{C\in\mathcal{C}(A) | F_{C}\neq\emptyset\}$, where the order is given by inclusion. Then the ascending chain condition ensures that there is an upper bound. An application of Zorn's Lemma tells us that $\mathcal{Q}_{B}$ has a maximal element. By Lemma 2.24(2), the set $\mathcal{Q}_{B}$ is upwards directed so this maximal element must be unique. If $C$ is this maximal element and $F_{C}=\{\lambda\}$, then we recognize $F$ as the closure of $(C,\lambda)$. For C*-algebras where the ascending chain condition applies, such as $n$-level systems, and assuming the axiom of choice, the points of the locale $\Sigma_{\ast}$ correspond to the points of the topological space $\Sigma_{\ast}$. 
\end{proof}

Next we consider the points of the locale $\tilde{\Sigma}^{\ast}$, associated to the frame $\mathcal{O}\Sigma^{\ast}$. 

\begin{lem}
Let $F$ be an irreducible closed subset of $\Sigma^{\ast}$. Suppose there is a context $C\in\mathcal{V}(A)$ such that for all $D\subset C$ we have $F_{D}=\emptyset$, while $F_{C}\neq\emptyset$. Then there is a unique $\lambda\in\Sigma_{C}$ such that $F$ is the closure of $(C,\lambda)$.
\end{lem}

\begin{proof}
By definition of $\mathcal{O}\Sigma^{\ast}$, a set $F$ is closed iff the following two conditions are satisfied:
\begin{enumerate}
\item For every $C\in\mathcal{V}(A)$ the set $F_{C}$ is closed in $\Sigma_{C}$, 
\item If $\lambda\in F_{C}$, $C\subseteq C'$ and $\lambda'\in\rho_{C'C}^{-1}(\lambda)$ then $\lambda'\in F_{C'}$. 
\end{enumerate}
Define $F_{1}$ as follows: for every $C'$ different from $C$ we take $(F_{1})_{C'}=F_{C'}$, and at the context $C$ we take $(F_{1})_{C}=\emptyset$. It is easily checked that $F_{1}\subset F$ and that $F_{1}$ is closed. Define $F_{2}$ as follows: if $C'\supseteq C$, then $(F_{2})_{C'}=\rho^{-1}_{C'C}(F_{C})$. For all other $C'\in\mathcal{V}(A)$, define $(F_{2})_{C'}=\emptyset$. The set $F_{2}$ is closed and $F=F_{1}\cup F_{2}$. By irreducibility of $F$ it follows that $F=F_{2}$.

Suppose that $F_{C}$ has more than one element. In that case $F_{C}$ is reducible in $\Sigma_{C}$ and we find two proper closed subsets $F_{1C},F_{2C}\subset F$ such that $F_{1C}\cup F_{2C}=F_{C}$. Define the sets $F'_{i}$, for $i=1,2$, as follows. If $C'\supseteq C$, then $(F_{i}')_{C'}=\rho^{-1}_{C'C}(F_{iC})$. For all other $C'\in\mathcal{V}(A)$ take $(F'_{i})_{C'}=\emptyset$. Again, $F'_{i}\subset F$, the $F'_{i}$ are closed, and $F=F'_{1}\cup F'_{2}$. As $F$ is irreducible, $F_{C}$ must be a singleton. If $F_{C}=\{\lambda\}$, then $F$ is clearly the closure of $(C,\lambda)$. 
\end{proof}

\begin{poe}
Let $\mathcal{C}(A)$ satify the following \textbf{descending chain property}: every chain of contexts
\begin{equation}
... \subseteq C_{3}\subseteq C_{2}\subseteq C_{1},
\end{equation}
stabilizes in the sense that there exists an $n\in\mathbb{N}$ such that for all $m\geq n$ we have $C_{m+1}=C_{m}$. Then the space $\Sigma^{\ast}(A)$ is sober. In particular, if $A$ is finite-dimensional, then $\Sigma^{\ast}(A)$ is sober.
\end{poe}

\begin{proof}
Take any totally ordered subset of $\mathcal{Q}_{I}=\{C\in\mathcal{V}(A) | F_{C}\neq\emptyset\}$, where the order is now given by reversed inclusion. Then the descending chain condition ensures that there is an upper bound. An application of Zorn's Lemma tells us that $\mathcal{Q}_{I}$ has a maximal element, which is a minimal context $C$ such that $F_{C}\neq\emptyset$. It follows from Lemma 2.26 that this minimal context must be unique. For C*-algebras where the descending chain condition applies, such as $n$-level systems, and assuming the axiom of choice, the points of the locale $\Sigma^{\ast}$ correspond to the points of the topological space $\Sigma^{\ast}$. 
\end{proof}

\subsection{Gelfand Transform}

In this subsection we will only consider the covariant approach, as internal Gelfand duality seems irrelevant in the contravariant approach because the locale $\uS^{\ast}$ is not regular and hence cannot arise as the Gelfand spectrum of any commutative C*-algebra in $[\mathcal{V}(A)\op,\mathbf{Set}]$. The goal of this subsection is the explicit computation of the externalized Gelfand transform of $\uA$ (given by (62)-(63)). 

By constructive Gelfand duality, the internal commutative C*-algebra $\underline{A}$ with internal spectrum $\underline{\Sigma}_{\ast}$ is isomorphic to the internal commutative C*-algebra of continuous maps $\underline{C}(\underline{\Sigma}_{\ast},\uC)$ (which is the object of frame maps $\mathcal{O}\underline{\C}\to\mathcal{O}\uSs$). Here $\uC$ denotes the internal locale of complex numbers, given explicitly by the external description $\pi_{1}:\mathcal{C}(A)\times\C\to\mathcal{C}(A)$ (e.g.\cite{banmul3}). Let $\uA_{sa}$ be the self-adjoint part of $\uA$, defined by the functor $\uA_{sa}(C)=C_{sa}$. Then $\uA_{sa}$ is naturally isomorphic to the object $\underline{C}(\underline{\Sigma}_{\ast},\underline{\R})$, where $\underline{\R}$ is the internal locale of real numbers. The object $\underline{C}(\uSs,\underline{\R})$ is the object of internal frame maps $\underline{\text{Frm}}(\mathcal{O}\underline{\R},\mathcal{O}\uSs)$. For $C\in\mathcal{C}(A)$ we have
\begin{equation}
\underline{\text{Frm}}(\mathcal{O}\underline{\mathbb{R}},\mathcal{O}\uS_{\ast})(C)=\text{Nat}_{\text{Frm}}(\mathcal{O}\underline{\mathbb{R}}|_{\uparrow C},\mathcal{O}\uS_{\ast}|_{\uparrow C}).
\end{equation}
The external description of $\mathcal{O}\underline{\mathbb{R}}|_{\uparrow C}$ is the frame map
\begin{equation}
\pi^{-1}_{\mathbb{R}}:\mathcal{O}(\uparrow C)\to\mathcal{O}(\uparrow C\times\mathbb{R}),
\end{equation}
which is the inverse image of the continuous map $\pi_{\mathbb{R}}:(\uparrow C)\times\mathbb{R}\to(\uparrow C)$, the projection on the first coordinate. Here $(\uparrow C)$ has the Alexandrov topology and $(\uparrow C)\times\mathbb{R}$ carries the product topology. In \cite[Section 5]{chls} the right hand side of (52) is shown to be equal to the set of frame maps
\begin{equation}
\phi^{\ast}_{C}:\mathcal{O}(\uparrow C\times\R)\to\mathcal{O}\Sigma_{\ast}|_{\uparrow C}
\end{equation}
that satisfy the property that for every $C'\supseteq C$,
\begin{equation}
\phi^{\ast}_{C}(\uparrow C'\times\R)=\Ss|_{\uparrow C'}=\coprod_{C''\in\uparrow C'}\Sigma_{C''}.
\end{equation}
We denote the set of frame maps satisfying this property by
\begin{equation}
\text{Frm}'(\mathcal{O}(\uparrow C\times\R),\mathcal{O}\Sigma_{\ast}|_{\uparrow C}).
\end{equation}
Under the identification of (52) with (54), the Gelfand transformation becomes the natural isomorphism
\begin{equation}
\tilde{\mathcal{G}}:\uA_{sa}\stackrel{\cong}{\longrightarrow}\text{Frm}'(\mathcal{O}(\uparrow-\times\R),\mathcal{O}\Sigma_{\ast}|_{\uparrow-}),
\end{equation}
defined by
\begin{equation}
\hat{a}_{C}^{-1}:=\tilde{\mathcal{G}}_{C}(a):\mathcal{O}(\uparrow C\times\R)\to\mathcal{O}\Ss|_{\uparrow C},
\end{equation}
\begin{align}
\hat{a}_{C}^{-1}(\uparrow C'\times(p,q)) &=\{(C'',\lambda'')\mid C''\in\uparrow C',\ \lambda''(a)\in(p,q)\}\\
&= \coprod_{C''\in\uparrow C'}(\hat{a}^{(C'')})^{-1}(p,q),
\end{align}
where $a\in C_{sa}$, and $\hat{a}^{(C'')}$ denotes the (classical) Gelfand transform of $a$, seen as element of $C''\supseteq C$. This frame map is the inverse image of the continuous map
\begin{equation}
\hat{a}_{C}:\Ss|_{\uparrow C}\to(\uparrow C\times\R),\ \ (C',\lambda')\mapsto(C', \lambda'(a)).
\end{equation}
Note that continuous maps $f:\Ss|_{\uparrow C}\to(\uparrow C\times\R)$ such that $\pi_{1}\circ f=\pi$ correspond bijectively to continuous maps $f:\Ss|_{\uparrow C}\to\R$. The Gelfand isomorphism $\tilde{\mathcal{G}}$ thus induces the natural isomorphism
\begin{equation}
\mathcal{G}:\uA_{sa}\stackrel{\cong}{\longrightarrow} C(\Sigma_{\ast}|_{\uparrow-},\R),
\end{equation}
\begin{equation}
\mathcal{G}_{C}(a)=\hat{a}_{C}:\Sigma_{\ast}|_{\uparrow C}\to\R,\ \ \hat{a}_{C}(C',\lambda')=\lambda'(a).
\end{equation}
This may look surprising at first glance, but in fact a continuous map $f:\Ss|_{\uparrow C}\to\R$ is determined by $f|_{\Sigma_{C}}$. This is because continuity implies that $f(C',\lambda')=f(C,\lambda'|_{C})$, giving a bijection $C(\Sigma_{C},\R)\simeq C(\Ss|_{\uparrow C},\R)$. Next, note that by (61), $\hat{a}_{C}|_{\Sigma_{C'}}=\hat{a}^{(C')}$. If we are using (classical) Gelfand duality to identify $C\simeq C(\Sigma_{C},\R)$ and subsequently identify  $C(\Sigma_{C},\R)\simeq C(\Ss|_{\uparrow C},\R)$, we recover (63). We conclude that the internal Gelfand transformation of $\underline{A}_{sa}$, looked upon externally, combines the Gelfand transformations of all the contexts into a single presheaf. This was already pointed out in \cite{hlsw}.

\se{Elementary Propositions and Daseinisation}

In this section we investigate elementary propositions and daseinisation in both the contravariant and the covariant approaches. Daseinisation plays an important role in the contravariant approach in at least two ways \cite{di2}. Firstly, it is used to define elementary propositions. These are propositions of the form $[a\in\Delta]$ with $a\in A_{sa}$ and $\Delta\in\mathcal{O}\R$. Secondly, daseinisation in a more advanced form is used to define an arrow $\breve{\delta}(a):\uS\to\underline{\R}^{\leftrightarrow}$, for every $a\in A_{sa}$, where $\underline{\R}^{\leftrightarrow}$ is called the \textbf{value object} of the topos $[\mathcal{V}(A)\op,\mathbf{Set}]$.

In the covariant approach there is a daseinisation arrow, too. For each $a\in A_{sa}$ this is an arrow $\underline{\delta}(a):\uSs\to\uIR$, where $\uIR$ is the interval domain internal to the topos $[\mathcal{C}(A),\mathbf{Set}]$. This daseinisation arrow is an internal locale map. The original covariant daseinisation arrow of \cite{hls} did not use the daseinisation techniques of the contravariant approach \cite{di2}, but it can be greatly simplified by a minor modification that does (see Subsection 3.2). Subsequently, any $\Delta\in\mathcal{O}\R$ defines a point $\underline{\Delta}:\underline{1}\to\mathcal{O}\uIR$. Combining this with the daseinisation arrow $\underline{\delta}(a)^{-1}$ produces the covariant  version of elementary propositions 

\noindent $[a\in\Delta]:\underline{1}\to\mathcal{O}\uSs$. 

In Subsection 3.1 we look at daseinisation of self-adjoint operators as originally defined in the contravariant approach. The elementary propositions of the contravariant approach are also introduced. In Subsection 3.2 we look at the covariant daseinisation arrow. After defining this arrow, we adapt it in order to apply the daseinisation techniques of Subsection 3.1 to the covariant setting. This leads to an explicit combination of the daseinisation arrow and elementary propositions. In Subsection 3.3 we study the contravariant daseinisation arrow $\breve{\delta}(a):\uS\to\underline{\R}^{\leftrightarrow}$, and in particular we compare it with the adapted version of the covariant daseinisation arrow. Subsection 3.4 discusses the relation between the so-called antonymous and observable functions, and the covariant daseinisation arrow. Finally, in Subsection 3.5 we consider the physical interpretation of daseinisation.

\su{Contravariant Approach}

We start with elementary propositions and daseinisation of selfadjoint operators in the contravariant approach. The reader familiar with daseinisation can skip this subsection, as it contains no new material. An extensive discussion of daseinisation can be found in the paper \cite{doe} by D\"oring. First we deal with outer daseinisation of projection operators, as we need these to define elementary propositions. In order to motivate outer daseinisation, let $a\in A_{sa}$ and $\Delta\in\mathcal{O}\R$. In quantum logic \`a la von Neumann, the elementary proposition ``$a\in\Delta$" is represented by a projection operator $p=\chi_{\Delta}(a)$, where $\chi_{\Delta}$ is defined by functional calculus (or, equivalently, by the Spectral Theorem for Borel functions). A proposition in the contravariant approach is a closed open subobject of the spectral presheaf $\underline{S}\rightarrowtail\uS$. Therefore, for every context $C\in\mathcal{V}(A)$ we want to associate a closed open subset $\underline{S}(C)$ of the spectrum $\Sigma_{C}$ in such a way that these choices combine to give a presheaf. If $p\in C$, then the natural choice would be
\begin{equation}
\underline{S}(C)=\{\lambda\in\Sigma_{C}\mid\lambda(p)=1\},
\end{equation}
but what about the other contexts? Let $C\in\mathcal{V}(A)$ be any context. Following \cite{di2}, we approximate the projection operator $p$ using the projection operators available in $C$ as follows:
\begin{equation}
\daspo=\bigwedge\{q\in\mathcal{P}(C)\mid q\geq p\},
\end{equation}
where $\mathcal{P}(C)$ is the lattice of projections in $C$. Hence $\daspo$ is the smallest projection operator $C$ that is larger than $p$. Note that if $p\in C$, then $\daspo=p$. Also note that $\daspo$ must be an element of $C$, since the projections in a von Neumann algebra form a \textit{complete} lattice \cite{kari}.\footnote{If the context $C$ is a commutative unital C*-algebra, then it could very well be that  $\daspo\notin C$, but for abelian von Neumann algebras or the larger class of commutative AW*-algebras the daseinisation operation works.} Next, define
\begin{equation}
\underline{\delta}^{o}(p)(C)=\{\lambda\in\Sigma_{C}\mid\lambda(\daspo)=1\}.
\end{equation}
This is a closed open subset of $\Sigma_{C}$, because the Gelfand transform of $\delta^{o}(p)_{C}$ is a continuous function on $\Sigma_{C}$. Noting that for $C\subseteq C'$ we have $\daspo\geq\delta^{o}(p)_{C'}$, it is easy to check that $\underline{\delta}^{o}(p)$ defines a closed open subobject of the spectral presheaf. The elementary proposition $\underline{[a\in\Delta]}\rightarrowtail\underline{\Sigma}$ is defined as
\begin{equation}
\underline{[a\in\Delta]}(C):=\underline{\delta}^{o}(\chi_{\Delta}(a))(C)=\{\lambda\in\Sigma_{C}\mid\lambda(\delta^{o}(\chi_{\Delta}(a))_{C})=1\},
\end{equation}
where $\chi_{\Delta}(a)$ denotes the spectral projection operator associated to ``$a\in\Delta$''. Note that because for $C\subseteq C'$ we have $\delta^{o}(p)_{C}\geq\delta^{o}(p)_{C'}$, the definition of elementary propositions fits very well with the coarse-graining philosophy.

In addition to the daseinisation of projection operators given in (65), which we will call \textbf{outer daseinisation}, we will also consider \textbf{inner daseinisation}. Inner daseinisation approximates a projection operator $p$ by taking, in each context, the largest projection operator in $C$ that is smaller than $p$. In other words:
\begin{equation}
\daspi=\bigvee\{q\in\mathcal{P}(C)\mid q\leq p\}.
\end{equation}
Note that if $p\in C$, we have $\daspi=p$ and that if $C\subseteq C'$, then $\daspi\leq\delta^{i}(p)_{C'}$. Inner daseinisation does not yield propositions in the same way as outer daseinisation, but it remains an important construction. For example, it is needed for the definition of the outer daseinisation of self-adjoint operators and it is important in defining the daseinisation arrow $\breve{\underline{\delta}}(a):\uS\to\underline{\R}^{\leftrightarrow}$, which will be discussed in Subsection 3.3. 

Next, we turn our attention to daseinisation of self-adjoint operators. By the spectral theorem \cite{kari}, every self-adjoint element $a\in A$ has a spectral resolution $\{e^{a}_{\lambda}\}_{\lambda\in\R}$, where $e^{a}_{\lambda}=\chi_{(-\infty,\lambda]}(a)$. The daseinisation of a self-adjoint operator proceeds by daseinisation of  its spectral resolution, as we will see in a moment.

Thus far, we only made use of the partial order $\leq$ on self-adjoint operators, where $a\leq b$ means that $b-a$ is a positive operator. In what follows, we will in addition use a different partial order on $A_{sa}$, which was first considered in \cite{ols}. Let $a,b\in A_{sa}$, with spectral resolutions $\{e^{a}_{\lambda}\}$ and $\{e^{b}_{\lambda'}\}$. Then $a$ is below $b$ in the spectral order, denoted $a\leq_{s}b$, if for every $\lambda\in\R$ we have $e^{a}_{\lambda}\geq e^{b}_{\lambda}$.\footnote{Equivalently, for positive operators $a$ and $b$, $a\leq_{s} b$ iff $\forall_{n\in\N}\ a^{n}\leq b^{n}$ \cite[Theorem 3]{ols}.} The spectral order is coarser than the linear order in the sense that $a\leq_{s} b$ implies $a\leq b$, while the converse need not hold in general. However, let $p$ be a projection operator in $A$. Then the spectral resolution of $p$ is given by
\begin{equation*}
e^{p}_{\lambda}=\left\{
\begin{array}{rl}
0 & \text{if} \ \lambda\in (-\infty,0);\\
1-p & \text{if}\  \lambda\in [0,1);\\
1 & \text{if } \lambda\in [1,\infty). 
\end{array} 
\right. 
\end{equation*} 
From this it follows that if $p$ and $q$ are projections in $A$, then $p\leq_{s} q$ iff $p \leq q$. Also, if $a,b\in A_{sa}$ such that $[a,b]=0$, then similarly $a\leq_{s} b$ iff $a\leq b$. So in every context $C\in\mathcal{V}(A)$ the spectral order $\leq_{s}$ reduces to the usual order $\leq$. The proof of this last claim and more information on the spectral order can be found in \cite{gro1}. 

\begin{dork}
Let $a\in A_{sa}$. Define the outer and inner daseinisations of $a$ at context $C\in\mathcal{V}(A)$ by, respectively,
\begin{equation}
\dasao=\bigwedge\ \{b\in C_{sa}\mid b\geq_{s}a\},
\end{equation}
\begin{equation}
\dasai=\bigvee\ \{b\in C_{sa}\mid b\leq_{s}a\}.
\end{equation}
\end{dork}

The self-adjoint operators $\dasao$ and $\dasai$ are elements of $C$ because $C_{sa}$ is a boundedly complete lattice with respect to the spectral order. The daseinisation of self-adjoint operators can be described by the daseinisation of the projections in their spectral resolution. Let $\lambda\mapsto e_{\lambda}$ be the spectral resolution of a self-adjoint bounded operator $a$. Then
\begin{align}
\lambda & \mapsto\bigwedge_{\mu>\lambda}\delta^{o}(e_{\mu})_{C},\\
\lambda & \mapsto\delta^{i}(e_{\lambda})_{C},
\end{align}
are also spectral resolutions of self-adjoint bounded operators \cite{di, gro1}.

\begin{lem}
Let $a\in A_{sa}$. Then the spectral resolutions of the outer and inner daseinisations of $a$ at context $C$ are
\begin{align}
\delta^{o}(a)_{C} &=\int\lambda d(\delta^{i}(e^{a}_{\lambda})_{C});\\
\delta^{i}(a)_{C} &=\int\lambda d(\bigwedge_{\mu>\lambda} \delta^{o}(e^{a}_{\mu})_{C}).
\end{align}
\end{lem}

Note that the \emph{outer} daseinisation of $a$ uses the \emph{inner} daseinisation of the spectral resolution $\lambda\mapsto e^{a}_{\lambda}$, and vice versa. It also follows from the definition that for any $D,C\in\mathcal{V}(A)$ with $D\subseteq C$ we have 
\begin{equation}
\delta^{i}(a)_{D}\leq_{s}\dasai\leq_{s}a\leq_{s}\dasao\leq_{s}\delta^{o}(a)_{D}.
\end{equation}
If $a\in C$, then $a=\dasai=\dasao$. Let $p$ be a projection operator. Then the outer daseinisation of $p$ as a self-adjoint operator, as in Definition 3.1, coincides with the outer daseinisation of $p$ as a projection, as in (65). For inner daseinisation we have a similar situation.

For a projection operator $p$, (75) implies that if we move from a context $C$ to a coarser context $D$ then outer daseinisation approximates $p$ by a larger projection operator in the coarser context $D$. Hence a coarser context means a weaker proposition, fitting well with the idea of coarse-graining. For inner daseinisation, moving to a coarser context amounts to taking a smaller projection operator. This does not seem to fit with the idea of coarse-graining. 

We can consider a different view that does fit with coarse-graining and involves both inner and outer daseinisation. By Gelfand duality, for any $a\in A_{sa}$, we can see $\delta^{i}(a)_{C}$ and $\delta^{o}(a)_{C}$ as real-valued continuous functions on the spectrum $\Sigma_{C}$. Given a local state $\lambda\in\Sigma_{C}$, we cannot assign a sharp value of $a$ to that state (except in the special case $a\in C_{sa}$). However, we can assign the closed interval $[\lambda(\delta^{i}(a)_{C}),\lambda(\delta^{o}(a)_{C})]\subset\mathbb{R}$ to $a$ and state $\lambda$. If we restrict the state to a coarser context $D$, then (75) tells us that we associate a larger interval  $[\lambda|_{D}(\delta^{i}(a)_{D}),\lambda|_{D}(\delta^{o}(a)_{D})]$ to $a$. As contexts become coarser, the associated values become less sharp. At a heuristic level this two-sided daseinisation fits with coarse-graining. We will return to the use of two-sided daseinisation in the contravariant approach in Section 3.3.

The reader might wonder why the spectral order $\leq_{s}$ is used, instead of the natural order $\leq$. For example, why not define an inner daseinisation by
\begin{equation}
\delta_{i}(a)_{C}=\bigvee\ \{b\in C_{sa}\mid b\leq a\}.
\end{equation}
This supremum $\delta_{i}(a)_{C}$ exists and is an element of $C$, because the spectral order and the order $\leq$ coincide on $C$. However, $\delta_{i}(a)_{C}\leq a$ may not hold as is shown in the following example using $A=M_{2}(\C)$. Define
\begin{equation*} 
a =\left( 
\begin{matrix} 
0 & 1 \\ 
1 & 1  
\end{matrix} \right),\ \ 
b_{1}=\left( 
\begin{array}{cc}
-1 & 0 \\
0 & 0
\end{array}\right),\ \ 
b_{2}=\left( 
\begin{array}{cc}
-1/4 & 0 \\
0 & -3
\end{array}\right).
\end{equation*}
For any $\underline{v}=(v_{1},v_{2})^{t}\in\C^{2}$ it is easily seen that
\begin{equation}
(\underline{v},(a-b_{1})\underline{v})\geq(|v_{1}|-|v_{2}|)^{2}\geq0,
\end{equation}
\begin{equation}
(\underline{v},(a-b_{2})\underline{v})\geq(1/4|v_{1}|-4|v_{2}|)^{2}\geq0.
\end{equation}
We find $b_{1},b_{2}\leq a$. But $b_{1}\vee b_{2}\nleq a$, which follows from
\begin{equation*}
b_{1}\vee b_{2}=\left(
\begin{array}{cc}
-1/4 & 0 \\
0 & 0
\end{array}\right),\ \ 
\underline{w}=\left(
\begin{matrix}
-i \\ i \end{matrix} \right),
\end{equation*}
\begin{equation}
(\underline{w},(a-b_{1}\vee b_{2})\underline{w})=-3/4.
\end{equation}
It is because of the spectral order that the daseinisation of an operator can be compared with the operator itself, as in (75).

\subsection{Covariant Approach}

In this subsection we investigate the covariant version of the daseinisation map. The original daseinisation arrow of the covariant approach was first introduced in \cite{hls}, where all the details of its construction can be found. We will from the start present a different definition of the daseinisation arrow, which we regard as an improvement or at least as a simplification of the original one. Subsequently we recall the original definition \cite{hls} and compare it with this new definition.

Before we can define the daseinisation arrow, a discussion of Scott's interval domain is in order \cite{abju}. As a set, the interval domain $\IR$ consists of all compact $[a,b]$ with $a,b\in\R$, $a\leq b$. This includes the singletons $[a,a]=\{a\}$. The elements of $\IR$ are ordered by reverse inclusion.\footnote{We might think of elements of $\IR$ as approximations of real numbers (this idea goes back to L.E.J. Brouwer). A smaller set provides more information about the real number it approximates than a larger interval. The smaller interval is higher in the information order.} The interval domain is equipped with the so-called Scott topology. A set $U\subseteq\IR$ is Scott closed if it satisfies the following two conditions. Firstly, it is downward closed in the sense that if $[a,b]\in U$ and $[a,b]\subseteq[a',b']$ then $[a',b']\in U$. Secondly, it is closed under suprema of directed subsets. The collection
\begin{equation}
(p,q)_{S}:=\{[r,s]\mid p<r\leq s<q\},\ \ p,q\in\Q,\ \ p<q,
\end{equation}
defines a basis for the Scott topology $\mathcal{O}\IR$. We will also need the interval domain $\uIR$, internal to $[\mathcal{C}(A),\mathbf{Set}]$. This is an internal locale, whose associated frame $\mathcal{O}\uIR$ has external description
\begin{equation}
\pi_{1}^{-1}:\mathcal{O}(\mathcal{C}(A))\to\mathcal{O}(\mathcal{C}(A)\times\IR),
\end{equation}
where $\mathcal{O}(\mathcal{C}(A))$ denotes the Alexandrov topology, and where $\pi_{1}^{-1}$ is the inverse image of the continuous projection
\begin{equation}
\pi_{1}:\mathcal{C}(A)\times\IR\to\mathcal{C}(A),\ \ \ (C,[a,b])\mapsto C.
\end{equation}

Next, we would like to use the daseinisation of self-adjoint operators introduced in Subsection 3.1, but we are immediately faced with a problem: these constructions do not work for arbitrary C*-algebras, because these generally do not have enough projections. For the remainder of this subsection, also in the covariant approach we will therefore use the context category $\mathcal{V}(A)$ of abelian von Neumann subalgebras, and work in the topos $[\mathcal{V}(A),\mathbf{Set}]$ of covariant functors. Von Neumann algebras have the advantage that a covariant daseinisation arrow can be given explicitly, in terms of the daseinisation maps of the contravariant approach. This makes it easier to compare the two topos approaches.

Without further ado we now define the covariant daseinisation map.

\begin{dork}
The \textbf{covariant daseinisation map} is the function
\begin{equation}
\delta:A_{sa}\to C(\Sigma_{\ast},\IR),\ \ \delta(a):(C,\lambda)\mapsto[\lambda(\delta^{i}(a)_{C}),\lambda(\delta^{o}(a)_{C})].
\end{equation}
Define in $[\mathcal{V}(A),\Set]$, the arrow $\underline{\delta}(a)^{-1}:\mathcal{O}\underline{\IR}\to\mathcal{O}\uS_{\ast}$ by
\begin{equation}
\underline{\delta}(a)^{-1}_{C}(\uparrow C'\times (p,q)_{S})=\delta(a)^{-1}(p,q)_{S}\cap\Sigma_{\ast}|_{\uparrow C'},
\end{equation}
where $\Sigma_{\ast}|_{\uparrow C'}=\coprod_{C''\in\uparrow C'}\Sigma_{C''}$, and $(\uparrow C')\times(p,q)$ denotes the basic open subset $\{(C,[r,s])\mid C\in\mathcal{C}(A), C\supseteq C', p<r\leq s<q\}$ of $\mathcal{C}(A)\times\IR$.
\end{dork}

The map $\delta$ is well-defined, which requires some checking.

\begin{poe}
For each $a\in A_{sa}$, the map $\delta(a):\Sigma_{\ast}\to\IR$ is continuous. Furthermore, $\underline{\delta}(a)^{-1}$ is a frame map in $[\mathcal{V}(A),\mathbf{Set}]$, and thus defines a locale map $\underline{\delta}(a):\uSs\to\uIR$.
\end{poe}

\begin{proof}
In order to prove continuity, note that
\begin{align}
(\delta(a)^{-1}(p,q)_{S})_{C}&=\{\lambda\in\Sigma_{C}\mid\lambda(\delta^{i}(a)_{C})>p\}\cap\{\lambda\in\Sigma_{C}\mid\lambda(\delta^{o}(a)_{C})<q\}\\
&= X^{C}_{\delta^{i}(a)_{C}-p}\cap X^{C}_{q-\delta^{o}(a)_{C}}=X^{C}_{(\delta^{i}(a)_{C}-p)\wedge(q-\delta^{o}(a)_{C})}.
\end{align}
Therefore, $\delta(a)^{-1}(p,q)_{S}$ satisfies the first condition for opens of $\Sigma_{\ast}$ given in Definition 2.23. The second condition follows from (75). 

The map $\delta(a)$ defines an internal locale map $\underline{\delta}(a)$, with external description simply given by the commutative triangle of continuous maps
$$\xymatrix{
\Sigma_{\ast} \ar[d]_{\pi} \ar[rr]^{\langle\pi,\delta(a)\rangle} & & \mathcal{V}(A)\times\IR \ar[dll]^{\pi_{1}} & & (C,\lambda) \ar@{|->}[r] \ar@{|->}[d] & (C,\delta(a)(C,\lambda)) \ar@{|->}[dl]\\
\mathcal{V}(A) & & & & C &}$$
\end{proof}

We may use the daseinisation map $\delta:A_{sa}\to C(\Ss,\IR)$ to define elementary propositions. 
\begin{dork}
Let $a\in A_{sa}$ and $(p,q)\in\mathcal{O}\R$. Then the \textbf{covariant elementary proposition} $[a\in(p,q)]\in\mathcal{O}\Sigma_{\ast}$ is defined by
\begin{align}
[a\in(p,q)] &=\delta(a)^{-1}(p,q)_{S}\\
&=\coprod_{C\in\mathcal{V}(A)}\{\lambda\in\Sigma_{C}\mid[\lambda(\delta^{i}(a)_{C}),\lambda(\delta^{o}(a)_{C})]\in(p,q)_{S}\}.
\end{align}
 Each elementary proposition $[a\in(p,q)]$ defines an open of the spectrum of $\uA$ by
 \begin{equation}
 \underline{[a\in(p,q)]}:\underline{1}\to\mathcal{O}\uS_{\ast},\  \underline{[a\in(p,q)]}_{C}(\ast)=\coprod_{C'\in\uparrow C}[a\in(p,q)]_{C'}.
 \end{equation}
 \end{dork}

If we define
\begin{equation}
\underline{(p,q)}:\underline{1}\to\mathcal{O}\uIR,\ \ \ \underline{(p,q)}_{C}(\ast)=\uparrow C\times(p,q)_{S},
\end{equation}
then
\begin{equation}
\underline{[a\in(p,q)]}=\underline{\delta}(a)^{-1}\circ\underline{(p,q)}:\underline{1}\to\mathcal{O}\uSs.
\end{equation}
 
Compare the covariant elementary proposition (91) with the contravariant elementary propostions which, under the identification of $\mathcal{O}_{cl}\uS$ as a subframe of $\mathcal{O}\Sigma^{\ast}$, are given by
\begin{equation}
[a\in(p,q)]=\coprod_{C\in\mathcal{V}(A)}\{\lambda\in\Sigma_{C}\mid\lambda(\delta^{o}(\chi_{(p,q)}(a))_{C})=1\}.
\end{equation}
This clearly differs from the covariant version. In the contravariant approach, which is motivated by coarse-graining, the spectral projection associated to $``a\in(p,q)"$ as a whole is approximated, whereas in the covariant approach the operator $a$ itself is approximated, which in turn implies the formula (87) for $[a\in(p,q)]$. The covariant approach uses both inner and outer daseinisation, whereas the contravariant approach only uses outer daseinisation. We could have chosen to define covariant elementary propositions in a different way such that it closely mirrors the contravariant version. Consider
\begin{equation}
[a\in(p,q)]=\coprod_{C\in\mathcal{V}(A)}\{\lambda\in\Sigma_{C}\mid\lambda(\delta^{i}(\chi_{(p,q)}(a))_{C})=1\}.
\end{equation}
This subset of $\Sigma_{\ast}$ is open because it is equal to $\delta(\chi_{(p,q)}(a))^{-1}(1-\epsilon,1+\epsilon)$ (for any positive $\epsilon$ smaller than 1). In Subsection 3.3 we will investigate if the covariant elementary proposition of Definition 3.5 has a natural counterpart in the contravariant approach. Using the material in Subsection 3.4, we will see in Subsection 3.5 how the covariant elementary proposition of Definition 3.5 is related to (93). In Subsection 3.5 we will also consider the physical interpretation of the covariant elementary propositions.

Next, we compare the covariant daseinisation arrow with the Gelfand transform $\mathcal{G}$ of Subsection 2.4. Let $i:\R\to\IR$, $x\mapsto[x,x]$, be the inclusion map, and let $\delta(a)|_{\uparrow C(a)}$ denote the restriction of $\delta(a):\Sigma_{\ast}\to\IR$ to $\Sigma_{\ast}|_{\uparrow C(a)}$, where $C(a)$ is the context generated by $a$. Then we have the following commutative triangle
$$\xymatrix{
\Sigma_{\ast}|_{\uparrow C(a)} \ar[rr]^{\delta(a)|_{\uparrow C(a)}} \ar[drr]_{\mathcal{G}_{C(a)}} & & \IR \\
& & \R\ \ . \ar@{^{(}->}[u]_{i} } $$
Hence on the open $\Sigma_{\ast}|_{\uparrow C(a)}\in\mathcal{O}\Sigma_{\ast}$ the daseinisation of $a$ coincides with the Gelfand transform of $a$, formulated as a locale map.

Next, we show how the new covariant daseinisation arrow of Definition 3.3 is related to the original covariant daseinisation arrow of \cite{hls}. The covariant daseinisation map of \cite{hls} is a function $\delta:A_{sa}\to C(\uS_{\uA},\uIR)$, which for every $a\in A_{sa}$, gives  a locale map $\underline{\delta}(a):\uS_{\uA}\to\uIR$ internal to $[\mathcal{C}(A),\mathbf{Set}]$. The inverse image of this daseinisation map, i.e. $\underline{\delta}(a)^{-1}:\mathcal{O}\uIR\to\mathcal{O}\uS_{\uA}$, is given by the frame maps 
\begin{equation}
\underline{\delta}(a)^{-1}_{C}:\mathcal{O}(\uparrow C\times\IR)\to\mathcal{O}\uS_{\uA}(C),
\end{equation}
\begin{eqnarray}
\underline{\delta}(a)^{-1}_{C}(\uparrow C'\times (p,q)_{S}):C''\to \left\{
\begin{array}{rl} 
\emptyset & \text{if } C''\nsupseteq C'\\
Y_{C''}(p,q,a) & \text{if} \ C''\supseteq C'.
\end{array} \right.
\end{eqnarray}
In (95), $Y_{C''}(p,q,a)\subseteq L_{C''}$ is defined as $D^{C''}_{b}\in Y_{C''}(p,q,a)$ iff
\begin{equation}
D^{C''}_{b}\drie_{C''}\{D^{C''}_{(a_{0}-r)\wedge(s-a_{1})}\mid a_{0},a_{1}\in C''_{sa}, \ a_{0}\leq a\leq a_{1},\ [r,s]\in(p,q)_{S}\}.
\end{equation}

Recall that the covering relation $\drie_{C}$ was introduced in Lemma 2.15. In order to relate $\underline{\delta}(a)$ to the daseinisation arrow of Definition 3.3, we replace the context category $\mathcal{C}(A)$ by the category $\mathcal{V}(A)$, and replace $a_{0}\leq a\leq a_{1}$ in (96) by $a_{0}\leq_{s} a\leq_{s} a_{1}$, where $\leq_{s}$ is the spectral order. 
\begin{lem}
Define $\omega=(\dasai-p)\wedge(q-\dasao)$, where $\dasai$ and $\dasao$ are the daseinisations of $a$, as in Subsection 3.1. Then $D^{C}_{b}\drie_{C}\{D^{C}_{\omega}\}$ iff
\begin{equation}
D^{C}_{b}\in \{D^{C}_{(a_{0}-r)\wedge(s-a_{1})}\mid a_{0},a_{1}\in C_{sa}, \ a_{0}\leq_{s}a\leq_{s}a_{1},\ [r,s]\in(p,q)_{S}\}.
\end{equation}
\end{lem}

\begin{proof}
Call the set in (97) $X$ for convenience. If $a_{0}\leq_{s}a$, then by definition $a_{0}\leq_{s}\dasai$, and if $a\leq_{s}a_{1}$ then $\dasao\leq_{s}a_{1}$. If $[r,s]\in(p,q)_{S}$ then
$D^{C}_{(a_{0}-r)\wedge(s-a_{1})}\leq D^{C}_{\omega}$. This proves that $X\drie_{C}\{D^{C}_{\omega}\}$. In order to prove that $D^{C}_{\omega}\drie_{C} X$, it suffices to show that for $\epsilon\in\Q^{+}$ small enough, there is a $[r,s]\in(p,q)$ such that 
\begin{equation}
D^{C}_{\omega-\epsilon}\leq D^{C}_{(\dasai-r)\wedge(s-\dasao)}. 
\end{equation}
Take $\epsilon$ such that $p+\epsilon<q-\epsilon$. Taking $r=p+\epsilon$ and $s=q-\epsilon$ gives the desired inequality.
\end{proof}

This lemma may be used to simplify the covariant daseinisation map of \cite{hls}. For $C\subseteq C'\subseteq C''$ we have
\begin{equation}
D^{C''}_{b}\in\underline{\delta}(a)^{-1}_{C}(\uparrow C'\times(p,q)_{S})(C'')\ \text{iff} \ D^{C''}_{b}\drie_{C}\{D^{C''}_{(\delta^{i}(a)_{C''}-p)\wedge(q-\delta^{o}(a)_{C''})}\}.
\end{equation}
Identifying $\uS_{\uA}$ with $\uSs$ by Corollary 2.18,  $\underline{\delta}(a)^{-1}_{C}(\uparrow C'\times(p,q)_{S})$ corresponds to the following open set of $\mathcal{O}\Ss|_{\uparrow C}$:
\begin{equation}
\{(C'',\lambda'')\in\Sigma\mid C''\supseteq C',\ \lambda''(\delta^{i}(a)_{C''})>p,\ \lambda''(\delta^{o}(a)_{C''})<q\}.
\end{equation}
As $\delta^{i}(a)_{C''}\leq_{s}\delta^{o}(a)_{C''}$ and the spectral order is coarser than the order $\leq$, the set in (100) is equal to
\begin{equation}
\{(C'',\lambda'')\in\Sigma\mid C''\supseteq C',\ \ [\lambda''(\delta^{i}(a)_{C''}),\lambda''(\delta^{o}(a)_{C''})]\in(p,q)_{S}\}.
\end{equation}
This is exactly the inverse image of $\uparrow C'\times (p,q)_{S}$ of the continuous function
\begin{equation}
\delta(a):\Ss\to\mathcal{V}(A)\times\IR,\ \ (C,\lambda)\to(C, [\lambda(\delta^{i}(a)_{C}),\lambda(\delta^{o}(a)_{C})]),
\end{equation}
We recognize this as the external description of the daseinisation arrow $\underline{\delta}(a):\uS_{\ast}\to\uIR$ given in Proposition 3.4. Hence, by the observations above and Lemma 3.6, the daseinisation arrow Definition 3.3 simply follows from the original covariant daseinisation arrow of \cite{hls} by replacing the partial order $\leq$ by the spectral order $\leq_{s}$ (and replacing $\mathcal{C}(A)$ by $\mathcal{V}(A)$ accordingly). Looking at  Section 5.2 of \cite{hls}, and in particular equation (54) and footnote 20, close relationship is not surprising.

\subsection{Contravariant Daseinisation Map}

The covariant daseinisation arrow $\underline{\delta}(a):\uSs\to\uIR$ in (84) was inspired\footnote{The original covariant daseinisation map was also inspired by the daseinisation constructions of the contravariant approach.} by the contravariant constructions (73) and (74). In this subsection we investigate the contravariant counterpart for the daseinisation arrow (84), which was first introduced in \cite{di3}.

In Subsection 3.1 it was shown that there is an outer daseinisation map $\delta^{o}:\mathcal{P}(A)\to\mathcal{O}_{cl}\uS$, associating a proposition of the contravariant approach to each projection operator of $A$. This map has interesting properties, which are discussed in \cite[Section 5]{di}. There is also another daseinisation arrow in the contravariant approach \cite{di3}. Before giving its definition, we first try to mimick the daseinisation map given at the end of the previous subsection. Let $\mathcal{V}(A)$ have the anti-Alexandrov topology, and give $\mathcal{V}(A)\times\IR$ the product topology. The internal interval domain in the topos $[\mathcal{V}(A)\op,\mathbf{Set}]$ is externally described by the continuous projection $\pi_{1}:\mathcal{V}(A)\times\IR\to\mathcal{V}(A)$. Analogously to (83), consider the map
\begin{equation}
f:\Sigma^{\ast}\to\IR,\ \ (C,\lambda)\mapsto [\lambda(\dasai),\lambda(\dasao)].
\end{equation}
Is this map continuous? If $(C',\lambda')\in f^{-1}(\downarrow C\times(p,q)_{S})$, then $C'\subseteq C$ and $[\lambda(\delta^{i}(a)_{C'}),\lambda(\delta^{o}(a)_{C'})]\in(p,q)_{S}$. Take any $C''\subseteq C'$ and consider $(C'',\lambda'|_{C''})$. It may very well be that
\begin{equation}
[(\lambda'|_{C''})(\delta^{i}(a)_{C''}),(\lambda'|_{C''})(\delta^{o}(a)_{C''})] <_{\IR} [\lambda'(\delta^{i}(a)_{C'}),\lambda'(\delta^{o}(a)_{C'})],
\end{equation}
where the subscript $\IR$ was added to remind the reader that in the partial order of $\IR$ we have $[a,b]\leq_{\IR}[c,d]$ whenever $a\leq c$ and $b\geq d$. Hence there is no reason why $(C'',\lambda'|_{C''})$ should be in $f^{-1}(\downarrow C\times (p,q)_{S})$, so that the map $f$ may \underline{not} be continuous and therefore may not define a locale map $\uS^{\ast}\to\uIR$ internal to $[\mathcal{V}(A)\op,\mathbf{Set}]$.

Instead of the internal interval domain $\uIR$, the contravariant approach makes use of a certain presheaf $\underline{\R}^{\leftrightarrow}$. Let $C\in\mathcal{V}(A)$ be a context. A function $\mu:(\downarrow C)\to\R$ is called \textbf{order-preserving} if for any $C'\subseteq C''$ we have $\mu(C')\leq\mu(C'')$. Denote the set of order preserving functions $(\downarrow C)\to\R$ by $\textit{OP}(\downarrow C,\R)$. A function $\nu:(\downarrow C)\to\R$ is called \textbf{order reversing} if whenever $C'\subseteq C''$ we have $\nu(C')\geq\nu(C'')$. Let $\textit{OR}(\downarrow C,\R)$ be the set of such functions. We let $\mu\leq\nu$ if for every $C'\in(\downarrow C)$ we have $\mu(C')\leq\nu(C')$. Define the presheaf $\unR^{\leftrightarrow}$ by
\begin{equation}
\unR^{\leftrightarrow}(C)=\{(\mu,\nu)\mid\mu\in \textit{OP}(\downarrow C,\R),\nu\in \textit{OR}(\downarrow C,\R),\mu\leq\nu\},
\end{equation}
\begin{equation}
\unR^{\leftrightarrow}(i):\unR^{\leftrightarrow}(C)\to\unR^{\leftrightarrow}(C')\ \ (\mu,\nu)\mapsto(\mu|_{\downarrow C'},\nu|_{\downarrow C'}),
\end{equation}
where $C,C'\in\mathcal{V}(A)$ and $C'\subseteq C$. For $a\in A$, the daseinisation $\breve{\delta}(a):\uS\to\unR^{\leftrightarrow}$ is given by
\begin{equation}
\breve{\delta}(a)_{C}(\lambda):C'\mapsto(\lambda(\delta^{i}(a)_{C'}),\lambda(\delta^{o}(a)_{C'})).
\end{equation}
The reader may check that this is a well-defined natural transformation, or find a proof and further details in \cite{di3}. 

Is it possible to see $\breve{\delta}(a)$ as a locale map? Recall that for the covariant approach, in Definition 3.5, we defined the elementary proposition $[a\in(p,q)]$ as $\delta(a)^{-1}(p,q)_{S}$. If $\breve{\delta}(a)$ can be seen as a locale map, then this would provide a contravariant counterpart for this elementary proposition defined by both inner and outer daseinisation.

In order to make sense of the question if $\breve{\delta}(a)$ is a locale map, $\unR^{\leftrightarrow}$ should have the structure of a locale in the first place. We do this once again by considering the projection $\pi_{1}:\mathcal{V}(A)\times\IR\to\mathcal{V}(A)$. We equip $\mathcal{V}(A)\times\IR$ with a different topology from the product topology, such that the projection remains continuous, and the internal frame $\mathcal{O}(\underline{\mathcal{V}(A)\times\IR})$ can be identified with $\unR^{\leftrightarrow}$. For any $C\in\mathcal{V}(A)$, pick any $(\mu,\nu)\in\unR^{\leftrightarrow}(C)$ and define $U(\mu,\nu)\subset\mathcal{V}(A)\times\IR$ as follows: $(C',[r,s])\in U(\mu,\nu)$ if $C'\subseteq C$ (equivalently $C'\in\text{dom}(\mu)=\text{dom}(\nu)$), and $\mu(C')<r\leq s<\nu(C')$. The reader may check  that these sets, with varying $C\in\mathcal{V}(A)$, form a basis for a topology. It is straightforward to check that the projection $\pi_{1}:\mathcal{V}(A)\times\IR\to\mathcal{V}(A)$ is continuous. Hence we obtain an internal locale with associated frame
\begin{equation}
\mathcal{O}(\underline{\mathcal{V}(A)\times\IR})(C)=\mathcal{O}(\mathcal{V}(A)\times\IR)|_{\downarrow C\times\IR}.
\end{equation}
The value presheaf $\underline{\R}^{\leftrightarrow}$ is a subobject of this presheaf, given by
\begin{equation}
i:\underline{\R}^{\leftrightarrow}\rightarrowtail\mathcal{O}(\underline{\mathcal{V}(A)\times\IR}),\ \ i_{C}(\mu,\nu)\mapsto U(\mu,\nu).
\end{equation}
We return to the function $f$ from (103), but now taking values in $\mathcal{V}(A)\times\IR$ equipped with the new topology. If this function is continuous, then it defines an internal locale map $\uS^{\ast}\to\underline{\mathcal{V}(A)\times\IR}$. However, as the reader may verify, the function $f$ is not continuous. Thus we have been unable to find a contravariant counterpart to the covariant elementary propositions of Definition 3.5.

\su{Observable Functions and Antonymous Functions}

In this subsection we investigate the connection between the observable functions and the antonymous functions \cite{doe, di} on the one hand, and the covariant daseinisation map of Definition 3.3 on the other. These functions were introduced in the contravariant approach in \cite{di3}, and are based on work of de Groote \cite{gro}.

Let $A$ be a von Neumann algebra, and $C\in\mathcal{V}(A)$. Let $\mathcal{F}_{C}$ denote the set of filters in $\mathcal{P}(C)$.\footnote{Recall that $F\subset\mathcal{P}(C)$ is a filter if the following three conditions are satisfied. Firstly, $0\notin F$. Secondly, if $p,q\in F$, then $p\wedge q\in F$. Thirdly, if $p\in F$, and $q\geq p$, then $q\in F$} We give $\mathcal{F}_{C}$ a topology $\mathcal{O}\mathcal{F}_{C}$, by taking the following sets as a basis:
\begin{equation}
Ext(p)=\{F\in\mathcal{F}_{C}\mid p\in F\},\ \ p\in\mathcal{P}(C).
\end{equation}
We combine the filter spaces into one ambient space, just like we did for the Gelfand spectra.\footnote{Note that the filters spaces $\mathcal{F}_{C}$ define a presheaf $\mathcal{F}:\mathcal{V}(A)\op\to\mathbf{Set}$, by $\mathcal{F}(C)=\mathcal{F}_{C}$ and for $C\subseteq C'$ the function $\mathcal{F}(i_{CC'}):\mathcal{F}_{C'}\to\mathcal{F}_{C}$ is defined as $F'\mapsto F'\cap\mathcal{P}(C)$.}

\begin{dork}
Define the set $\mathcal{F}=\coprod_{C\in\mathcal{V}(A)}\mathcal{F}(C)$. Then $\mathcal{F}$ is given the topology $\mathcal{O}\mathcal{F}$, where $U$ is open iff the following two conditions are satisfied:
\begin{enumerate}
 \item $\forall C\in\mathcal{V}(A),\ \ U_{C}\in\mathcal{O}\mathcal{F}_{C}$.
 \item If $C\subseteq C'$, $F\in U_{C}$ and $F'\in\mathcal{F}_{C'}$ such that $F'\cap\mathcal{P}(C)=F$, then $F'\in U_{C'}$.
 \end{enumerate}
The projection $\pi:\mathcal{F}\to\mathcal{V}(A)$ defines the locale $\underline{\mathcal{F}}$ in $[\mathcal{V}(A),\mathbf{Set}]$ that has the projection $\pi$ as its external description.
 \end{dork}

\begin{lem}
The map
\begin{equation}
J:\Sigma_{\ast}\to\mathcal{F},\ \ (C,\lambda)\mapsto(C,F_{\lambda}),
\end{equation}
where $F_{\lambda}=\{p\in\mathcal{P}(C)\mid\lambda(p)=1\}$, is continuous and injective, and hence it defines an injective locale map $\uSs\to\underline{\mathcal{F}}$ in $[\mathcal{V}(A),\mathbf{Set}]$.
\end{lem}

\begin{proof}
We only prove continuity of $J$, leaving the rest to the reader. Take $U\in\mathcal{O}\mathcal{F}$. We need to check that $J^{-1}(U)_{C}\in\mathcal{O}\Sigma_{C}$. First note that, $J^{-1}(U)_{C}=J^{-1}(U_{C})$. Without loss of generality, we assume that $U_{C}=Ext(p)$. We find
\begin{equation}
J^{-1}(Ext(p))_{C}=\{\lambda\in\Sigma_{C}\mid\lambda(p)=1\},
\end{equation}
which is open in $\Sigma_{C}$. Next, assume that $\lambda\in J^{-1}(U)_{C}$, $C\subseteq C'$, and $\lambda'\in\Sigma_{C'}$ such that $\lambda'|_{C}=\lambda$. From  $\lambda\in J^{-1}(U)_{C}$ it follows that $F_{\lambda}\in U_{C}$. From  $\lambda'|_{C}=\lambda$ it follows that $F_{\lambda'}\cap\mathcal{P}(C)=F_{\lambda}$. By definition of $\mathcal{O}\mathcal{F}$, we find $F_{\lambda'}\in U_{C'}$. We conclude that $\lambda'\in J^{-1}(U)_{C'}$, proving that $J$ is continuous.
\end{proof}

Now  we introduce the antonymous functions and the observable functions. Let $\mathcal{N}$ be a von Neumann algebra (read $A$ or $C$ for $\mathcal{N}$), and let $\mathcal{F}(\mathcal{N})$ denote the set of all filters in the projection lattice $\mathcal{P}(\mathcal{N})$. Let $a\in\mathcal{N}_{sa}$, with spectrum $\sigma(a)$ and spectral resolution $\{e^{a}_{r}\}_{r\in\R}$. Then the \textbf{antonymous function}  $g_{a}^{\mathcal{N}}$ is defined by  \cite{doe}
\begin{equation}
g_{a}^{\mathcal{N}}:\mathcal{F}(\mathcal{N})\to\sigma(a),\ \ F\mapsto\text{sup}\{r\in\R\mid1-e^{a}_{r}\in F\}.
\end{equation}
The \textbf{observable function}$f_{a}^{\mathcal{N}}$ is defined by \cite{doe} 
\begin{equation}
f_{a}^{\mathcal{N}}:\mathcal{F}(\mathcal{N})\to\sigma(a),\ \ F\mapsto\text{inf}\{r\in\R\mid e^{a}_{r}\in F\}.
\end{equation}

\begin{poe}
Define the map
\begin{equation}
h(a):\mathcal{F}\to\IR,\ \ h(a)(C,F)=[g^{C}_{\delta^{i}(a)_{C}}(F),f^{C}_{\delta^{o}(a)_{C}}(F)].
\end{equation}
This map is continuous and defines a locale map $\underline{\mathcal{F}}\to\uIR$ in $[\mathcal{V}(A),\mathbf{Set}]$.
\end{poe}

\begin{proof}
We use the shorthand notation $h$ for $h(a)$. For $(r,s)\in\mathcal{O}\IR$ we need to show that $h^{-1}(r,s)$ is open in $\mathcal{F}$. If $p\in\mathcal{P}(C)$, then $(\uparrow p)=\{q\in\mathcal{P}(C)\mid q\geq p\}$ is the smallest filter $\mathcal{P}(C)$ containing $p$. If $F\subseteq F'$, then it follows from the definition of $h$ that $h(F)\leq_{\IR} h(F')$. Consequently, if $(\uparrow p)\in h^{-1}(r,s)_{C}$ then $Ext(p)\subseteq h^{-1}(r,s)$. We conclude that
\begin{equation}
\bigcup_{(\uparrow p)\in h^{-1}(r,s)}Ext(p)\subseteq h^{-1}(r,s).
\end{equation}
The next step is to show that if $F\in h^{-1}(r,s)$, then there exists a $p\in F$ with the property $(\uparrow p)\in h^{-1}(r,s)_{C}$. IOnce this has been shown, the inclusion of (116) becomes an equality. If $F\in h^{-1}(r,s)$, then $r<g^{C}_{\delta^{i}(a)_{C}}(F)\leq f^{C}_{\delta^{o}(a)_{C}}(F)<s$. Define $x=g^{C}_{\delta^{i}(a)_{C}}$ and $y=f^{C}_{\delta^{o}(a)_{C}}$ and $\epsilon=1/2\text{min}(x-r,s-y)$. By definition, $y=\text{inf}\{t\in\R\mid e^{a}_{t}\in F\}$. Choose any $\epsilon_{1}\leq\epsilon$ such that $e^{a}_{y+\epsilon_{1}}\in F$. Similarly, choose an $\epsilon_{0}\leq\epsilon$ such that $1-e^{a}_{x-\epsilon_{0}}\in F$. Define $p=e^{a}_{y+\epsilon_{1}}\wedge(1-e^{a}_{x-\epsilon_{0}})$. Note that
\begin{equation}
f^{C}_{\delta^{o}(a)_{C}}(\uparrow e^{a}_{y+\epsilon_{1}})=y+\epsilon_{1}<s,
\end{equation}
\begin{equation}
g^{C}_{\delta^{i}(a)_{C}}(\uparrow(1-e^{a}_{x-\epsilon_{0}}))=x-\epsilon_{0}>r,
\end{equation}
\begin{equation}
h(\uparrow e^{a}_{y+\epsilon_{1}})\leq_{\IR},\ h(\uparrow p),\ \ \ h(\uparrow(1-e^{a}_{x-\epsilon_{0}}))\leq_{\IR}h(\uparrow p).
\end{equation}
We conclude that $(\uparrow p)\in h^{-1}(r,s)$ and that $F\in\bigcup_{(\uparrow p)\in h^{-1}(r,s)}Ext(p)$. Thus we have shown that $h^{-1}(r,s)_{C}$ is open in $\mathcal{F}_{C}$. It remains to show that if $F\in h^{-1}(r,s)_{C}$, $C\subseteq C'$ and $F\in\mathcal{F}_{C'}$ is such that $F'\cap\mathcal{P}(C)=F$, then $F'\in h^{-1}(r,s)_{C'}$. This is easily checked and left to the reader.
\end{proof}

\begin{tut}
The covariant daseinisation map $\delta(a)$ factors through $\mathcal{F}$. In other words, the following triangle is commutative:
$$\xymatrix{
\Sigma_{\ast} \ar[r]^{J} \ar[dr]_{\delta(a)} & \mathcal{F} \ar[d]^{h(a)}\\
& \IR\ \ .}$$
\end{tut}

\begin{proof}
This follows from the identities
\begin{equation}
\lambda(\delta^{i}(a)_{C})=g^{C}_{\delta^{i}(a)_{C}}(F_{\lambda}),\ \ \lambda(\delta^{o}(a)_{C})=f^{C}_{\delta^{o}(a)_{C}}(F_{\lambda}).
\end{equation}
A proof of these identities is given in \cite[Corollary 7, Corollary 9]{doe}.
\end{proof}

Now that we have shown the close relationship between the covariant daseinisation map on one hand, and the observable functions and antonymous functions on the other, we conclude by showing how the observable and antonymous functions can be of help in calculating the daseinisation arrow.

It is shown in \cite{doe} that 
\begin{equation}
g^{C}_{\delta^{i}(a)_{C}}(F_{\lambda})=g^{A}_{a}(\uparrow_{A}F_{\lambda}),\ \ f^{C}_{\delta^{o}(a)_{C}}(F_{\lambda})=f^{A}_{a}(\uparrow_{A}F_{\lambda}),
\end{equation}
where $\uparrow_{A}F_{\lambda}=\{p\in\mathcal{P}(A)\mid\exists q\in F_{\lambda},\ p\geq q\}$. This identity also follows from continuity of $h^{a}$ and the observation $(\uparrow_{A}F_{\lambda})\cap\mathcal{P}(C)=F_{\lambda}$.

Combining (120) and (121), we obtain useful identities for calculating the daseinisation arrow, viz.
\begin{equation}
\lambda(\delta^{i}(a)_{C})=\text{sup}\{r\in\R\mid1-e^{a}_{r}\in(\uparrow_{A}F_{\lambda})\},
\end{equation}
\begin{equation}
\lambda(\delta^{o}(a)_{C})=\text{inf}\{r\in\R\mid e^{a}_{r}\in(\uparrow_{A}F_{\lambda})\}.
\end{equation}

We will use the identities (122)-(123) in the proof of Lemma 3.11.

\su{The Physical Interpretation of Elementary Propositions}

In Subsection 3.2 we proposed two different versions of covariant elementary propositions. We repeat these for convenience, and add labels in order to distinguish between them:
\begin{align}
[a\in(p,q)]_{1} &=\coprod_{C\in\mathcal{V}(A)}\{\lambda\in\Sigma_{C}\mid[\lambda(\delta^{i}(a)_{C}),\lambda(\delta^{o}(a)_{C})]\in(p,q)_{S}\},\\
[a\in(p,q)]_{2} &=\coprod_{C\in\mathcal{V}(A)}\{\lambda\in\Sigma_{C}\mid\lambda(\delta^{i}(\chi_{(p,q)}(a))_{C})=1\},
\end{align}
where (124) coincides with (88) and (125) with (93). The elementary propositions $[a\in(p,q)]_{1}$, are closest to the covariant elementary propositions in \cite{hls}. We first investigate if these covariant elementary propositions fit in the neorealism scheme of Isham and D\"oring. In their setting, a physical quantity is represented by an arrow $\underline{a}:\underline{S}\to\underline{R}$, where $\underline{S}$ is the state object of the topos and $\underline{R}$ the value object. In the covariant setting we have a locale map $\underline{\delta}(a):\uS_{\ast}\to\uIR$, which means that we have an internal frame morphism $\underline{\delta}(a)^{-1}:\mathcal{O}\uIR\to\mathcal{O}\uS_{\ast}$, which is not an arrow from a state object to a value object. So we redefine $\uS_{\ast}$ and $\uIR$ a little\footnote{Alternatively, we could have adapted the ``neorealism'' formalism of D\"oring and Isham by requiring that the state and value types are always represented by locales in the topos representations, and that physical quantities are represented by continuous maps. However, in Subsection 3.3 we were unable to view the contravariant daseinisation maps as locale maps, therefore we do not use this alternative.}. Define the covariant functors $\underline{S},\underline{R}:\mathcal{V}(A)\to\Set$ by $\underline{S}(C)=\Sigma_{\ast}|_{\uparrow C}$ and $\underline{R}(C)=(\uparrow C)\times\IR$, using truncation for the transition maps. Next, for $a\in A_{sa}$ we rewrite the locale map $\underline{\delta}(a)$ as the natural transformation,
\begin{equation}
\underline{a}:\underline{S}\to\underline{R},\ \ \ \underline{a}_{C}(C',\lambda)=(C',[\lambda(\delta^{i}(a)_{C'}),\lambda(\delta^{o}(a)_{C'})]),
\end{equation}
where $C'\in(\uparrow C)$ and $\lambda\in \Sigma_{C'}$. The open $(p,q)_{S}\in\mathcal{O}\IR$ can be seen as a subobject $\underline{(p,q)}$ of $\underline{R}$ in a natural way. In the neorealist formalism of D\"oring and Isham one not only considers arrows $\underline{a}:\underline{S}\to\underline{R}$, representing physical quantities, but also the subobjects\footnote{With some abuse of notation. Strictly speaking one considers the term $\{\tilde{s}\mid a(\tilde{s})\in\tilde{\Delta}\}$ in the local language $\mathcal{L}(A)$ of the system $S$, where $a$ is the linguistic precursor of $\underline{a}$. The object $\{\tilde{s}\mid\underline{a}(\tilde{s})\in\tilde{\Delta}\}$ that we consider is the representation of this term in the functor topos.} $\{\tilde{s}\mid\underline{a}(\tilde{s})\in\tilde{\Delta}\}$, of $\underline{S}$, where $\tilde{s}$ is a variable of type $\underline{S}$, represented by $id:\underline{S}\to\underline{S}$, and $\tilde{\Delta}$ is a variable of type $\mathcal{P}\underline{R}$, represented by the identity $id:\mathcal{P}\underline{R}\to\mathcal{P}\underline{R}$  (e.g. Section 4.2. of \cite{di}). 

If we consider the object $\{\tilde{s}\mid\underline{a}(\tilde{s})\in\underline{(p,q)}\}$ instead, it is easy to prove that for $\underline{a}$ given by (126), we have
\begin{equation}
\{\tilde{s}\mid\underline{a}(\tilde{s})\in\underline{(p,q)}\}(C)=[a\in(p,q)]_{1}|_{\uparrow C}.
\end{equation}
The elementary propositions $[a\in(p,q)]_{1}$ fit well with the D\"oring and Isham formalism because this formalism uses the internal `set theory' of the topos. But aside from the formalism, do these elementary propositions make sense physically?

It was noted in Subsection 3.1 that at least at a heuristic level, the maps $(C,\lambda)\mapsto[\lambda(\delta^{i}(a)_{C}),\lambda(\delta^{o}(a)_{C})]$ fit well with the coarse-graining philosophy of the contravariant approach, as the value assigned to $a$ becomes less sharp when we move from a context $C$ to a coarser context $D$. Conversely, we can say that if we move from a context $C$ to a more refined context $C'$, the values assigned to quantities become sharper. At a heuristic level, this fits well with the Kripke model interpretation of the covariant approach given in Subsection 1.2. Moving to a more refined context, we can learn more about the system and assign sharper values to quantities.

The covariant elementary propositions $[a\in(p,q)]_{2}$ also fit very well with this Kripke model perspective. These elementary propositions assign to every context $C$, the largest available projection in $C$ that is smaller than the spectral projection $\chi_{(p,q)}(a)$. Heuristically, for each context $C$, we take the weakest proposition that can be investigated by the means of $C$, such that verification of this proposition entails that $a\in(p,q)$ holds. But what do we mean by: $``a\in(p,q)"$ holds? We could read it in an instumentalist way by saying that a measurement of $a$ with certainty yields a value in $(p,q)$. However, we prefer to consider a different way of understanding that $a\in(p,q)$ holds. Note that from (75), it follows that $\delta^{i}(\chi_{(p,q)}(a))_{C}\leq\delta^{o}(\chi_{(p,q)}(a))_{C'}$ for every $(p,q)\in\mathcal{O}\IR$, $a\in A_{sa}$ and $C,C'\in\mathcal{V}(A)$. We interpret $[a\in(p,q)]_{2}$ at context $C$ as the weakest proposition that can be investigated by the means of $C$ \emph{and} that implies that the proposition $a\in(p,q)$ is true in the sense of the contravariant approach\footnote{The elementary propositions $[a\in(p,q)]_{2}$ are defined in terms of the inner daseinisation map. In the Kripke model interpretation, what knowledge can be gained about the system from a context, plays an important role. The name \textit{daseinisation}, which (with capital D) is a reference to Heidegger, seems somewhat misplaced in this interpretation.}. 

In Section 4 we will see how states and elementary propositions pair to give truth values. In light of the interpretation of covariant elementary propositions given above, we should then check the following. If, in the covariant approach, the pairing of a state $\rho$ with an elementary proposition $[a\in(p,q)]_{2}$ yields ``true" for some context $C$, then for any context $C'\in\mathcal{V}(A)$, the pairing of $\rho$ with the contravariant elementary proposition $[a\in(p,q)]$ should yield ``true" in the contravariant approach. We will check this in Subsection 4.2.

The next lemma shows the way the two covariant elementary propositions are related.

\begin{lem}
For any $a\in A_{sa}$, $(p,q)\in\mathcal{O}\IR$, $C\in\mathcal{V}(A)$, and $\lambda\in\Sigma_{C}$, the following holds:
\begin{enumerate}
\item If $\lambda(\delta^{i}(a)_{C})> p$ and $\lambda(\delta^{o}(a)_{C})< q$, then $\lambda(\delta^{i}(\chi_{(p,q)}(a))_{C})=1$. In short, $[a\in(p,q)]_{1}\subseteq[a\in(p,q)]_{2}$.
\item If $\lambda(\delta^{i}(\chi_{(p,q)}(a))_{C})=1$, i.e. $(C,\lambda)\in[a\in(p,q)]_{2}$, then $\lambda(\delta^{i}(a)_{C})\geq p$ and $\lambda(\delta^{o}(a)_{C})\leq q$.
\end{enumerate}
\end{lem}
 
\begin{proof}
First rewrite the identities (122) and (123) of the previous subsection as
\begin{equation}
\lambda(\delta^{i}(a)_{C})=\text{sup}\{r\in\mathbb{R}\mid\exists p\in\mathcal{P}(C),\ \lambda(p)=1,\  p\leq\chi_{[r,\infty)}(a)\},
\end{equation}
\begin{equation}
\lambda(\delta^{o}(a)_{C})=\text{inf}\{r\in\mathbb{R}\mid\exists p\in\mathcal{P}(C),\ \lambda(p)=1,\  p\leq\chi_{(-\infty,r)}(a)\}.
\end{equation}
For the remainder of the proof we will denote $(p,q)_{S}\in\mathcal{O}\IR$ by $(r,s)_{S}$ instead, as we use the letters $p$ and $q$ to denote projections of $C$. If $\lambda(\delta^{i}(a)_{C})> r$, then by (128), for $\epsilon>0$ small enough there exists a projection $p_{i}\in\mathcal{P}(C)$ such that $\lambda(p_{i})=1$ and $p_{i}\leq\chi_{[r+\epsilon,\infty)}(a)$. If $\lambda(\delta^{o}(a)_{C})< s$, then by (129) there exists a projection $p_{o}\in\mathcal{P}(C)$, such that $\lambda(p_{o})=1$ and $p_{o}\leq\chi_{(-\infty,s)}(a)$.

Defining $p=p_{i}\cdot p_{o}$, we obtain a projection $p\in\mathcal{P}(C)$ such that $\lambda(p)=1$ and $p\leq\chi_{[r+\epsilon,s)}(a)\leq\chi_{(r,s)}(a)$. From the definition of the inner daseinisation map we now conclude $p\leq\delta^{i}(a)_{C}$, and subsequently $\lambda(\delta^{i}(\chi_{(r,s)}(a))_{C})=1$, proving the first claim of the lemma.

For the second claim, assume that  $\lambda(\delta^{i}(\chi_{(r,s)}(a))_{C})=1$. Noting that $\delta^{i}(\chi_{(r,s)}(a))_{C})\leq\chi_{(-\infty,s)}(a)$, the claim  $\lambda(\delta^{o}(a)_{C})\leq s$ follows from (129). Using  $\delta^{i}(\chi_{(r,s)}(a))_{C}\leq\chi_{[r,\infty)}(a)$, the claim $\lambda(\delta^{i}(a)_{C})\geq r$ follows from (128).
\end{proof}
Note that the converse of the second claim in the lemma does not hold in general. For a counterexample, consider $a=rp$, with $p$ a nontrivial projection.

Which version of covariant elementary proposition is to be preferred? We will consider this question in Subsection 4.2, where we will see how the elementary propositions pair with states. 

We close this section with a short comparison between the covariant approach and the work of Coecke on intuitionistic quantum logic \cite{coe}. In the contravariant approach, one often considers the map $\delta^{o}:\mathcal{P}(A)\to\mathcal{O}_{cl}\uS$, where $\mathcal{O}_{cl}\uS$ denotes the complete Heyting algebra of closed open subobjects of the spectral presheaf (e.g. Section 4 of \cite{doe}). In order to compare this with the covariant construction, we see $\delta^{o}$ as a map
\begin{equation}
\delta^{o}:\mathcal{P}(A)\to\mathcal{O}\Sigma^{\ast},\ \ \ \delta^{o}(p)=\coprod_{C\in\mathcal{V}(A)}\{\lambda\in\Sigma_{C}\mid\lambda(\delta^{o}(p)_{C})=1\}.
\end{equation}
This map is injective, and preserves all joins of the projection lattice $\mathcal{P}(A)$, but not the meets. It cannot preserve both, as $\mathcal{O}\Sigma^{\ast}$ (and likewise $\mathcal{O}_{cl}\uS$) is distributive, whereas $\mathcal{P}(A)$ is nondistributive. Dually, for the covariant approach we define 
\begin{equation}
\delta^{i}:\mathcal{P}(A)\to\mathcal{O}\Sigma_{\ast},\ \ \ \delta^{i}(p)=\coprod_{C\in\mathcal{V}(A)}\{\lambda\in\Sigma_{C}\mid\lambda(\delta^{i}(p)_{C})=1\}.
\end{equation}
This injective map preserves all meets of $\mathcal{P}(A)$, but generally it does not preserve the joins\footnote{This can be shown in the same way as for the dual properties of $\delta^{o}$.}. Thinking of $\mathcal{P}(A)$ as a property lattice, the map $\delta^{i}$ is in fact a balanced inf-embedding (the relevant definitions can be found in \cite{coe}). For each property $p\in\mathcal{P}(A)$, the open $\delta^{i}(p)\in\mathcal{O}\Sigma_{\ast}$ can be thought of as the set of pure states in context $(C,\lambda)$, such that the property $p$ is actual in that state. Thinking of the elements of $\mathcal{P}(A)$ in operational terms and thinking of conjunction and disjunction intuitionistically (just like in the Kripke model interpretation of the covariant approach), the meets of $\mathcal{P}(A)$ should be preserved by $\delta^{i}$, as these coincide with conjunctions. The joins of $\mathcal{P}(A)$ need not be preserved, as these do not coincide with disjunctions because of superposition\footnote{For two properties $p_{1},p_{2}\in\mathcal{P}(A)$, it may be the case that the property $p_{1}\vee p_{2}$ is actual for a given state, while neither property $p_{1}$ nor property $p_{2}$ is actual for that same state.}. 

Typically in Coecke's approach we consider a balanced inf-embedding 

\noindent $\mu:L\to H$, with $L$ the property lattice and $H$ a complete Heyting algebra which is the injective hull of $L$ (obtained using the Bruns-Lakser construction). Although $\mu$ need not preserve all joins, it should preserve those joins where superposition is not an issue. This means that (under the assumption of superpositional faithfulness) $\mu$ should preserve all distributive joins. A join $p_{1}\vee p_{2}$ in $L$, is called distributive if for every $q\in L$ we have $q\wedge(p_{1}\vee p_{2})=(q\wedge p_{1})\vee(q\wedge p_{2})$. The frame $\mathcal{O}\Sigma_{\ast}$ is the injective hull of $\mathcal{P}(A)$ iff the following two conditions are satisfied.
\begin{enumerate}
\item $\delta^{i}(\mathcal{P}(A))$ is join dense in $\mathcal{O}\Sigma_{\ast}$: 

\noindent If $U\in\mathcal{O}\Sigma_{\ast}$, then $U=\bigcup\{\delta^{i}(p)\mid\delta^{i}(p)\subseteq U\}$.
\item $\delta^{i}$ preserves distributive joins.
\end{enumerate}
General, neither condition is satisfied, as we can see from the simple example $A=M_{3}(\mathbb{C})$. Consider the singleton $U=\{(C,\lambda)\}\in\mathcal{O}\Sigma_{\ast}$, where $C$ is any maximal context and $\lambda$ any element of its spectrum. The only $\delta^{i}(p)$ that is a subset of $U$, is $\delta^{i}(0)=\emptyset$. This is not a real problem. The reader can check that for any von Neumann algebra $A$, the opens of the form $\delta^{i}(p)$ form a basis for a topology on $\Sigma$, and if we restrict $\mathcal{O}\Sigma_{\ast}$ to the topology generated by the opens $\delta^{i}(p)$, then the first condition is satisfied. 

Returning to the $A=M_{3}(\mathbb{C})$ example, we will show $\delta^{i}$ does not preserve distributive joins. Consider the set $X\subset\mathcal{P}(M_{3}(\mathbb{C}))$ consisting of all rank 1 projections. The reader can check that $X$ has a distributive join. For every context $C\in\mathcal{V}(A)$, $\delta^{i}(\bigvee X)_{C}=\delta^{i}(1)_{C}=\Sigma_{C}$. For the trivial context\footnote{In fact, any nonmaximal context can be used for this counterexample.} $\mathbb{C}1$, $\bigvee_{p\in X}\delta^{i}(p)_{\mathbb{C}1}=\bigvee_{p\in X}\emptyset=\emptyset$. The distributive join is not preserved. 

Distributive joins are not preserved because of contextuality. If $\lambda\in\Sigma_{C}$ is a state in context $C$, then for any property $p\in\mathcal{P}(A)$ we can only say that $p$ is certain for state $\lambda$ in context $C$ if there is a property $q\in\mathcal{P}(A)$ which can be invesitgated from the context $C$, i.e.   $q\in\mathcal{P}(C)$, and which implies $p$, i.e. $q\leq p$. In the example above  $\bigvee_{p\in X}\delta^{i}(p)_{\mathbb{C}1}=\emptyset$ because the context $\mathbb{C}1$ is so coarse that only trvial properties such as $1=\bigvee X$ can be inferred from it.

\se{States and Truth Values}

Our comparison would not be complete without a discussion of states and the way these states combined with propositions yield truth values. We start with a short discussion of state-related objects in the contravariant approach. In Subsection 4.2 the states of the covariant approach are introduced and the physical interpretation of the state proposition pairing is discussed. In Subsection 4.3 we compare the states and pairings of the two approaches.

\su{Contravariant Approach}

We start with state-related objects in the contravariant approach. None of the material presented in this subsection is new. We first discuss pseudo-states and truth objects. After that we treat the more recent measures introduced by D\"oring \cite{doe3, doe2}. By the Kochen-Specker theorem, the spectral presheaf typically does not have global points. Thus global points $\underline{1}\to\uS$ do not give a fruitful concept of state. Let $A=\mathcal{B}(\mathcal{H})$, and let $\stat\in\mathcal{H}$ be a unit vector. In the contravariant approach one associates two closely related objects to the vector $\stat$, namely the truth object $\underline{\mathbb{T}}^{\stat}$ and the pseudo-state $\underline{\mathfrak{w}}^{\stat}$. A more complete discussion of these objects may be found in \cite[Section 6]{di}, \cite{doe2}.

In order to define $\truob$ it is convenient to first introduce the so-called \textbf{outer presheaf} $\underline{O}:\mathcal{V}(A)\op\to\mathbf{Set}$. For $C\in\mathcal{V}(A)$ we have $\underline{O}(C)=\mathcal{P}(C)$, i.e. the set of projection operators in $C$. If $C\subseteq C'$ we have $\underline{O}(i):\mathcal{P}(C')\to\mathcal{P}(C)$ given by $P\mapsto\daspo$. Each projection operator $p\in\mathcal{P}(\mathcal{H})$ defines a global point $\underline{1}\to\underline{O}$ of the outer presheaf by outer daseinisation $\underline{\delta}^{o}(p)$, which at stage $C$ picks the projection operator $\daspo$. The truth object $\truob$ is a subobject of the outer presheaf, given by
\begin{equation}
\truob_{C}=\{p\in\mathcal{P}(C)\mid\langle\psi\vert p\stat=1\}=\{ p\in\mathcal{P}(C)\mid p\geq\stat\langle\psi\vert\}.
\end{equation}
It is shown in Subsection 6.5.2 of \cite{di} that there is a monic arrow $\underline{O}\rightarrowtail\mathcal{O}_{cl}\uS$. The truth object $\truob$ can be seen as a subobject of $\mathcal{O}_{cl}\uS$, or, equivalently, as a point of $P\mathcal{O}_{cl}\uS$. The truth object has been defined for a vector state $\stat$, but there is also a generalization for mixed states, which we are going to discuss at the end of this subsection.

The point $\underline{\delta}^{o}(p):\underline{1}\to\underline{O}$ can also be viewed as an open closed subobject of the spectral presheaf, as we have seen in Subsection 3.1. Thus it represents a proposition in the contravariant approach. Together with the truth object $\truob$, it forms the sentence $\underline{\delta}^{o}(p)\in\truob$ in the language of $[\mathcal{V}(A)\op,\mathbf{Set}]$. A sentence is represented by a subobject of the terminal object $\underline{1}$ and hence is equivalent to a truth value $\underline{1}\to\underline{\Omega}$. Recall that $\underline{\Omega}(C)$ is the set of sieves on $C$. At context $C$, the truth value of $\underline{\delta}^{o}(p)\in\truob$ is given by
\begin{equation}
\nu(\underline{\delta}^{0}(p)\in\truob)_{C}=\{C'\in(\downarrow C)\mid\langle\psi\vert\delta^{o}(p)_{C'}\stat=1\}.
\end{equation}

The second state-related object is the \textbf{pseudo-state} $\pseu$. This is a subobject of the spectral presheaf, defined by
\begin{equation}
\pseu_{C}=\underline{\delta}^{o}(\stat\langle\psi\vert)_{C}=\{\lambda\in\Sigma_{C}\mid\lambda(\delta^{o}(\stat\langle\psi\vert)_{C})=1\},
\end{equation}
where $\stat\langle\psi\vert$ denotes the projection onto the ray $\C\stat$. Once again, consider $\underline{\delta}^{o}(p)_{C}$. Rather than as a point of the outer presheaf, it is now seen as a subobject of the spectral presheaf, as in (66). We form the sentence $\pseu\subseteq\underline{\delta}^{o}(p)$, whose associated truth value is
\begin{equation}
\nu(\pseu\subseteq\underline{\delta}^{o}(p))_{C}=\{C'\in(\downarrow C)\mid \delta^{o}(\stat\langle\psi\vert)_{C'}\leq\delta^{o}(p)_{C'}\}.
\end{equation}

\begin{poe}{(Section 6.4.2 \cite{di})}
Let $A=\mathcal{B}(\mathcal{H})$, and let $\stat\in\mathcal{H}$ be a unit vector. Then,
\begin{equation}
\forall_{C\in\mathcal{V}(A)}\ \ \nu(\underline{\delta}^{0}(p)\in\truob)_{C}=\nu(\pseu\subseteq\underline{\delta}^{o}(p))_{C}.
\end{equation}
\end{poe}

\begin{proof}
The observation to make is that $\langle\psi\vert\delta^{o}(p)_{C}\stat=1$ iff $\stat\langle\psi\vert\leq_{s}\delta^{o}(p)_{C}$, \cite{di}. Suppose that $C'\in\nu(\underline{\delta}^{0}(p)\in\truob)_{C}$, which is equivalent to $C'\Vdash\underline{\delta}^{0}(p)\in\truob$. This implies that $\delta^{o}(p)_{C'}\geq_{s}\stat\langle\psi\vert$. By definition, $\delta^{0}(\stat\langle\psi\vert)_{C'}$ is the smallest projection operator in $C'$ that is greater than $\stat\langle\psi\vert$. It follows that  $\delta^{o}(p)_{C'}\geq\delta^{0}(\stat\langle\psi\vert)_{C'}$. Thus $C'\in\nu(\pseu\subseteq\underline{\delta}^{o}(p))_{C}$. Conversely, assume that $C'\in\nu(\pseu\subseteq\underline{\delta}^{o}(p))_{C}$, which is equivalent to $C'\Vdash\pseu\subseteq\underline{\delta}^{o}(p)$. Then  $\delta^{o}(p)_{C'}\geq\delta^{0}(\stat\langle\psi\vert)_{C'}\geq\stat\langle\psi\vert$. It is immediate that $C'\in\nu(\underline{\delta}^{0}(p)\in\truob)_{C}$.
\end{proof}

In \cite{doe3} and \cite{doe2}, D\"oring uses measures of closed open subobjects of the spectral presheaf in order to describe states. This description of states has the advantage that it generalizes to mixed states. 

\begin{dork}
A measure on the spectral presheaf is a function $\mu:\mathcal{O}_{cl}\uS\to\text{OR}(\mathcal{V}(A),[0,1])$, such that for every $C\in\mathcal{V}(A)$ and $\forall \underline{S}_{1},\underline{S}_{2}\in\mathcal{O}_{cl}\uS$,
\begin{itemize}
\item $\mu(\uS)(C)=1$;
\item $\mu(\underline{S}_{1})(C)+\mu(\underline{S}_{2})(C)=\mu(\underline{S}_{1}\wedge\underline{S}_{2})(C)+\mu(\underline{S}_{1}\vee\underline{S}_{2})(C)$.
\item For any fixed $C\in\mathcal{V}(A)$, the function $\mu^{C}:=\mu(-)(C):\mathcal{O}_{cl}\uS\to[0,1]$, $\underline{S}\mapsto\mu(\underline{S})(C)$ depends only on $\underline{S}_{C}$. We write $\mu^{C}(\underline{S})=\mu^{C}(\underline{S}_{C})$ with slight abuse of notation.
\end{itemize}
\end{dork}

Any state, in the guise of a normalized, positive linear functional, $\rho:A\to\mathbb{C}$, defines a measure by $\mu_{\rho}(\underline{S})(C)=\rho(p_{\underline{S}(C)})$, where $p_{\underline{S}(C)}$ denotes the projection corresponding to the closed open subset $\underline{S}(C)$ of $\Sigma_{C}$. In order to see that these measures in fact generalize pseudo-states and truth objects, we use the internal structure of the topos. For example, the lower reals $\underline{[0,1]}_{l}$ in the topos $[\mathcal{V}^{\op},\Set]$ are given by the presheaf $\underline{[0,1]}_{l}(C)=\textit{OR}(\downarrow C,[0,1])$ (in order to show this, recall that $[\mathcal{V}^{\op},\Set]$ is equivalent to the topos of sheaves on $\mathcal{V}(A)$ equipped with the downset topology). Any measure as in Definition 4.2 defines a natural transformation
\begin{equation}
\underline{\mu}:\mathcal{O}\uS_{cl}\to\underline{[0,1]}_{l},\ \ \ (\underline{\mu})_{C}(\underline{S})=\mu(\underline{S})|_{\downarrow C},
\end{equation}
where $\mathcal{O}\uS_{cl}$ is the presheaf $\mathcal{O}\uS_{cl}(C)=\mathcal{O}_{cl}\uS|_{\downarrow C}$, which is the frame of Proposition 2.3. The function $\underline{1}_{l}:\mathcal{V}(A)\to[0,1]$ that is constantly $1$ can be seen as a global point $\underline{1}_{l}:\underline{1}\to\underline{[0,1]}_{l}$. A measure $\underline{\mu}$ as in (137) together with a subobject $\underline{S}\in\text{Sub}_{cl}\uS$, define (by means of the language of the topos) a truth value $[\underline{\mu}(\underline{S})=\underline{1}_{l}]:\underline{1}\to\underline{\Omega}$. If $\underline{\mu}$ comes from a vector state $\psi$, then
\begin{equation}
(\underline{\mu}_{\psi})_{C}(\underline{S})(C)=\langle\psi|p_{\underline{S}(C)}|\psi\rangle.
\end{equation}
Taking $\underline{S}=\underline{\delta}^{o}(p)$, we then find
\begin{equation}
\nu(\underline{\mu}_{\psi}(\underline{\delta}^{o}(p))=\underline{1}_{l})_{C}=\{C'\in(\downarrow C)\mid \langle\psi|\delta^{o}(p)_{C}|\psi\rangle=1\}.
\end{equation}

The measures of Definition 4.2 that come from vector states yield exactly the same truth values as pseudo-states and truth objects paired with propositions. In this sense, the measures of Definition 4.2 are a generalization of both pseudo-states and truth objects.

\su{Covariant Approach}

Next, we consider covariant states and how these combine with elementary propositions. We only present a short discussion of the subject. A more complete treatment can be found in \cite[Section 4]{hls}. 

In the covariant approach, a state is described by a probability valuation on the spectrum $\uS_{\uA}$. 

\begin{dork}
Let $X$ be a locale, and let $[0,1]_{l}$ be the set of lower reals between 0 and 1. A \textbf{probability valuation} on $X$ is a monotone map $\mu:\mathcal{O}X\to[0,1]_{l}$ satisfying the following conditions. Let $U,V\in\mathcal{O}X$ and $\{U_{\lambda}\}_{\lambda\in I}\subseteq\mathcal{O}X$ be a directed subset. Then:
\begin{itemize}
\item $\mu(0)=0,\  \mu(1)=1$;
\item $\mu(U)+\mu(V)=\mu(U\wedge V)+\mu(U\vee V)$;
\item $\mu(\bigvee_{\lambda\in I}U_{\lambda})=\bigvee_{\lambda\in I}\mu(U_{\lambda})$.
\end{itemize}
\end{dork}

Assume for conveniece that the C*--algebra $A$ is a subalgebra of some $\mathcal{B}(\mathcal{H})$. A unit vector $\stat\in\mathcal{H}$ defines a state on $A$ (in the sense of a positive normalized linear functional) by $\rho_{\psi}:A\to\C$, $\rho_{\psi}(a)=\langle\psi\vert a\stat$. A state $\rho_{\psi}:A\to\C$ defines a probability integral $I_{\psi}:\uA_{sa}\to\underline{\R}$, on the Bohrification $\uA$ \cite[Definition 10, Theorem 14]{hls}. By the generalized Riesz-Markov Theorem, \cite{coqsp2, hls}, probability integrals $I:\uA_{sa}\to\underline{\R}$ correspond to probability valuations $\underline{\mu}:\mathcal{O}\uS_{\uA}\to\underline{[0,1]}_{l}$. In this way, any unit vector $\stat\in\mathcal{H}$ gives rise to a probability valuation $\underline{\mu}_{\psi}$ on $\uS_{\uA}$.

Before we explore what probability valuations on $\uS_{\uA}$ look like externally, we first explain how these valuations combine with propositions, so as to give truth values. As before, identify the internal spectrum $\uS_{\uA}$ with the locale $\uSs$. The lower reals $\underline{[0,1]}_{l}$ in $[\mathcal{C}(A),\mathbf{Set}]$ are given by $\underline{[0,1]}_{l}(C)=L(\uparrow C,[0,1])$, \cite[Appendix A.3]{chls}, where the right-hand side stands for the set of lower semicontinuous functions $(\uparrow C)\to[0,1]$. A function $f:(\uparrow C)\to[0,1]$ is lower semicontinuous iff it is order-preserving: if $C\subseteq C'$, then $f(C)\leq f(C')$. 

Let $1_{C}:\uparrow C\to[0,1]$ denote the function that is constantly 1. Define 
\begin{equation}
\underline{1}_{l}:\mathcal{O}\uSs\to\underline{[0,1]}_{l}, \ \ (\underline{1}_{l, C})(U)=1_{C}. 
\end{equation}
Let $\underline{\mu}:\mathcal{O}\uSs\to\underline{[0,1]}_{l}$ be a probability valuation on $\uSs$. Using the internal language of $[\mathcal{C}(A),\mathbf{Set}]$, we form the arrow
\begin{equation}
[\underline{\mu}=\underline{1}_{l}]:\mathcal{O}\uSs\to\underline{\Omega}.
\end{equation}
Any open $U\in\mathcal{O}\Sigma_{\ast}$ yields a point $\underline{U}:\underline{1}\to\mathcal{O}\uSs$. For any probability valuation $\underline{\mu}$ on the spectrum of $\uA$ and any proposition $U\in\mathcal{O}\Sigma_{\ast}$, we obtain a truth value
\begin{equation}
[\underline{\mu}(\underline{U})=\underline{1}_{l}]=[\underline{\mu}=\underline{1}_{l}]\circ\underline{U}:\underline{1}\to\underline{\Omega}.
\end{equation}

A probability valuation $\underline{\mu}$ is completely determined by $\mu=\underline{\mu}_{\C\cdot1}$. Interpreting Definition 4.3 in $[\mathcal{C}(A),\mathbf{Set}]$, we obtain the following proposition.

\begin{poe}
Let $\textit{OP}(\mathcal{C}(A),[0,1])$ denote the set of order-preserving functions $\mathcal{C}(A)\to[0,1]$. Then a function $\mu:\mathcal{O}\Sigma_{\ast}\to\textit{OP}(\mathcal{C}(A),[0,1])$ defines a probability valuation $\underline{\mu}$ on $\mathcal{O}\uSs$ by $\underline{\mu}_{\C\cdot1}=\mu$, iff the following four conditions are satisfied for all $C\in\mathcal{C}(A)$:
\begin{enumerate}
\item If $U,V\in\mathcal{O}\Sigma_{\ast}$, such that $U\subseteq V$, then $\mu(U)(C)\leq\mu(V)(C)$.
\item $\mu(\Sigma_{\ast})(C)=1,\  \mu(\emptyset)(C)=0$.
\item If $U,V\in\mathcal{O}\Sigma_{\ast}$, then
\begin{equation}
\mu(U)(C)+\mu(V)(C)=\mu(U\cap V)(C)+\mu(U\cup V)(C).
\end{equation}
\item If $\{U_{i}\}_{i\in I}\subset\mathcal{O}\Sigma_{\ast}$ is a directed set, then
\begin{equation}
\mu\left(\bigcup_{i\in I}U_{i}\right)(C)=\sup\{\mu(U_{i})(C)\mid i\in I\}.
\end{equation}
\end{enumerate}
\end{poe}

\begin{proof}
The proposition follows from Definition 4.3 by applying sheaf semantics \cite[Section VI.7]{mm}. For the reader unfamiliar with sheaf semantics, we prove part of the proposition by showing that (1) of Proposition 4.4 follows from monotonicity of the valuation $\underline{\mu}:\mathcal{O}\uSs\to\underline{[0,1]}_{l}$.  By monotonicity, for every $C\in\mathcal{V}(A)$ we have
\begin{equation}
C\Vdash\forall\underline{U},\underline{V}\in\mathcal{O}\uSs\left(\underline{U}\leq\underline{V}\Rightarrow\underline{\mu}(\underline{U})\leq\underline{\mu}(\underline{V})\right)\nonumber.
\end{equation}
Applying the rules of sheaf semantics, this is equivalent to: for every $C\in\mathcal{V}(A)$ and any $U,V\in\mathcal{O}\Sigma_{\ast}|_{\uparrow C}$, if $U\subseteq V$, then $\underline{\mu}_{C}(U)\leq\underline{\mu}_{C}(V)$. From naturality of $\underline{\mu}$, we know
\begin{equation}
\forall_{U\in\mathcal{O}\Sigma_{\ast}|_{\uparrow C}}\ \underline{\mu}_{C}(U)=\underline{\mu}_{\C\cdot1}(U)|_{\uparrow C}.\nonumber
\end{equation}
Note that
\begin{equation}
\underline{\mu}_{C}(U)\leq\underline{\mu}_{C}(V)\Leftrightarrow\forall_{C'\in(\uparrow C)}\ \underline{\mu}_{C}(U)(C')\leq\underline{\mu}_{C}(V)(C').\nonumber
\end{equation}
Property (1) of Proposition 4.4 follows from combining these observations.
\end{proof} 

\begin{dork}
A \textbf{covariant state} is a function  $\mu:\mathcal{O}\Sigma_{\ast}\to\textit{OP}(\mathcal{C}(A),[0,1])$ that satisfies the conditions of Proposition 4.4. 
\end{dork}

We would like to define a covariant state from a unit vector $\stat\in\mathcal{H}$. An obvious definition would be 
\begin{equation} 
\mu_{\psi}(B^{C}_{a}) = \left\{ 
\begin{array}{rl} 
0 & \text{if } \langle\psi\vert a\stat=0\\ 
1 & \text{if } \langle\psi\vert a\stat>0
\end{array} \right., 
\end{equation}
where $B^{C}_{a}=\{(C',\lambda')\mid C'\in(\uparrow C),\lambda'(a)>0\}$, and $a\in C^{+}$. Here $0,1$ do not denote numbers but constant functions $\mathcal{C}(A)\to[0,1]$. However, (145) is not a good choice, because using only constant functions for $\mu_{\psi}(U):\mathcal{C}(A)\to[0,1]$ has the following undesirable consequences.

Let $\mu$ be a covariant state such that for every $U\in\mathcal{O}\Sigma_{\ast}$, $\mu(U)$ is a constant function. Assume that $C$ is an element of the truth value $\nu(\underline{\mu}(\underline{U})=\underline{1}_{l})_{\C\cdot 1}$. This translates to $\mu(U\vert_{\uparrow C})\vert_{\uparrow C}=1_{C}$. Because of monotonicity (property (1) of Proposition 4.4), and using that $\mu(U)$ is constant, we find $\mu(U)=1_{\C\cdot1}$. We conclude that $\C\cdot1\in\nu(\underline{\mu}(\underline{U})=\underline{1}_{l})_{\C\cdot 1}$. In short, for every $U\in\mathcal{O}\Sigma_{\ast}$ we find $\nu(\underline{\mu}(\underline{U})=\underline{1}_{l})_{\C\cdot 1}\in\{\emptyset,\mathcal{C}(A)\}$. Hence a proposition is either true at every stage $C$ or at no stage $C$. On a related note, suppose that we replace the lower reals $[0,1]_{l}$ in Definition 4.3 by the Dedekind reals $[0,1]_{d}$. This amounts to replacing the lower semicontinuous functions $(\uparrow C)\to[0,1]$ by continuous functions $(\uparrow C)\to[0,1]$ (D4.7 \cite{jh1}). This in turn amounts to replacing order preserving functions $(\uparrow C)\to[0,1]$ by constant functions $(\uparrow C)\to[0,1]$. Replacing lower reals by Dedekind reals thus entails that $\nu(\underline{\mu}(\underline{U})=\underline{1}_{l})_{\C\cdot 1}\in\{\emptyset,\mathcal{C}(A)\}$. Therefore, in the definition of covariant states we will use the additional freedom given by the lower reals, by considering non-constant functions for $\mu(U):\mathcal{C}(A)\to[0,1]$.

Next, we prepare for the definition of the covariant state defined by a C*-algebraic state $\rho:A\to\mathbb{C}$. This covariant state $\mu_{\rho}$ will be the external description of the internal probability valuation on $\uS_{\uA}$, obtained by applying the generalized Riesz-Markov theorem to the internal state $I_{\rho}:\uA_{sa}\to\underline{\R}$.\footnote{See also \cite{hls}, in particular page 96, just below equation (61).}  

Take any context $C\in\mathcal{C}(A)$ and note that by the generalized Riesz-Markov Theorem, the state $\rho$ defines a probability valuation
\begin{equation}
\mu^{C}_{\rho}:\mathcal{O}\Sigma_{C}\to[0,1],\ \mu^{C}_{\rho}(X^{C}_{a})=\sup\{\rho(n\cdot a\wedge1)\mid n\in\N\},
\end{equation}
where $a\in C^{+}$. First, consider the case where $a$ is a projection $p\in\mathcal{P}(C)$, in which case $X^{C}_{p}=\{\lambda\in\Sigma_{C}\mid \lambda(p)>0\}$ is both closed and open. In that case, we have 
\begin{equation}
\mu^{C}_{\rho}(X^{C}_{p})=\sup\{\rho(n\cdot p\wedge1)\mid n\in\N\}=\rho(p).
\end{equation}
From this point onwards we restrict to von Neumann algebras. As a consequence, the closed open subsets of $\Sigma_{C}$ form a basis for the topology on $\Sigma_{C}$. Any open $U_{C}$ of $\Sigma_{C}$ is the directed join of all the closed open sets contained in it. It follows from Proposition 4.4(4) that the maps of the next lemma are exactly the covariant states obtained by the Riesz-Markov theorem.

\begin{lem}
For each state $\rho:A\to\mathbb{C}$, the map $\mu_{\rho}:\mathcal{O}\Sigma_{\ast}\to\textit{OP}(\mathcal{C}(A),[0,1])$,
 \begin{equation}
 \mu_{\rho}(U)(C)=\mu^{C}_{\rho}(U_{C})=\sup\{\rho(p)\mid p\in\mathcal{P}(C), X^{C}_{p}\subseteq U_{C}\}.
\end{equation}
defines a covariant state in the sense of Definition 4.5.
\end{lem}

Consider the pairing of the states in Lemma 4.6, and the covariant elementary propositions $[a\in(p,q)]_{1}$ and $[a\in(p,q)]_{2}$, given by (124) resp. (125).  We start with $[a\in(p,q)]_{2}$, as these are the easiest to compute.

\begin{lem}
For a state $\underline{\mu}_{\rho}$, that has external description $\mu_{\rho}$ as in Lemma 4.6, and any elementary proposition $\underline{[a\in(r,s)]}_{2}\in\mathcal{O}\uS_{\ast}$ the following are equivalent:
\begin{enumerate}
\item $C\Vdash\underline{\mu}_{\rho}(\underline{[a\in(r,s)]}_{2})=\underline{1}_{l}$;
\item $\rho(\delta^{i}(\chi_{(r,s)}(a))_{C})=1$. 
\end{enumerate}
\end{lem}

\begin{proof}
Condition (1) is equivalent to the external identity $\mu_{\rho}([a\in(r,s)]_{2})=1_{l}$. The desired equivalence then follows straight from the equality
\begin{equation*}
\mu_{\rho}^{C}(([a\in(r,s)]_{2})_{C})=\mu^{C}_{\rho}(X^{C}_{\delta^{i}(\chi_{(r,s)}(a))_{C}})=\rho(\delta^{i}(\chi_{(r,s)}(a))_{C}).\qedhere
\end{equation*}
\end{proof}

Note that truth at stage $C$ for this pairing implies that for every context $C'$ in $\mathcal{V}(A)$, one has $C'\Vdash\underline{\mathfrak{w}}^{\stat}\subseteq\underline{[a\in(p,q)]}$ in the contravariant approach. The pairing of the previous lemma is consistent with the interpretation of covariant elementary propositions given in Subsection 3.5.

Next, consider the elementary propositions $[a\in(r,s)]_{1}$, given by (124), and note that
\begin{equation}
([a\in(r,s)]_{1})_{C}=X^{C}_{\delta^{i}(a)_{C}-r}\cap X^{C}_{s-\delta^{o}(a)_{C}}.
\end{equation}

\begin{lem}
For a state $\underline{\mu}_{\rho}$ that has an external description as in Lemma 4.6, and any elementary proposition $\underline{[a\in(r,s)]}_{1}\in\mathcal{O}\uS_{\ast}$ the following are equivalent:
\begin{enumerate}
\item $C\Vdash\underline{\mu}_{\rho}(\underline{[a\in(r,s)]}_{1})=\underline{1}_{l}$;
\item $\rho(\delta^{i}(a)_{C})\in[r,s]$ and  $\rho(\delta^{o}(a)_{C})\in[r,s]$. 
\end{enumerate}
\end{lem}

\begin{proof}
By Proposition 4.4(3), $\mu^{C}_{\rho}(X^{C}_{\delta^{i}(a)_{C}-r}\cap X^{C}_{s-\delta^{o}(a)_{C}})=1$ iff $\mu^{C}_{\rho}(X^{C}_{\delta^{i}(a)_{C}-r})=1$ and $\mu^{C}_{\rho}(X^{C}_{s-\delta^{o}(a)_{C}})=1$. Define $b=\delta^{i}(a)_{C}-r$ and $c=s-\delta^{o}(a)_{C}$. We want to calculate $\mu^{C}_{\rho}(X^{C}_{b})$ and $\mu^{C}_{\rho}(X^{C}_{c})$. Note that the closure of $X^{C}_{b}$, which is $\{\lambda\in\Sigma_{C}\mid\lambda(b)\geq0\}$, is both open and closed in $\Sigma_{C}$, as the closure of any open set  in the spectrum of a commutative von Neumann algebra is again open \cite{tak}. Let $p_{0}\in\mathcal{P}(C)$ be the projection corresponding to this closed open subset. By definition of $\mu_{\rho}$, we have $\mu^{C}_{\rho}(X^{C}_{b})\leq\rho(p_{0})$. 

For every $q\in\mathbb{Q}^{+}$, let $p_{q}\in\mathcal{P}(C)$ be the projection corresponding to the closed open set $\{\lambda\in\Sigma_{C}\mid\lambda(b)\geq q\}$. For every $q\in\mathbb{Q}^{+}$, clearly $X^{C}_{p_{q}}\subseteq X^{C}_{b}$, so $\mu^{C}_{\rho}(X^{C}_{b})\geq\rho(p_{q})$. Also
\begin{equation}
\rho(p_{0})=\sup\{\rho(p_{q})\mid q\in\mathbb{Q}^{+}\}\leq\mu^{C}_{\rho}(X^{C}_{b}).
\end{equation}
We conclude that 
\begin{equation}
\mu^{C}_{\rho}(X^{C}_{b})=\rho(p_{0})=\rho(\chi_{[0,\infty)}(b))=\rho(\chi_{[r,\infty)}(\delta^{i}(a)_{C})).
\end{equation}
By (151), $\mu^{C}_{\rho}(X^{C}_{b})=1$ is equivalent to $\rho(\delta^{i}(a)_{C})\in[r,\infty)$. In a similar way we can show that $\mu^{C}_{\rho}(X^{C}_{c})=1$ iff $\rho(\delta^{o}(a)_{C})\in(-\infty,s]$. The lemma follows easily.
\end{proof}

The proposition $[a\in(p,q)]_{1}$ is true at stage of knowledge $C$ iff the expectation values of the best approximations of $a$ in $C$ lie in $[p,q]$. This in turn implies $\rho(a)\in[p,q]$, because $\delta^{i}(a)_{C}\leq a\leq\delta^{o}(a)_{C}$. Note that the use of the spectral order in the definition of daseinisation is crucial here. If we had used the usual order on operators instead of the spectral order (which amounts to using the original covariant daseinisation arrow \cite{hls}), then $ \delta^{i}(a)_{C}\nleq a$ or $a\nleq\delta^{o}(a)_{C}$ may very well be the case (see the example at the end of Subsection 3.1), and $C\Vdash\underline{\mu}_{\psi}(\underline{[a\in(p,q)]}_{1})=\underline{1}_{l}$ need not tell us anything about $a$. 

We can now assign truth values to both versions of elementary propositions, when paired with states. The obvious question is now: which version of elementary propositions is to be preferred: the propositions $[a\in(p,q)]_{1}$, which are close to the original covariant elementary propositions\footnote{The only real difference is the use of the spectral order.}, or the elementary propositions $[a\in(p,q)]_{2}$, which mirror the contravariant elementary propositions? The next theorem, which is based on Lemma 3.11, will help us in answering that question.

\begin{tut}
For any state $\rho:A\to\mathbb{C}$, $a\in A_{sa}$, $(p,q)\in\mathcal{O}\mathbb{R}$, and $C\in\mathcal{V}(A)$, the following are equivalent:
\begin{enumerate}
\item $\rho(\delta^{i}(\chi_{(p,q)}(a))_{C})=1$;
\item $\rho(\delta^{i}(a)_{C})\in[p,q]$ and  $\rho(\delta^{o}(a)_{C})\in[p,q]$. 
\end{enumerate}
\end{tut}

\begin{proof}
Assume that (1) holds. This means that $\mu^{C}_{\rho}(X^{C}_{\delta^{i}(\chi_{(p,q)}(a))_{C}})=1$. By Lemma 3.11(2) and monotonicity of $\mu_{\rho}$,
\begin{equation}
\mu^{C}_{\rho}(\{\lambda\in\Sigma_{C}\mid\lambda(\delta^{i}(a)_{C})\geq p,\ \ \lambda(\delta^{o}(a)_{C})\leq q\})=1,
\end{equation}
from which (2) follows. Conversely, assume (2). Then (152) holds. In the proof of Lemma 4.7, we saw that for any $b\in C_{sa}$, $\mu_{\rho}^{C}(X^{C}_{b})=\mu_{\rho}^{C}(\overline{X^{C}_{b}})$, where $\overline{X^{C}_{b}}$ denotes the closure of $X^{C}_{b}$. We conclude 
\begin{equation}
\mu^{C}_{\rho}(\{\lambda\in\Sigma_{C}\mid\lambda(\delta^{i}(a)_{C})>p,\ \ \lambda(\delta^{o}(a)_{C})<q\})=1.
\end{equation}
Claim (1) now follows from Lemma 3.11(1) and monotonicity of $\mu_{\rho}$
\end{proof}

As far as truth values are concerned, it does not matter if we take $[a\in(p,q)]_{1}$ or $[a\in(p,q)]_{2}$.

\su{Contravariant States in the Covariant Approach}

The goal of this subsection is twofold. We introduce pseudostates in the covariant approach and introduce covariant states that are closely related to the measures of closed open subobjects of $\uS$, used by D\"oring, \cite{doe3, doe2}.

Using the covariant version of daseinisation presented in Subsection 3.2, we introduce a covariant counterpart to the pseudostates of Subsection 4.1. A pseudostate together with a proposition yields a truth value, and we can compare this with the truth value obtained by using the covariant state $\mu_{\psi}$ of the previous subsection. In order to use the daseinisation technique, we restrict $\mathcal{C}(A)$ to its subset $\mathcal{V}(A)$ of von Neumann algebras. Define
\begin{equation}
\mathfrak{w}^{\stat}=[\dirac=1]\ ;
\end{equation}
\begin{equation}
\mathfrak{w}^{\stat}_{C}=\{\lambda\in\Sigma_{C}\mid\lambda(\delta^{i}(\dirac)_{C})=1\ \text{and}\ \lambda(\delta^{o}(\dirac)_{C})=1\}.
\end{equation}
As $\delta^{i}(\dirac)_{C}\leq\delta^{o}(\dirac)_{C}$, this simplifies to
\begin{equation}
\mathfrak{w}^{\stat}_{C}=\{\lambda\in\Sigma_{C}\mid\lambda(\delta^{i}(\dirac)_{C})=1\}.
\end{equation}
Compare this to the contravariant pseudostate (134). The inner daseinisation of a one-dimensional projection is rather simple. If $\dirac\in C$, then $\delta^{i}(\dirac)_{C}=\dirac$, whereas if $\dirac\notin C$, then $\delta^{i}(\dirac)_{C}=0$. Let $C(\dirac)$ be the context generated by $\dirac$, and let $\chi_{\stat}:\Sigma|_{\uparrow C(\dirac)}\to 2$ be the Gelfand transform of $\dirac$, as in (63). It is easy to check that $\mathfrak{w}^{\stat}=\chi_{\stat}^{-1}(1)$.

The open $\mathfrak{w}^{\stat}\in\mathcal{O}\Ss$ defines an internal open $\pseu$ by 
\begin{equation}
\pseu_{C}(\ast)=\coprod_{C'\in\uparrow C}\mathfrak{w}^{\stat}_{C'}. 
\end{equation}
Let $\underline{U}:\underline{1}\to\uSs$ be any covariant proposition. Then the truth value $\nu(\pseu\subseteq\underline{U})_{C}$ at stage $C$ is given by: $C'\in\nu(\pseu\subseteq\underline{U})_{C}$ iff $C'\supseteq C$, and every context $C''\supseteq C'$ such that $\dirac\in C''$ satisfies the condition that if $\lambda''(\dirac)=1$, then $\lambda''\in U_{C''}$. Can we obtain the same truth values using the covariant states of Subsection 4.2? 

The short answer is no. Suppose that $C$ is any maximal context satisfying $\stat\langle\psi|\notin C$. Then $\nu(\pseu\subseteq\underline{U})_{C}=\{C\}$ for every $\underline{U}\in\mathcal{O}\uS_{\ast}$, even for $\underline{U}=\underline{0}$. For every covariant state $\underline{\mu}$, we have $\underline{\mu}(\underline{0})=\underline{0}_{l}$, so $\nu(\underline{\mu}(\underline{0})=\underline{1}_{l})_{C}=\emptyset$. The truth values obtained from a covariant pseudostate differ from truth values obtained by covariant states. Note that the situation is different in the contravariant approach, where the measures on the spectral presheaf (Definition 4.2), generalize contravariant pseudostates. The next two propositions show how close we can get to obtaining the truth values of covariant pseudostates by means of covariant states.

\begin{poe}
The map $\mu^{0}_{\psi}:\mathcal{O}\Sigma_{\ast}\to\textit{OP}(\mathcal{V}(A),[0,1])$, defined as
\begin{equation} 
\mu^{0}_{\psi}(U)(C)= \left\{ 
\begin{array}{rl} 
1 & \forall_{C'\in(\uparrow C)}\ \text{If}\ \stat\langle\psi\vert\in C'\ \text{then}\ X^{C'}_{\stat\langle\psi\vert}\subseteq U_{C'}\\ 
0 & \text{otherwise}
\end{array} \right.,
\end{equation}
satisfies conditions (1), (3) and (4) in Proposition 4.4 for a covariant state.
\end{poe}

\begin{proof}
Most of the proof is easy and is left to the reader. For the proof that $\mu^{0}_{\psi}$ satisfies (3) of Proposition 4.4, we consider the case that $\mu^{0}_{\psi}(U)(C)=\mu^{0}_{\psi}(V)(C)=0$ and that there is a $E\in(\uparrow C)$ such that $\stat\langle\psi\vert\in E$. By definition of $\mu^{0}_{\psi}$, there exist contexts $C',C''\in(\uparrow C)$ such that $\stat\langle\psi\vert\in C',C''$, and $X^{C'}_{\stat\langle\psi\vert}\nsubseteq U_{C'}$, $X^{C''}_{\stat\langle\psi\vert}\nsubseteq V_{C''}$. Let $D\in(\uparrow C)$ be the commutative C*-algebra generated by $C$ and $\stat\langle\psi\vert$. Then $D$ is the smallest context satisfying $D\in(\uparrow C)$ and $\stat\langle\psi\vert\in D$. We find $X^{D}_{\stat\langle\psi\vert}\nsubseteq U_{D}\cup V_{D}$. This is because $D\subseteq C',C''$ and $X^{D}_{\stat\langle\psi\vert}\subseteq U_{D}$ implies $X^{C'}_{\stat\langle\psi\vert}\subseteq U_{C'}$, whilst $X^{D}_{\stat\langle\psi\vert}\subseteq V_{D}$ implies $X^{C''}_{\stat\langle\psi\vert}\subseteq V_{C''}$. We conclude that $\mu^{0}_{\psi}(U\vee V)(C)=0$. 

For the proof that $\mu^{0}_{\psi}$ satisfies (4) of Proposition 4.4, let $\{U_{i}\}_{i\in I}\subset\mathcal{O}\Sigma_{\ast}$ be directed, and assume that $\mu^{0}_{\psi}(\cup_{i}U_{i})(C)=1$. Also assume that there is a context $C'\in(\uparrow C)$ such that $\stat\langle\psi\vert\in C'$. Then let, as before $D$, be the smallest such context. By assumption, $X^{D}_{\stat\langle\psi\vert}\subseteq\bigcup_{i\in I}U_{i, D}$. For every $\lambda\in X^{D}_{\stat\langle\psi\vert}$, take an open neighborhood $V_{\lambda}$ small enough, such that $V_{\lambda}\subseteq U_{i, D}$ for $i\in I$ sufficiently large. The open and closed set $X^{D}_{\stat\langle\psi\vert}$ is the union of these $V_{\lambda}$, and by compactness of $\Sigma_{C}$ we only need a finite number of $V_{\lambda}$. By directedness, we find that  $X^{D}_{\stat\langle\psi\vert}\subseteq U_{i,D}$ for $i\in I$ sufficiently large. This demonstrates that $\sup\{\mu^{0}_{\psi}(U_{i})(C)\mid i\in I\}=1$.
\end{proof}

\begin{poe}
For every $C\in\mathcal{V}(A)$ and every $U\in\mathcal{O}\Sigma_{\ast}$,
\begin{equation}
\nu(\underline{\mu}^{0}_{\psi}(\underline{U})=\underline{1}_{l})_{C}=\nu(\underline{\mathfrak{w}}^{\stat}\subseteq\underline{U})_{C}.
\end{equation}
\end{poe}

The truth values of Proposition 4.10 coincide, but $\mu^{0}_{\psi}$ is not quite a covariant state, because it does not satisfy $\mu^{0}_{\psi}(\emptyset)(C)=0$ in Proposition 4.4(2).

We close this section with a discussion of a covariant version of measures on the spectral presheaf, as in Definition 4.2. Define
\begin{equation} 
\mathcal{O}_{cl}\Sigma_{\ast}=\{U\in\mathcal{O}\Sigma_{\ast}\mid\forall C\in\mathcal{V}(A),\ U_{C}\ \text{is closed in}\  \Sigma_{\ast}\}.
\end{equation}

\begin{dork}
A measure on $\mathcal{O}_{cl}\Sigma_{\ast}$, is a map $\mu:\mathcal{O}_{cl}\Sigma_{\ast}\to\text{OP}(\mathcal{V}(A),[0,1])$, such that for every $C\in\mathcal{V}(A)$ and $\forall U_{1}, U_{2}\in\mathcal{O}_{cl}\Sigma_{\ast}$,
\begin{itemize}
\item $\mu(\Sigma_{\ast})(C)=1$;
\item $\mu(U_{1})(C)+\mu(U_{2})(C)=\mu(U_{1}\wedge U_{2})(C)+\mu(U_{1}\vee U_{2})(C)$.
\item For any fixed $C\in\mathcal{V}(A)$, the function $\mu^{C}:=\mu(-)(C):\mathcal{O}_{cl}\Sigma_{\ast}\to[0,1]$, $U\mapsto\mu(U)(C)$ depends only on $U_{C}$. We write $\mu^{C}(U)=\mu^{C}(U_{C})$ with slight abuse of notation.
\end{itemize}
\end{dork}

Note that we use order preserving functions instead of order reversing functions (as in Definition 4.2). Internal to the respective topoi there is no difference. The set $\text{OR}(\mathcal{V}(A),[0,1])$ is the external description of the lower reals $\underline{[0,1]}_{l}$ in $[\mathcal{V}(A)^{\op},\Set]$, while $\text{OP}(\mathcal{V}(A),[0,1])$ is the external description of  $\underline{[0,1]}_{l}$ in $[\mathcal{V}(A),\Set]$.

Let $\mathcal{M}(\mathcal{O}_{cl}\Sigma_{\ast})$ denote the set of measures as in Definition 4.12, and let $\mathcal{S}(\mathcal{O}\Sigma_{\ast})$ denote the set of covariant states, which are functions satisfying the conditions of Proposition 4.4. If $\mu:\mathcal{O}\Sigma_{\ast}\to\text{OP}(\mathcal{V}(A),[0,1])$ is a covariant state, then the restriction of $\mu$ to $\mathcal{O}_{cl}\Sigma_{\ast}$ clearly satisfies the first two conditions for measures on $\mathcal{O}_{cl}\Sigma_{\ast}$. If $\mu=\mu_{\rho}$, comes from a quasi-state as in Lemma 4.6, then $\mu$ also satisfies the third condition of Definition 4.12. By the generalized Riesz-Markov theorem every covariant state comes from a quasi-state on $A$, so we can conclude that the function
\begin{equation}
r:\mathcal{S}(\mathcal{O}\Sigma_{\ast})\to\mathcal{M}(\mathcal{O}_{cl}\Sigma_{\ast}),\ \ \ r(\mu)=\mu|_{\mathcal{O}_{cl}(\Sigma_{\ast})}
\end{equation}
is well-defined.

\begin{tut}
Let $A$ be a von Neumann algebra without type $I_{2}$ summand. Then the map $r$ of (161) is a bijection.
\end{tut}

\begin{proof}
For a von Neumann algebra without type $I_{2}$ summand, any quasi-state $\rho:A\to\mathbb{C}$ is a state \cite{buwi}. By the Riesz-Markov theorem there is a bijective correspondence between states $\rho$ on $A$, and covariant states $\mu\in\mathcal{S}(\mathcal{O}\Sigma_{\ast})$.

Next, we want to show that the measures in $\mathcal{M}(\mathcal{O}_{cl}\Sigma_{\ast})$ also correspond bijectively to states $\rho$ on $A$. The proof of this claim is an adaptation of a proof of a similar result in the contravariant setting, given in \cite{doe3}. The underlying idea is the following:
\begin{enumerate}
\item Define for every $\mu\in\mathcal{M}(\mathcal{O}_{cl}\Sigma_{\ast})$ the function $m:\mathcal{P}(A)\to[0,1]$ by 

\noindent $m(p)=\mu(U)(C)$, where $U\in\mathcal{O}_{cl}\Sigma_{\ast}$ satisfies $U_{C}=X^{C}_{p}$.
\item Show that $m$ is well-defined (that it does not depend on the choice of $U$ and $C$).
\item Show that $m(1)=1$, and for any pair of orthogonal projections $p,q\in\mathcal{P}(A)$ we have $m(p)+m(q)=m(p+q)$. This makes $m$ into a finitely additive probability measure on the projections of $A$.
\item For a von Neumann algebra without type $I_{2}$ summand, $m$ has a unique extension to a state on $A$ by Gleason's theorem \cite{mae}.
\item To each state $\rho$ we can assign a measure $\mu_{\rho}$ by $\mu_{\rho}(U)(C)=\rho(p)$, where $U_{C}=X^{C}_{p}$. The assignment $\rho\mapsto\mu_{\rho}$ is easily seen to be an inverse of the assignment $\mu\mapsto\rho_{\mu}$ obtained by the previous 4 steps.
\end{enumerate}
For the proof of step (2), first assume that $U\in\mathcal{O}_{cl}\Sigma_{\ast}$ has the property that $U_{C}=X^{C}_{p}$, and that there is a context $C'\subset C$, such that $p\in C'$ and $U_{C'}=X^{C'}_{p}$. We will show that $\mu(U)(C)=\mu(U)(C')$. Choose any $V\in\mathcal{O}_{cl}\Sigma_{\ast}$ such that $V_{C}=X^{C}_{1-p}$ and $V_{C'}=X^{C'}_{1-p}$. From Definition 4.12(1) and (2) we deduce
\begin{equation}
\mu(U)(C)+\mu(V)(C)=1=\mu(U)(C')+\mu(V)(C').
\end{equation}
As both $\mu(U)$ and $\mu(V)$ are order preserving, we conclude $\mu(U)(C)=\mu(U)(C')$.

Next consider the case that $U,V\in\mathcal{O}_{cl}\Sigma_{\ast}$ satisfy $U_{C}=X^{C}_{p}=V_{C}$. We want to show that $\mu(U)(C)=\mu(V)(C)$. Using Definition 4.12(2) and (3) we find
\begin{align}
\mu(U)(C)+\mu(V)(C) &= \mu(U\vee V)(C)+\mu(U\wedge V)(C)\\
&= \mu^{C}(U_{C}\cup V_{C})+\mu^{C}(U_{C}\cap V_{C})\\
&= \mu^{C}(U_{C})+\mu^{C}(U_{C})\\
&= \mu(U)(C)+\mu(U)(C),
\end{align}
from which $\mu(U)(C)=\mu(V)(C)$ follows. Now we can prove (2). Suppose that  $U,V\in\mathcal{O}_{cl}\Sigma_{\ast}$, satisfy $U_{C}=X^{C}_{p}$ and $V_{C'}=X^{C'}_{p}$. Consider the context $C\cap C'$, which contains $p$. By the previous two steps we find
\begin{align}
\mu(U)(C) &= \mu(\delta^{i}(p))(C)\\ 
&= \mu(\delta^{i}(p))(C\cap C')\\
&= \mu(\delta^{i}(p))(C') = \mu(V)(C'),
\end{align}
which proves (2). For the proof of step (3), we remark that as $p$ and $q$ commute, there exists a context $C$ containing both. For such a context $C$ we find
\begin{equation}
m(p+q)=\mu(\delta^{i}(p+q))(C)=\mu^{C}(\delta^{i}(p+q)_{C})=\mu^{C}(\delta^{i}(p)_{C}\cup\delta^{i}(q)_{C}),
\end{equation}
while $\mu(\delta^{i}(p)_{C}\cap\delta^{i}(q)_{C})=\mu(\emptyset)=0$. Note that the orthogonality of $p$ and $q$ where used for these identities. By Definition 4.12(2), $m(p)+m(q)=m(p+q)$, proving step (3). We leave proving step (5) to the reader.

Now we can prove Theorem 4.13. Start with a state $\rho$ on $A$. This induces a covariant state $\mu_{\rho}$ as in Lemma 4.6. If for $U\in\mathcal{O}_{cl}\Sigma_{\ast}$ we have $U_{C}=X^{C}_{p}$, then for restriction of this covariant state $r(\mu_{\rho})$ we find $r(\mu_{\rho})(U)(C)=\rho(p)$. The associated map $m: \mathcal{P}(A)\to[0,1]$ is then given by $m(p)=\rho(p)$ and the unique state on $A$ corresponding to $r(\mu_{\rho})$ is $\rho$. In short, if $\mathcal{S}(A)$ denotes the set of states on $A$, then we have shown that the following square, where the vertical arrows are bijections, commutes.
$$\xymatrix{
\mathcal{S}(\mathcal{O}\Sigma_{\ast}) \ar[rr]^{r} & & \mathcal{M}(\mathcal{O}_{cl}\Sigma_{\ast}) \ar[d]^{\cong}\\
\mathcal{S}(A) \ar[u]^{\cong}  \ar[rr]_{id} & & \mathcal{S}(A)}$$
This implies that $r$ is a bijection.
\end{proof}

\newpage

\se{Summary}

What have we learned? Regarding the formalism of the covariant approach, we have learned quite a lot. Corollary 2.18 gives an explicit description of the spectrum of $\underline{A}$ in $[\mathcal{C}(A),\Set]$ by means of a bundle of topological spaces $\pi:\Sigma_{\ast}\to\mathcal{C}(A)$. As a consequence, the external spectrum is a spatial locale (Corollary 2.20).

Restricting attention from general (unital) C*-algebras to von Neumann algebras we introduced a daseinisation map (Definition 3.3). This map was defined to mirror the contravariant daseinisation map $\breve{\delta}(a):\uS\to\unR^{\leftrightarrow}$ as closely as possible. Lemma 3.6 and the discussion following it show that the only difference between this new daseinisation arrow and the original covariant daseinisation arrow \cite{hls} is that the new version uses the spectral order. By continuity of the daseinisation arrows $\delta(a)$ (Proposition 3.4), we define covariant elementary propositions $[a\in(p,q)]_{1}$ in the same way as for the original covariant daseinisation arrow (Definition 3.5). We also introduce covariant elementary propositions $[a\in(p,q)]_{2}$ in a different way (93), which mirrors the contravariant elementary propositions more closely. Although the two versions are different, Lemma 3.11 shows that they are closely connected.

In Subsection 4.2 we study the truth values obtained by pairing elementary propositions with covariant states. Lemma 4.7 treats the $[a\in(p,q)]_{2}$ version of elementary propositions and Lemma 4.8 discusses the  $[a\in(p,q)]_{1}$ version of elementary propositions. Theorem 4.9 showed that the truth values obtained from $[a\in(p,q)]_{1}$ and $[a\in(p,q)]_{2}$ are exactly the same.

In Subsections 1.2 and 3.5 we consider a possible physical interpretation for the covariant approach. In essence, this interpretation sees the Kripke-Joyal semantics of the topos $[\mathcal{V}(A),\Set]$ as a (physical) Kripke model. We think of a context $C\in\mathcal{V}(A)$ as a stage of knowledge about the system, and think of the elementary proposition $[a\in(p,q)]$ as representing a property of the system (just as in the contravariant approach). Although we think of $[a\in(p,q)]$ in the same way as in the contravariant approach, the covariant notion of truth differs greatly the contravariant one. In the covariant approach, truth of the property  $[a\in(p,q)]$ at a context $C$ relative to some state is interpreted as follows: From the knowledge that $C$ provides, we can verify that the system has the property $[a\in(p,q)]$ relative to that state. If, in the covariant approach, the pairing of a property and a state yields ``true" at \textit{some} context, then the pairing of the same property and state in the contravariant approach yields ``true" at \textit{every} context (in that the property holds with respect to the given state). Note that truth in the covariant approach is concerned with our \textit{knowledge} of the system, whereas truth in the contravariant approach is concerned with the question to what extent a system \textit{has} a property. The state-proposition pairings of Lemma 4.7 and Lemma 4.8 seems consistent with this physical interpretation. This consistency relies on the use of the spectral order in the definition of the daseinisation map (see discussion after Lemma 4.8).

In this paper, the contravariant approach has mainly been investigated by taking a topos internal perspective on the various constructions of the approach. In Subsection 2.1, the closed open subobjects of the spectral presheaf were shown to give a locale in the topos $[\mathcal{V}(A)^{\op},\Set]$ (Proposition 2.3). Instead of only looking at closed open subobjects, we studied the locale of open subobjects (a locale more closely connected to the spectrum of the covariant approach), which was described by a bundle of topological spaces $\pi:\Sigma^{\ast}\to\mathcal{V}(A)$ (Theorem 2.2). This internal locale has been shown to be compact (Corollary 2.7) but generally not regular (Corollary 2.10). Therefore, the locale in question cannot arise as the spectrum of some internal commutative C*-algebra. In Subsection 3.3, an attempt is made to view the contravariant daseinisation map internally as a continuous map. However, this investigation does not yield a contravariant counterpart to the covariant elementary proposition (which is defined by both inner and outer daseinisation). Applying (internal) constructions from the covariant approach to the contravariant approach turns out to be quite difficult, giving rise to the following question: What role might the internal language of the topos $[\mathcal{V}(A),\Set]$ play in the contravariant approach? For the covariant approach the use of internal language has been crucial so far, but for the contravariant approach, where the semantics is dictated by coarse-graining, the role of the internal language is not a priori clear.

Mathematically, the two approaches are closely related, and also at the level of interpretation there may be important connections. Entering the realm of speculation, it would seem desirable to have a formalism which incoorporates both approaches. In such a formalism, one could both coarse-grain and refine, moving more freely between contexts. Furthermore, one could investigate both to what extent a system has a property and in which way we can verify that it has that property. However, at the time of writing, it seems unclear how such a formalism can be obtained.

\se{Acknowledgements}

The author would like to thank Bas Spitters for sharing his ideas (such as the suggestion of using the space $\Sigma^{\ast}$), and Klaas Landsman for his comments, which greatly improved this paper. Also, thanks to the editor and the anonymous referees, for suggestions that helped improve this paper.

\end{document}